\documentclass[11pt]{article}

\usepackage[margin=1in]{geometry}

\usepackage[T1]{fontenc}
\usepackage[utf8]{inputenc}
\usepackage{amsmath,amssymb,amsfonts,amsthm,mathtools}
\usepackage{authblk}
\usepackage{bm}
\usepackage{dsfont}
\usepackage{booktabs}
\usepackage{enumitem}
\usepackage{graphicx}
\usepackage{subcaption}
\usepackage[numbers,sort&compress,square]{natbib}
\usepackage[colorlinks=true,citecolor=blue,linkcolor=blue,urlcolor=blue,breaklinks=true]{hyperref}
\usepackage{tikz}
\usetikzlibrary{arrows.meta}
\usetikzlibrary{calc}
\theoremstyle{plain}
\newtheorem{theorem}{Theorem}

\usepackage{newtxtext,newtxmath}
\usepackage{xcolor}

\theoremstyle{definition}

\definecolor{bg}{RGB}{244,244,244}
\definecolor{grid}{RGB}{198,202,202}

\newcommand{\ZN}{\mathbb Z_N}

\newcommand{\U}{\mathrm U}

\newcommand{\supp}{\operatorname{supp}}

\newcommand{\ket}[1]{\left|#1\right\rangle}
\newcommand{\dd}{\mathrm d}

\newcommand{\loc}{\mathrm{loc}}

\graphicspath{{./figures/}}

\newcommand{\ZZ}{\mathbb Z}

\newcommand{\OverlapWord}[2]{%
	\mathord{\vcenter{\hbox{%
				\begin{tikzpicture}[x=1cm,y=0.36cm]
					\draw[black,line width=0.9pt,line cap=round] (0,0) -- (3,0);
					
					\foreach \x in {0,1,2,3}
					\fill[black] (\x,0) ellipse[x radius=0.045,y radius=0.16];
					
					\node[inner sep=0.5pt] at (0.00,0.72) {$\scriptstyle -#1$};
					\node[inner sep=0.5pt] at (1.00,0.72) {$\scriptstyle #2$};
					\node[inner sep=0.5pt] at (2.00,0.72) {$\scriptstyle #1$};
					\node[inner sep=0.5pt] at (3.00,0.72) {$\scriptstyle -#2$};
					
					\node[inner sep=0.5pt] at (0.00,-0.72) {$\scriptstyle x$};
					\node[inner sep=0.5pt] at (1.00,-0.72) {$\scriptstyle y'$};
					\node[inner sep=0.5pt] at (2.00,-0.72) {$\scriptstyle x'$};
					\node[inner sep=0.5pt] at (3.00,-0.72) {$\scriptstyle y$};
				\end{tikzpicture}%
	}}}%
}

\newcommand{\LinkedEllipses}[2]{%
	\ensuremath{\mathord{\vcenter{\hbox{%
					\begin{tikzpicture}[
						x=.55em,y=.55em,
						line cap=round,
						line join=round,
						inner sep=0pt,
						outer sep=0pt
						]
						\def\lw{.09em}%
						\def\mw{.20em}%
						\def\vax{.75}%
						\def\vay{1.85}%
						\def\hcx{1.75}%
						\def\hcy{.05}%
						\def\hax{2.20}%
						\def\hay{.55}%
						
						\draw[line width=\lw]
						(0,0) ellipse [x radius=\vax,y radius=\vay];
						
						\draw[line width=\lw]
						(\hcx,\hcy) ellipse [x radius=\hax,y radius=\hay];
						
						\draw[draw=white,line width=\mw,line cap=butt]
						({\vax*cos(0)},{\vay*sin(0)})
						arc[start angle=0,end angle=35,x radius=\vax,y radius=\vay];
						
						\draw[line width=\lw,line cap=butt]
						({\vax*cos(0)},{\vay*sin(0)})
						arc[start angle=0,end angle=35,x radius=\vax,y radius=\vay];
						
						\draw[draw=white,line width=\mw,line cap=butt]
						({\hcx+\hax*cos(226)},{\hcy+\hay*sin(226)})
						arc[start angle=226,end angle=260,x radius=\hax,y radius=\hay];
						
						\draw[line width=\lw,line cap=butt]
						({\hcx+\hax*cos(226)},{\hcy+\hay*sin(226)})
						arc[start angle=226,end angle=260,x radius=\hax,y radius=\hay];
						
						\node at (1,-2.34) {$#1$};
						\node at (4.5,-1.30) {$#2$};
					\end{tikzpicture}%
	}}}}%
}

\usetikzlibrary{arrows.meta,decorations.markings}

\definecolor{braidred}{RGB}{214,0,0}
\definecolor{braidblue}{RGB}{0,130,235}
\pgfdeclarelayer{foreground}
\pgfsetlayers{main,foreground}

\tikzset{
  lattice/.style={black,line width=0.75pt},
  time line/.style={black,line width=0.55pt},
  sheet/.style={line width=3.1pt,line cap=round,line join=round},
  arrowed/.style={
    postaction={
      decorate,
      decoration={
        markings,
        mark=at position #1 with {\arrow{Stealth[length=8.2pt,width=7.2pt]}}
      }
    }
  },
  arrowed/.default=0.55,
  defect/.pic={
    \draw[pic actions,line width=4.3pt,line cap=round]
      (-0.13,-0.13) -- (0.13,0.13)
      (-0.13,0.13) -- (0.13,-0.13);
  }
}

\title{Bockstein braiding statistics}
\author[1]{Po-Shen Hsin}
\author[2,$\dagger$]{Yu-An Chen}
\affil[1]{Department of Mathematics, King's College London, Strand, London WC2R 2LS, UK}
\affil[2]{International Center for Quantum Materials, School of Physics, Peking University, Beijing 100871, China}

\date{\today}

\newif\ifincludeappendices
\includeappendicestrue

\begin{document}
\maketitle

\begingroup
\renewcommand{\thefootnote}{$\dagger$}
\footnotetext{Contact author: \href{mailto:yuanchen@pku.edu.cn}{yuanchen@pku.edu.cn}}
\endgroup

\begin{abstract}
Braiding phenomena, from the charge-flux Aharonov-Bohm effect to anyonic statistics in fractional quantum Hall systems, are paradigmatic manifestations of topology in quantum physics.
Ordinary mutual braiding between $p$- and $q$-dimensional excitations occurs in $d=p+q+2$ spatial dimensions.
In this work, we introduce a universal construction of mutual statistics in the adjacent dimension, $d=p+q+1$, applicable to excitations obeying $\mathbb Z_N$ fusion for arbitrary $N$ and all excitation dimensions $p$ and $q$. The corresponding invariant is the Berry phase accumulated in a simple $4N$-step microscopic unitary process built from local excitation operators on lattices. This process measures the linking of one excitation with the $N$-fold fusion junction of the other, encompassing particle-particle statistics in one dimension, particle-loop statistics in two dimensions, and loop-loop or particle-membrane statistics in three dimensions. We establish the quantization and bilinearity of the invariant and show that its field-theory response is governed by the Bockstein homomorphism, motivating the name Bockstein braiding statistics.
Interpreting the excitation operators as open symmetry operators turns the same invariant into a direct microscopic diagnostic of mixed anomalies between symmetries. We demonstrate this diagnostic in a $(1{+}1)$D spin chain, where the nontrivial Bockstein braiding phase proves the mixed anomaly between the spin-flip symmetry $\prod_i X_i$ and the nearest-neighbor controlled-$Z$ symmetry $\prod_i \mathrm{CZ}_{i,i+1}$. We construct explicit $(2{+}1)$D and $(3{+}1)$D lattice analogs, yielding new anomalous symmetry pairs, and apply the framework to strongly coupled $(3{+}1)$D continuum gauge theories.
Nontrivial Bockstein braiding rules out a fully symmetric gapped phase, obstructs simultaneous condensation of the two excitations, and implies fractionalization of higher-form symmetries.
\end{abstract}

\section{Introduction}

\begin{figure}[t]
    \centering

    \begin{subfigure}[t]{0.4\textwidth}
        \centering
        \caption[Ordinary braiding statistics]{
            Ordinary braiding statistics\\
            $\left[U^{B},U^{A}\right]$
        }
        \label{fig:ordinary-braiding}
        \includegraphics[
            width=\linewidth
        ]{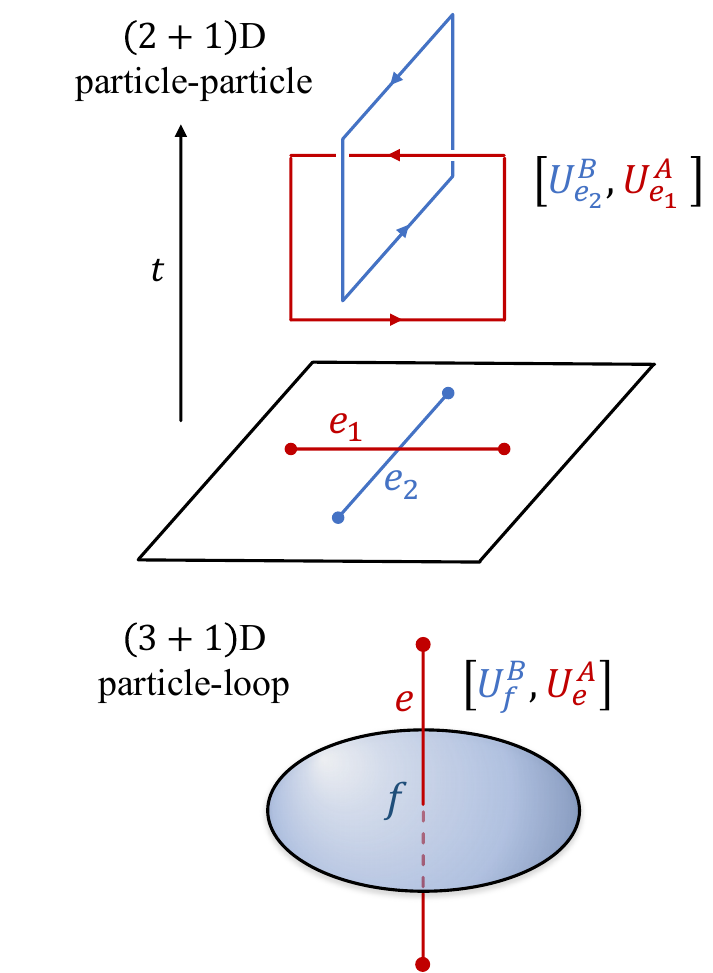}
    \end{subfigure}
    \hfill
    \begin{subfigure}[t]{0.58\textwidth}
        \centering
        \caption[Bockstein braiding statistics]{
            Bockstein braiding statistics\\
            $\bigl(Y^{-1}X^{-1}\bigr)^N(YX)^N$
        }
        \label{fig:bockstein-braiding}
        \includegraphics[
            width=\linewidth
        ]{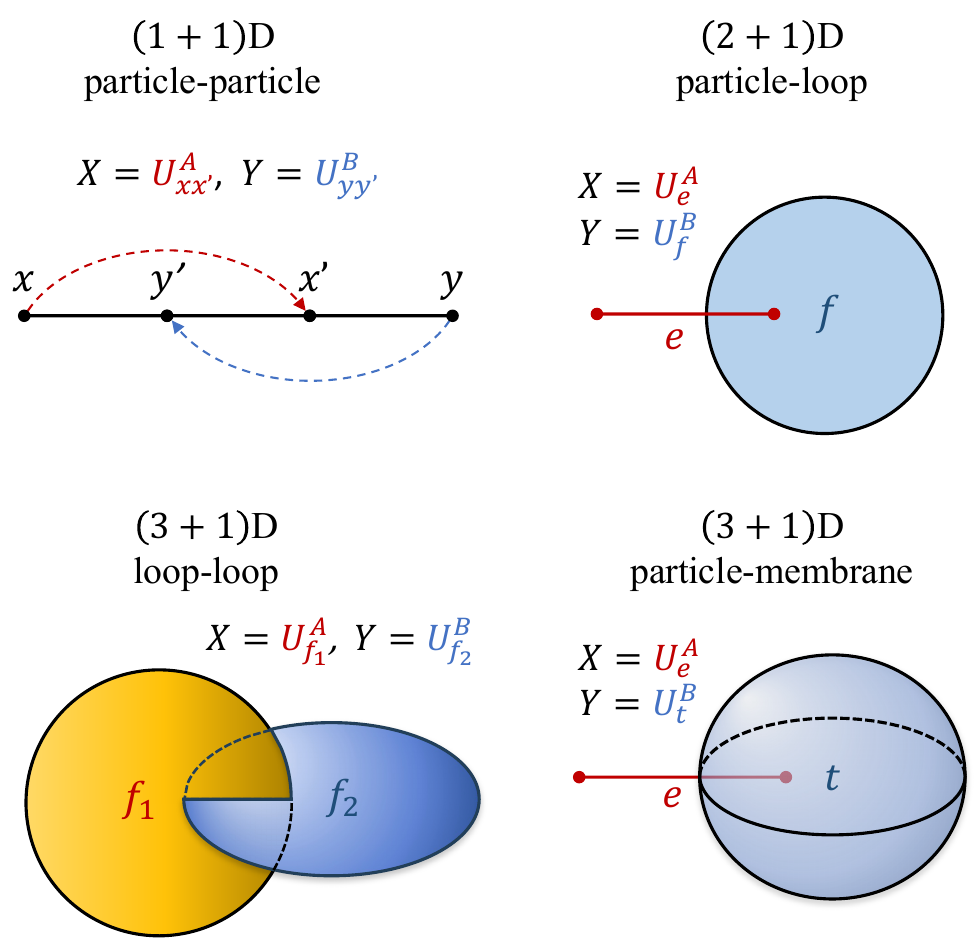}
    \end{subfigure}
    \caption{
        Comparison of the two mutual statistics. Operator \(U_\Delta\) is a creation or hopping operator supported on
        \(\Delta\), creating an excitation on \(\partial\Delta\).
        \textbf{(a)} Ordinary braiding is represented by the commutator
        $\left[U^{B},U^{A}\right]=(U^B)^{-1}(U^A)^{-1}U^B U^A$.  The examples
        depict particle-particle braiding in $(2{+}1)$ dimensions and
        particle-loop braiding in $(3{+}1)$ dimensions.
        \textbf{(b)} Bockstein braiding is represented by
        $W_N(X,Y)=\bigl(Y^{-1}X^{-1}\bigr)^N(YX)^N$ for excitation operators
        whose supports have a staggered one-dimensional overlap.  The examples
        depict particle-particle statistics in $(1{+}1)$ dimensions,
        particle-loop statistics in $(2{+}1)$ dimensions, and loop-loop and
        particle-membrane statistics in $(3{+}1)$ dimensions.
    }
    \label{fig:ordinary-and-bockstein-braiding}
\end{figure}

Braiding phases are among the most direct signatures of topology in quantum systems. In the Aharonov--Bohm effect, a charged particle acquires a phase by encircling a magnetic flux even when the electromagnetic field vanishes along its trajectory~\cite{AharonovBohm1959}. In a fractional quantum Hall fluid, quasiparticles acquire anyonic phases when they are exchanged or braided~\cite{ArovasSchriefferWilczek1984,Wilczek1982}.
Particle--loop braiding also occurs in three-dimensional superconductors: a Bogoliubov quasiparticle whose closed trajectory links an Abrikosov vortex loop acquires a $\pi$ phase from the loop's quantized magnetic flux~\cite{Hansson_2004}.
These phases are universal because they depend only on the topology of the spacetime history. They also provide the basic operational data behind many proposed schemes for topological quantum computation~\cite{Preskill1999TopologicalQuantumComputation, Kitaev2003Fault, freedman2003topological, Nayak2008NonAbelian}.

Their common origin is linking. In $d$ spatial dimensions,
ordinary mutual braiding between $p$- and $q$-dimensional excitations requires $d=p+q+2$,
so that their spacetime worldvolumes can link while remaining disjoint. If
$X$ and $Y$ create the two excitations and their supports meet at one interior
point (Fig.~\ref{fig:ordinary-and-bockstein-braiding}\subref{fig:ordinary-braiding}), the process is the commutator
\begin{equation}
    [Y,X]:=Y^{-1}X^{-1}YX\in\U(1).
    \label{eq:comm-def}
\end{equation}
For excitations with $\ZN$ fusion, this phase is quantized as
$\exp(2\pi i k/N)$ and is encoded by the ordinary inflow response
$(2\pi i k/N)\int A\cup B$
\cite{Gaiotto:2014kfa,Kapustin2014Anomalous,
kapustin2014anomaliesdiscretesymmetriesvarious}.

However, changing the spatial dimension destroys this disjoint-linking geometry. We therefore ask whether a similarly local and robust statistical process exists in different dimensions.
A statistical process is a closed sequence of excitation operators whose Berry phase is unchanged by local microscopic
redefinitions~\cite{xue2025statistics,kobayashi2024generalized}.
General constructions can be long and geometrically opaque: for example, the
Pontryagin $\mathbb Z_3$ membrane statistic uses 56 operators
\cite{feng2025anyonic}, with further complexity for higher-dimensional
excitations~\cite{feng2026paulistabilizerformalismtopological}.

In this work, we find a uniform construction for
\begin{equation}
    d=p+q+1.
    \label{eq:bockstein-dimension-matching}
\end{equation}
Here the two creation-operator supports generically overlap along a line. For
excitations with $\ZN$ fusion, arrange the supports in the staggered geometry
of Fig.~\ref{fig:ordinary-and-bockstein-braiding}\subref{fig:bockstein-braiding}
and define
\begin{equation}
    W_N(X,Y):=(Y^{-1}X^{-1})^N(YX)^N.
    \label{eq:WN-def}
\end{equation}
The phase measures the linking of an $A$-type excitation with a selected $N$-fold fusion junction of the $B$-type excitation.
The previously studied $N{=}2$ particle--membrane process~\cite{kobayashi2024generalized} in $(3{+}1)$D is a
special case of this construction.\footnote{Appendix~\ref{app:eq41} establishes the equivalence explicitly.}
The new results are the dimension-independent formula for arbitrary fusion order $N$ and all $d=p+q+1$, its Bockstein and mixed-anomaly interpretation, the quantization and bilinearity of the invariant, and explicit lattice realizations.

In field theory, the $N$-fold fusion junction is encoded by the Bockstein class $\beta_N B$.~\footnote{Here, $\beta_N$ is the connecting homomorphism
associated with the short exact sequence
$0\to\ZN\xrightarrow{\times N}\mathbb Z_{N^2}
\xrightarrow{\mathrm{mod}\,N}\ZN\to0$.
The background fields $A$ and $B$ are $\ZN$-valued cocycles, and
$\beta_N B$ is Poincar\'e dual to the selected junction at which $N$
$B$-type excitations fuse to the vacuum~\cite{Bockstein1942,hatcher2002algebraic}.}
The associated anomaly-inflow action is
\begin{equation}
    S_{\mathrm{inflow}}[A,B]
    =
    \frac{2\pi i k}{N}
    \int_{\mathcal M_{d+2}}
    A_{d-p}\cup\beta_N B_{d-q},
    \qquad k\in\ZN,
    \label{eq:bockstein-response-intro}
\end{equation}
where $\partial \mathcal M_{d+2} = \mathcal M_{d+1}$ is the physical spacetime.
We call the resulting invariant \emph{Bockstein braiding statistics}. It includes
particle--particle statistics in one dimension, particle--loop statistics in
two dimensions, and loop--loop or particle--membrane statistics in three
dimensions.
In three dimensions, the loop--loop realization defines an intrinsic loop
statistic distinct from three-loop braiding
\cite{Wang2014braiding,Else2017Cheshire,Levin2015loopbraiding,
Jiang2014Generalized,Bi2014anyonloopbraiding}.

Anomalies of lattice symmetries have been investigated from several
perspectives
\cite{Else2014Classifying,Else2020LSM,Kapustin2025Anomalous,kapustin2025higher,
kapustin2025higher2,feng2025higherformanomalieslattices,
Wilbur2026DisentanglingAnomaly}.
Closing an excitation operator turns it into a generalized-symmetry generator. The same junction-linking phase therefore gives a direct
microscopic diagnosis of the mixed anomaly $A\cup\beta_NB$. We demonstrate
this in a spin chain and in higher-dimensional lattice models, and connect it
to Abelian and non-Abelian gauge theories. A nontrivial phase obstructs
simultaneous condensation, excludes a fully symmetric gapped phase, and
enforces higher-form symmetry fractionalization.


\section{The Bockstein braiding process}
\label{sec:overlap-process}

To make the response in Eq.~\eqref{eq:bockstein-response-intro}
operational, we show that $W_N$ is a closed microscopic statistical process
and identify the fusion junction that it detects. In $d=p+q+1$, the supports
of $X$ and $Y$ overlap along a line, so their ordinary commutator need not be
a robust phase. Instead,
\begin{equation}
    W_N(X,Y)=\bigl((XY)^N\bigr)^{-1}(YX)^N~,
    \label{eq:WN-two-histories}
\end{equation}
compares two alternating histories that close by the $\ZN$ fusion rule.

Choose the orientation of $Y$ so that it raises a lifted $B$-type fusion
label $j$ by one. Along $(YX)^N$, $X$ acts at
$j=0,1,\ldots,N-1$, whereas along $(XY)^N$ it acts at
$j=1,2,\ldots,N$. Locality pairs the contributions at
$1\leq j\leq N-1$; the surviving comparison between $j=0$ and $j=N$
measures the linking of the $A$ excitation with the selected $N$-fold
$B$-fusion junction. This is the microscopic representative of
$\beta_NB$.
The cancellation makes the phase independent of local symmetric
decorations of the open operators. Such decorations only redefine local
configuration states and conjugate the movement operators, leaving a closed
statistical process unchanged \cite{FHH21,kobayashi2024generalized}.
An explicit local cancellation, including its higher-dimensional extension,
is given in Appendix~\ref{app:local-cancellation}.

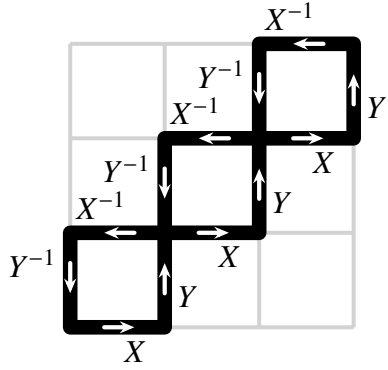
\begin{figure}
    \centering
    \begin{tikzpicture}[
    scale=1.25,
    edgelabel/.style={font=\large, inner sep=1pt},
    gridline/.style={gray!35, line width=1.5pt},
    squareedge/.style={
        black,
        line width=5.5pt,
        line join=round,
        line cap=round
    },
    edgearrow/.style={
        white,
        line width=1.6pt,
        line cap=round,
        -{Stealth[length=5.5pt,width=4.5pt]}
    }
]

\draw[gridline] (0,0) grid (3,3);

\newcommand{\orientedSquare}[2]{%
    \draw[squareedge]
        (#1,#2) --
        (#1+1,#2) --
        (#1+1,#2+1) --
        (#1,#2+1) --
        cycle;

    \draw[edgearrow] (#1+0.36,#2) -- (#1+0.68,#2);        
    \draw[edgearrow] (#1+0.68,#2+1) -- (#1+0.36,#2+1);    
    \draw[edgearrow] (#1,#2+0.68) -- (#1,#2+0.36);        
    \draw[edgearrow] (#1+1,#2+0.36) -- (#1+1,#2+0.68);    

    \node[edgelabel,below=4pt] at (#1+0.68,#2)     {$X$};
    \node[edgelabel,above=4pt] at (#1+0.32,#2+1)   {$X^{-1}$};
    \node[edgelabel,left=4pt]  at (#1,#2+0.68)     {$Y^{-1}$};
    \node[edgelabel,right=4pt] at (#1+1,#2+0.32)   {$Y$};
}

\orientedSquare{0}{0}
\orientedSquare{1}{1}
\orientedSquare{2}{2}

\end{tikzpicture}
    \caption{
    Configuration-space picture of the Bockstein braiding process for $N=3$.
    The lattice point $(i,j)$ represents the state with $i$ units of the
    $A$-type excitation and $j$ units of the $B$-type excitation, with
    $i,j\in\mathbb Z_N$.  Horizontal and vertical edges correspond to the
    actions of $X$ and $Y$, respectively.  The diagonal staircase compares
    the two alternating histories $(YX)^N$ and $(XY)^N$.
    }
    \label{fig:path}
\end{figure}

\subsection{Relation to one-dimensional fusion statistics}

In $(1{+}1)$D, the Bockstein invariant can also be expressed through the
fusion statistics of Refs.~\cite{CarolynZhangSPTEntangler,kobayashi2024generalized}.
For three ordered points $i<j<k$, let
\begin{equation}
    Z_3^{(\alpha)}
    :=
    \left[
        \left(U_{ij}^{(\alpha)}\right)^N,
        U_{jk}^{(\alpha)}
    \right],
    \qquad \alpha=A,B,AB .
\end{equation}
After resolving the common fusion point into two points $j_-$ and $j_+$ with $i<j_-<j_+<k$, the mixed part of the fusion statistic is
\begin{equation}
    W_N\!\left(U^A_{ij_+},U^B_{kj_-}\right)
    =
    \frac{Z_3^{(AB)}}{Z_3^{(A)}Z_3^{(B)}} .
\end{equation}
Appendix~\ref{app:fusion-bockstein} proves this relation. It is the
one-dimensional analog of extracting mutual braiding from the topological spins of anyons in two dimensions: $B(a_1,a_2)=\theta(a_1a_2)/[\theta(a_1)\theta(a_2)]$.

\subsection{Quantization of the Bockstein braiding phase}

The configuration torus also makes quantization immediate. Translate $W_N$ in Fig.~\ref{fig:path} horizontally through its $N$ possible initial fusion labels. The oriented edges of the translated loops cancel pairwise, while initial-state independence gives the same phase to every process~\cite{kobayashi2024generalized,xue2025statistics}. Hence
\begin{equation}
    W_N(X,Y)^N=1,
    \qquad
    W_N(X,Y)=\exp(2\pi i k/N),\quad k\in\ZN .
    \label{eq:WN-quantization}
\end{equation}
The invariant is also bilinear in the fusion labels,
\begin{equation}
    W_N(X^\alpha,Y^\beta)=W_N(X,Y)^{\alpha\beta}.
    \label{eq:WN-bilinearity-main}
\end{equation}
Appendix~\ref{app:WN-linearity} proves bilinearity. A condensed
$\mathbb Z_{N^2}$ toric code realizes every primitive phase, as constructed
in Appendix~\ref{app:toric-code-realization}.

\subsection{One-dimensional symmetry patches with nontrivial Bockstein braiding}
\label{subsec:one-dimensional-patch-example}

\begin{figure}[t]
    \centering
    \begin{tikzpicture}[
        x=1.2cm,
        y=1.05cm,
        every node/.style={font=\small},
        a_site/.style={circle,draw=black,fill=black,inner sep=1.9pt},
        b_site/.style={circle,draw=black,fill=black,inner sep=1.9pt},
        chain/.style={black!65,line width=0.7pt},
        guide/.style={
            black!55,
            line width=0.6pt,
            dash pattern=on 2.3pt off 2.3pt
        },
        xgate/.style={font=\normalsize\bfseries,text=braidred},
        ygate/.style={font=\normalsize\bfseries,text=braidblue!85!black},
        redgate/.style={braidred,line width=1.35pt,line cap=round},
        bluegate/.style={braidblue!85!black,line width=1.35pt,line cap=round},
        czlabelred/.style={font=\scriptsize\bfseries,text=braidred,inner sep=0.7pt},
        czlabelblue/.style={font=\scriptsize\bfseries,text=braidblue!85!black,inner sep=0.7pt},
        slabel/.style={font=\scriptsize\bfseries,text=braidred,inner sep=0.6pt},
        paneltitle/.style={font=\bfseries},
        formula/.style={font=\small,inner sep=1.3pt}
    ]

    \def\ya{0.55}
    \def\yb{-0.55}
    \def\panelgap{3.75}
    \def\yA{0}
    \def\yB{-\panelgap}


    \newcommand{\tiltedguide}[4]{%
        \draw[guide]
            ($(#1,#3+\yb)!-0.3!(#2,#3+\ya)$)
            --
            ($(#2,#3+\ya)!-0.3!(#1,#3+\yb)$);
        \node[below=12pt] at (#1,#3+\yb) {#4};
    }

    \newcommand{\drawtwolayer}[1]{%
        \fill[grid!65]
            (2,#1+\yb)
            --
            (1.5,#1+\ya)
            --
            (4.5,#1+\ya)
            --
            (5,#1+\yb)
            -- cycle;

        \draw[chain] (-0.85,#1+\ya) -- (8.15,#1+\ya);
        \draw[chain] (-0.65,#1+\yb) -- (8.15,#1+\yb);

        \node[text=braidred!80!black,font=\itshape] at (-1.05,#1+\ya) {$a$};
        \node[text=braidblue!80!black,font=\itshape] at (-1.05,#1+\yb) {$b$};

        \tiltedguide{0}{-0.5}{#1}{$x$}
        \tiltedguide{2}{1.5}{#1}{$y'$}
        \tiltedguide{5}{4.5}{#1}{$x'$}
        \tiltedguide{7}{6.5}{#1}{$y$}

        \node[inner sep=1pt,font=\scriptsize] at (3.5,#1-0.15) {$I\cap J$};

        \begin{pgfonlayer}{foreground}
            \foreach \j in {0,1,...,8} {
                \pgfmathsetmacro{\apos}{\j-0.5}
                \node[a_site] at (\apos,#1+\ya) {};
            }
            \foreach \j in {0,1,...,8} {
                \node[b_site] at (\j,#1+\yb) {};
            }
        \end{pgfonlayer}
    }

    \drawtwolayer{\yA}

    \node[paneltitle,anchor=west] at (-0.75,\yA+1.50) {(a)};
    \node[formula,anchor=west,text=braidred!85!black] at (0.25,\yA+1.50)
    {$
        X_I
        =
        \left(
            \prod_{j=x+1}^{x'} X^a_j
        \right)
        \left(
            \prod_{r=x}^{x'-1}
            S^b_r\,
            \mathrm{CZ}^b_{r,r+1}\,
            S^b_{r+1}
        \right)
    $};

    \foreach \p in {0.5,1.5,2.5,3.5,4.5} {
        \node[xgate] at (\p,\yA+\ya+0.36) {$X$};
    }

    \foreach \r in {0,1,2,3,4} {
        \pgfmathtruncatemacro{\rp}{\r+1}
        \draw[redgate]
            (\r,\yA+\yb-0.08)
            to[out=-65,in=-115]
            (\rp,\yA+\yb-0.08);
        \node[czlabelred] at (\r+0.5,\yA+\yb-0.60) {$\mathrm{CZ}$};
    }

    \node[slabel] at (0,\yA+\yb+0.34) {$S$};
    \foreach \p in {1,2,3,4} {
        \node[slabel] at (\p,\yA+\yb+0.34) {$Z$};
    }
    \node[slabel] at (5,\yA+\yb+0.34) {$S$};

    \drawtwolayer{\yB}

    \node[paneltitle,anchor=west] at (-0.75,\yB+1.50) {(b)};
    \node[formula,anchor=west,text=braidblue!85!black] at (0.25,\yB+1.50)
    {$
        Y_J
        =
        \mathrm{CZ}_{a_{y'+1},b_{y'}}
        \mathrm{CZ}_{a_{y'+1},b_{y'+1}}
        \left(
            \prod_{j=y'+1}^{y} X^b_j
        \right)
        \mathrm{CZ}_{a_{y+1},b_y}
        \mathrm{CZ}_{a_{y+1},b_{y+1}}
    $};

    \foreach \p in {3,4,5,6,7} {
        \node[ygate] at (\p,\yB+\yb-0.3) {$X$};
    }

    \draw[bluegate]
        (2.5,\yB+\ya) --
        node[czlabelblue,pos=0.43,left=3pt] {$\mathrm{CZ}$}
        (2,\yB+\yb);

    \draw[bluegate]
        (2.5,\yB+\ya) --
        node[czlabelblue,pos=0.43,right=3pt] {$\mathrm{CZ}$}
        (3,\yB+\yb);

    \draw[bluegate]
        (7.5,\yB+\ya) --
        node[czlabelblue,pos=0.43,left=3pt] {$\mathrm{CZ}$}
        (7,\yB+\yb);

    \draw[bluegate]
        (7.5,\yB+\ya) --
        node[czlabelblue,pos=0.43,right=3pt] {$\mathrm{CZ}$}
        (8,\yB+\yb);

    \end{tikzpicture}
    \caption{Double-layer representation of the finite patch symmetry
    operators in the one-dimensional spin chain. Each unit cell contains an
    $a$ qubit and a $b$ qubit, with the $a$-qubit lattice shifted by half a
    lattice spacing. \textbf{(a)} The $A$-type patch $X_I$ restricts $U_A$ to
    $I=[x,x']$ and contains the red $X^a$ string and the $b$-layer factors
    $S_r^b\mathrm{CZ}_{r,r+1}^bS_{r+1}^b$. \textbf{(b)} The $B$-type patch
    $Y_J$ contains the onsite $X^b$ string on $J=[y',y]$, dressed by the four
    endpoint controlled-$Z$ gates shown in blue. The shaded parallelogram is
    the staggered overlap $I\cap J=[y',x']$.}
    \label{fig:one-dimensional-patch-operators}
\end{figure}
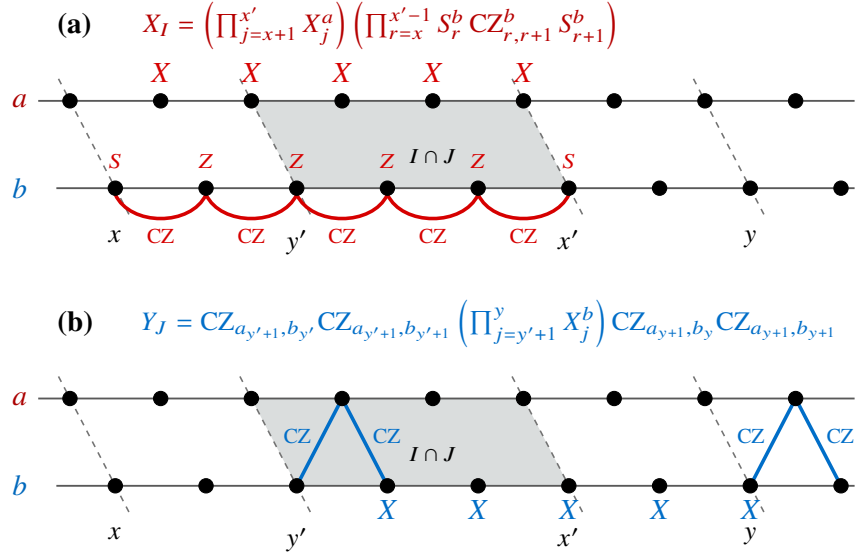

For a one-dimensional spin chain with two qubits, $a_j$ and $b_j$, per unit cell, take
\begin{equation}
    U_A
    =
    \left(\prod_j X_j^a\right)
    \left(\prod_j Z_j^b\,\mathrm{CZ}_{j,j+1}^b\right),
    \qquad
    U_B=\prod_j X_j^b .
    \label{eq:one-dimensional-symmetry-ZCZ-main}
\end{equation}
The onsite $a$-qubit factor $\prod_jX_j^a$ is auxiliary and does not affect the mixed anomaly, while the factor $\prod_jZ_j^b$ provides a convenient representative for the finite-patch calculation. On a closed chain with an even number of sites,
$\prod_jZ_j^b$ can be removed by a finite-depth change of basis that leaves
$U_B$ invariant, reducing $U_A$ to the familiar
$U_A^{\mathrm{CZ}}
=(\prod_jX_j^a)(\prod_j\mathrm{CZ}_{j,j+1}^b)$ representative.
Thus, the two representatives realize the same mixed anomaly between the
$\prod X$ and $\prod\mathrm{CZ}$ symmetries. The finite-depth circuit and the detailed patch calculations are given in
Appendix~\ref{app:patch-ZCZ-X-anomaly}.

Writing $S=\sqrt Z=\operatorname{diag}(1,i)$, truncate the two global
generators to staggered intervals $I=[x,x']$ and $J=[y',y]$ with
$x<y'<x'<y$. The resulting symmetry-preserving patch operators $X_I$ and
$Y_J$ are shown in Fig.~\ref{fig:one-dimensional-patch-operators}. Their
endpoint $S$ and controlled-$Z$ decorations make the open operators commute
with both global symmetries.
For these patches,
\begin{equation}
    [X_I,U_A]=[X_I,U_B]=[Y_J,U_A]=[Y_J,U_B]=1,
    \qquad
    W_2(X_I,Y_J)=-1 .
    \label{eq:one-dimensional-patch-W2-main}
\end{equation}
Thus the Bockstein word directly diagnoses the mixed anomaly between the
onsite spin-flip and the non-onsite controlled-$Z$ symmetry. Appendix
\ref{app:patch-ZCZ-X-anomaly} gives the complete patch definitions, symmetry
checks, and direct evaluation.

\subsection{Higher-dimensional lattice analogs}
For a $(2{+}1)$D realization, put one qubit on every edge of a triangular
lattice and define
\begin{equation}
    U_0
    =
    \prod_f\prod_{\{e,e'\}\subset\partial f}\mathrm{CZ}_{e,e'},
    \qquad
    U_1(v)=\prod_{e\ni v}X_e ,
    \label{eq:triangular-total-symmetries-main}
\end{equation}
where $f$ runs over triangular faces and $v$ over vertices. Fig.~\ref{fig:triangular-total-symmetries} displays the global $0$-form generator
and an elementary closed $1$-form transformation.

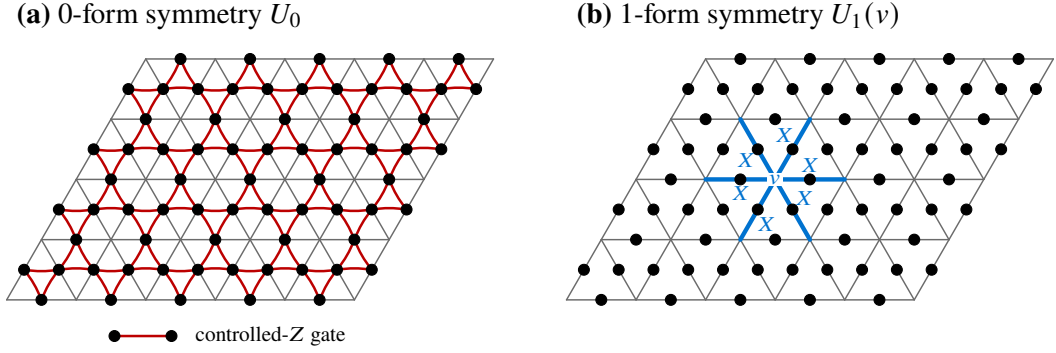
\begin{figure}[t]
    \centering
    \begin{tikzpicture}[
        x=0.92cm,
        y=0.92cm,
        latticeedge/.style={
            black!58,
            line width=0.55pt,
            line cap=round
        },
        edgequbit/.style={
            circle,
            draw=black,
            fill=black,
            inner sep=1.45pt
        },
        czlink/.style={
            braidred!88!black,
            line width=0.95pt,
            line cap=round,
            line join=round
        },
        starlink/.style={
            braidblue!88!black,
            line width=1.65pt,
            line cap=round
        },
        xlabel/.style={
            font=\scriptsize\bfseries,
            text=braidblue!88!black,
            inner sep=0.5pt
        },
        paneltitle/.style={
            font=\normalsize
        }
    ]

    \def\trih{0.8660254}

    \newcommand{\drawtrisymmetrylattice}{%
        \foreach \j in {0,...,4} {
            \foreach \i in {0,...,4} {
                \draw[latticeedge]
                    ({\i+0.5*\j},{\trih*\j})
                    --
                    ({\i+1+0.5*\j},{\trih*\j});
            }
        }
        \foreach \j in {0,...,3} {
            \foreach \i in {0,...,5} {
                \draw[latticeedge]
                    ({\i+0.5*\j},{\trih*\j})
                    --
                    ({\i+0.5*(\j+1)},{\trih*(\j+1)});
            }
            \foreach \i in {1,...,5} {
                \draw[latticeedge]
                    ({\i+0.5*\j},{\trih*\j})
                    --
                    ({\i-1+0.5*(\j+1)},{\trih*(\j+1)});
            }
        }
    }

    \newcommand{\drawtriedgequbits}{%
        \foreach \j in {0,...,4} {
            \foreach \i in {0,...,4} {
                \node[edgequbit] at
                    ({\i+0.5+0.5*\j},{\trih*\j}) {};
            }
        }
        \foreach \j in {0,...,3} {
            \foreach \i in {0,...,5} {
                \node[edgequbit] at
                    ({\i+0.25+0.5*\j},{\trih*(\j+0.5)}) {};
            }
            \foreach \i in {1,...,5} {
                \node[edgequbit] at
                    ({\i-0.25+0.5*\j},{\trih*(\j+0.5)}) {};
            }
        }
    }

    \newcommand{\drawtricznnetwork}{%
        \foreach \j in {0,...,3} {
            \foreach \i in {0,...,4} {
                \draw[czlink]
                    ({\i+0.5+0.5*\j},{\trih*\j})
                    to[bend left=8]
                    ({\i+0.75+0.5*\j},{\trih*(\j+0.5)})
                    ;
                \draw[czlink]
                    ({\i+0.75+0.5*\j},{\trih*(\j+0.5)})
                    to[bend left=8]
                    ({\i+0.25+0.5*\j},{\trih*(\j+0.5)})
                    ;
                \draw[czlink]
                    ({\i+0.25+0.5*\j},{\trih*(\j+0.5)})
                    to[bend left=8]
                    ({\i+0.5+0.5*\j},{\trih*\j});
            }
            \foreach \i in {1,...,5} {
                \draw[czlink]
                    ({\i+0.25+0.5*\j},{\trih*(\j+0.5)})
                    to[bend left=8]
                    ({\i+0.5*\j},{\trih*(\j+1)})
                    ;
                \draw[czlink]
                    ({\i+0.5*\j},{\trih*(\j+1)})
                    to[bend left=8]
                    ({\i-0.25+0.5*\j},{\trih*(\j+0.5)})
                    ;
                \draw[czlink]
                    ({\i-0.25+0.5*\j},{\trih*(\j+0.5)})
                    to[bend left=8]
                    ({\i+0.25+0.5*\j},{\trih*(\j+0.5)});
            }
        }
    }

    \begin{scope}
        \node[paneltitle,anchor=west] at (0,4.08)
            {\textbf{(a)} 0-form symmetry $U_0$};
        \drawtrisymmetrylattice
        \drawtricznnetwork

        \drawtriedgequbits

        \begin{scope}[shift={(1.55,-0.52)}]
            \draw[czlink] (0,0) -- (0.82,0);
            \node[edgequbit] at (0,0) {};
            \node[edgequbit] at (0.82,0) {};
            \node[font=\scriptsize,anchor=west] at (1.02,0)
                {controlled-$Z$ gate};
        \end{scope}
    \end{scope}

    \begin{scope}[shift={(8.05,0)}]
        \node[paneltitle,anchor=west] at (0,4.08)
            {\textbf{(b)} 1-form symmetry $U_1(v)$};
        \drawtrisymmetrylattice

        \draw[starlink] (3.00,{2*\trih}) -- (4.00,{2*\trih});
        \draw[starlink] (3.00,{2*\trih}) -- (2.00,{2*\trih});
        \draw[starlink] (3.00,{2*\trih}) -- (3.50,{3*\trih});
        \draw[starlink] (3.00,{2*\trih}) -- (2.50,{3*\trih});
        \draw[starlink] (3.00,{2*\trih}) -- (2.50,{\trih});
        \draw[starlink] (3.00,{2*\trih}) -- (3.50,{\trih});

        \drawtriedgequbits

        \node[font=\scriptsize,text=braidblue!88!black,
              fill=white,inner sep=0.6pt] at (3.00,{2*\trih}) {$v$};

        \node[xlabel,above=2pt]
            at (3.50,{2*\trih}) {$X$};
        \node[xlabel,above left=2pt]
            at (3.35,{2.5*\trih}) {$X$};
        \node[xlabel,below left=2pt]
            at (2.8,{2.55*\trih}) {$X$};
        \node[xlabel,below=2pt]
            at (2.50,{2*\trih}) {$X$};
        \node[xlabel,below right=2pt]
            at (2.65,{1.5*\trih}) {$X$};
        \node[xlabel,above right=2pt]
            at (3.2,{1.45*\trih}) {$X$};
    \end{scope}

    \end{tikzpicture}
    \caption{Triangular-lattice representation of the
    $\mathbb Z_2^{(0)}\times\mathbb Z_2^{(1)}$ symmetry generators in
    Eq.~\eqref{eq:triangular-total-symmetries-main}. A qubit (black dot) lives
    on every lattice edge. \textbf{(a)} Each triangular face contributes three pairwise controlled-$Z$ gates among its edge qubits, indicated by the red links; the product over all triangular faces is the global 0-form generator $U_0$.
    \textbf{(b)} The elementary closed $1$-form transformation is
    $U_1(v)=\prod_{e\ni v}X_e$; the six blue edges mark its support.}
    \label{fig:triangular-total-symmetries}
\end{figure}

Open disk and line patches of these generators have particle--loop
Bockstein braiding $W_2=-1$. Their explicit construction is given in
Appendix~\ref{app:square-lattice-2d}. The corresponding $(3{+}1)$D model
realizes an anomalous $\ZZ_2^{(0)}\times\ZZ_2^{(2)}$ symmetry and
particle--membrane braiding; see Appendix~\ref{app:cubic-lattice-3d}.

\section{Physical implications}
\label{sec:physical-implications}

In the preceding section, we formulated Bockstein braiding in the symmetry
language by truncating closed symmetry operators to finite patches. Their
boundaries carry the corresponding excitations, and $W_N$ detects the mixed
anomaly of the underlying symmetries. Conversely, placing the excitation operator on a closed support turns it into a generalized-symmetry generator.
An $r$-dimensional excitation couples to a background field of degree $d-r$, and its closed operator generates a $\ZN$ $(d-r-1)$-form symmetry
\cite{Gaiotto:2014kfa,Kapustin:2013uxa}. We now use this identification to
derive constraints on infrared dynamics and to construct continuum
gauge-theory realizations.

\subsection{Consequences of the Bockstein anomaly}

Consider a quantum system in $d=p+q+1$ spatial dimensions with a
$\mathbb Z_N$ $p$-form symmetry and a $\mathbb Z_N$ $q$-form symmetry.
Suppose that these symmetries have the mixed Bockstein anomaly~\eqref{eq:bockstein-response-intro}.
We use the standard completeness property of a local unitary
topological order, also known as remote detectability: a topological
defect that has trivial linking with every genuine operator of
complementary dimension lies in the condensation completion of the
infrared defect theory. Such a condensation defect admits a
topological boundary condition and can therefore be cut open and
capped, up to a local normalization factor~\cite{Theo2022Classification, Shao2023HigherGauging, Cordova:2019bsd}.

\begin{theorem}
\label{thm:no-symmetric-gapped-phase}
No fully gapped phase can preserve both symmetries in the presence of the Bockstein anomaly in Eq.~\eqref{eq:bockstein-response-intro}, with $k\not\equiv0\pmod N$.
\end{theorem}

\begin{proof}
    Suppose, to the contrary, that the infrared theory is fully gapped
    and preserves both symmetries. At long distances it is described
    by a local unitary TQFT. If the TQFT decomposes into several
    superselection sectors, we work in an indecomposable sector
    preserved by both symmetries.

    Let $\mathcal D_A$ be the symmetry defect Poincar\'e dual to $A$.
    Since $A$ is a $(q+1)$-form background field, $\mathcal D_A$
    generates the $\mathbb Z_N$ $q$-form symmetry. Because this
    symmetry is unbroken, the infrared TQFT contains no genuine
    $q$-dimensional topological operator carrying a nontrivial
    $\mathbb Z_N$ charge. Equivalently, every genuine topological
    operator that can link with $\mathcal D_A$ has trivial linking
    phase with it.
    By remote detectability, $\mathcal D_A$ must therefore be a
    condensation defect. Consequently, it admits a topological gapped
    boundary and may be cut open and capped without changing
    normalized correlators.\footnote{For an unbroken $r$-form
    symmetry with $r\leq1$, Proposition~1 of
    Ref.~\cite{Cordova:2019bsd} proves this cut-and-cap statement
    directly. Here we instead use the general
    remote-detectability/condensation-completion property of a local
    TQFT, which does not impose a restriction on $r$.}

    Let $\mathcal J_B$ be the $N$-fold fusion junction represented by
    $\beta_N B$. This is a $q$-dimensional junction operator in the
    sector twisted by the $B$-type symmetry defects, rather than a
    genuine operator of the untwisted infrared theory. Denote by
    $Z_{\mathrm{link}}$ the correlator in which $\mathcal D_A$ links
    once with $\mathcal J_B$, and by $Z_{\mathrm{unlink}}$ the
    otherwise identical unlinked correlator. Anomaly inflow gives
    \begin{equation}
        \frac{Z_{\mathrm{link}}}{Z_{\mathrm{unlink}}}
        =
        \exp\!\left(\frac{2\pi i k}{N}\right)
        \neq 1,
    \end{equation}
    because the two configurations differ by one unit of
    $\int A\cup\beta_N B$.

    Now cut $\mathcal D_A$ open in the same small region, disjoint
    from $\mathcal J_B$, in both correlators. In the linked
    configuration, move the capped end around $\mathcal J_B$ to
    remove the linking and then rejoin the defect. Since the cap is
    topological, this deformation does not change the normalized
    correlator. Moreover, the identical local cap factors in
    $Z_{\mathrm{link}}$ and $Z_{\mathrm{unlink}}$ cancel in their
    ratio. The two resulting configurations are identical, and hence $Z_{\mathrm{link}} = Z_{\mathrm{unlink}}$.
    This contradicts $k\not\equiv0\pmod N$. Therefore, no fully gapped
    phase can preserve both symmetries.
\end{proof}

A system with this anomaly must therefore be gapless
or spontaneously break at least one symmetry; related criteria for
symmetry-enforced gaplessness appear in Ref.~\cite{Hsin:2025ria}.
Moreover, the $q$-dimensional excitation associated with the $p$-form
symmetry carries charge $k/N$ under the $q$-form symmetry, and conversely
\cite{Hsin:2019fhf,Hsin:2025jot,Hsin:2025ido}. Thus one symmetry is
fractionalized when the other is broken. Equivalently, the anomaly obstructs
simultaneous gauging, or simultaneous condensation of the generator
excitations. Extending either symmetry from $\ZN$ to
$\mathbb Z_{N^2}$ trivializes the Bockstein class and removes the
obstruction.

\subsection{Continuum field-theory constructions}

We now describe continuum gauge theories that realize the same Bockstein
anomaly. For concreteness, take $d=3$ and $p=q=1$, so both relevant
symmetries are $\ZN$ 1-form symmetries.

\paragraph{Abelian gauge theory}

A $(3+1)$D $U(1)$ gauge theory coupled to a charge-$N$ Higgs field has an
electric $\ZN$ center 1-form symmetry and a magnetic $\ZN$ subgroup with
the $k=1$ Bockstein anomaly
\cite{Hansson_2004,Benini:2018reh,Hsin:2020nts,Hsin:2025ido}. The excitations
exhibiting this braiding statistic are loops carrying fractional charges and
fluxes. The uncondensed regime is gapless Maxwell theory.
After Higgs condensation, the
infrared theory is $\ZN$ gauge theory: the electric symmetry is broken and the magnetic symmetry is fractionalized on strings carrying flux $2\pi/N$.

\paragraph{Non-Abelian gauge theory}

Likewise, $\mathrm{SU}(N^2)/\mathbb Z_N$ gauge theory with
$N_f$ adjoint Weyl fermions has center and magnetic
$\mathbb Z_N$ 1-form symmetries with the same anomaly, originating from
the nonsplit extension
$1\to\mathbb Z_N\to SU(N^2)\to SU(N^2)/\mathbb Z_N\to1$. This anomaly is discussed in
Refs.~\cite{Benini:2018reh,Hsin:2020nts,Hsin:2025ido}. The excitations
exhibiting this braiding statistic are loops carrying fractional charges and
fluxes. The dynamics depends on the number $N_f$ of adjoint Weyl fermions and is generally unknown, except in a few cases with small $N_f$, where the
theory is believed to be strongly interacting and confining. In these cases,
the infrared dynamics matches the Bockstein anomaly.
The detailed Abelian and non-Abelian phase analyses are
given in Appendix~\ref{app:continuum-gauge-theories}.

\section{Conclusion and outlook}
\label{sec:conclusion}

We introduced Bockstein braiding as a robust mutual statistic for $\ZN$
excitations in the adjacent spatial dimension $d=p+q+1$. The staggered word
$W_N=(Y^{-1}X^{-1})^N(YX)^N$ cancels local ambiguities while retaining the
linking phase between one excitation and the $N$-fold fusion junction of the
other. Together with the fusion-junction interpretation, its quantization and bilinearity match the response $A\cup\beta_NB$.

For open symmetry operators, the statistical invariant measures a mixed generalized-symmetry anomaly.
A nontrivial value obstructs simultaneous condensation and gauging, excludes a fully symmetric gapped phase, and implies the corresponding symmetry fractionalization.
Extensions to non-invertible excitations and higher-group symmetries are natural directions for future work.

\section*{Acknowledgments}

Y.-A.C. is supported by the National Natural Science Foundation of China (Grant No.~12474491) and the Fundamental Research Funds for the Central Universities, Peking University.
P.-S.H. is supported by the Department of Mathematics, King's College London.

We thank Qingrui Wang for raising the question that motivated this work and Ruben Verresen for pointing out the one-dimensional example.
We thank the Croucher Summer School at The Chinese University of Hong Kong for hospitality and for providing a stimulating environment in which this work was initiated. Y.-A.C. also acknowledges the support of the Tsung-Dao Lee Center for Sciences and Arts (TDLCSA) in Dali, where part of this work was developed during the Generalized Symmetry Workshop.

\ifincludeappendices
\appendix

\section{Local cancellation}
\label{app:local-cancellation}

The freedom in choosing microscopic configuration states and
excitation-moving operators, together with the invariance of closed
statistical phases under local changes of these choices, has been studied in
detail in Refs.~\cite{FHH21,kobayashi2024generalized}. We do not repeat the
general arguments here. Instead, we restrict to Abelian defects for which
each fixed-configuration sector is one-dimensional. A local symmetric
unitary that preserves both the excitation configuration and the defect
superselection sector then changes only the representative of the
configuration state and acts within that sector by a phase.

This appendix shows directly how these phases cancel in the Bockstein word
$W_N$, first in the one-dimensional local model and then in its
higher-dimensional extension.

\subsection{Local cancellation in one spatial dimension}
\label{subsec:one-dimensional-cancellation}

The cancellation mechanism is already visible for two pointlike excitation
species in one spatial dimension. Choose four points ordered as in
Fig.~\ref{fig:ordinary-and-bockstein-braiding}\subref{fig:bockstein-braiding}:
\begin{equation}
    x<y'<x'<y .
    \label{eq:interval-model}
\end{equation}
Let $X$ move an $a$ particle from $x$ to $x'$, leaving the opposite charge at
$x$, and let $Y$ move a $b$ particle from $y$ to $y'$, leaving the opposite
charge at $y$. The ordering in Eq.~\eqref{eq:interval-model} is the
one-dimensional local model of the staggered overlap. Starting from the
vacuum, $(YX)^N$ returns the visible excitation configuration to the vacuum
by $\mathbb Z_N$ fusion, and $W_N$ measures its relative phase with the other
staircase $(XY)^N$.

Now decorate $X$ locally at a point $v\in[x,x']$ by
$X\mapsto X'=XO_v$. Because $O_v$ preserves the excitation configuration
and its one-dimensional superselection sector, its action on a configuration
state takes the form
\begin{equation}
    O_v|a\rangle
    =
    e^{i\varphi_v(a)}|a\rangle ,
    \label{eq:local-perturbation-phase}
\end{equation}
where $\varphi_v(a)$ depends only on the configuration near $v$.

Denote a
typical intermediate configuration by $\OverlapWord{i}{j}$, meaning that $i$
units of $a$ charge and $j$ units of $b$ charge have been transported across
the overlap, with $i,j\in\ZN$.

In the history $(YX)^N$, the perturbation contributes
\begin{equation}
    \Delta_{YX}
    =
    \sum_{i=0}^{N-1}
    \varphi_v\left(\OverlapWord{i}{i}\right).
    \label{eq:phase-shift-YX}
\end{equation}
In the history $(XY)^N$, it contributes
\begin{equation}
    \Delta_{XY}
    =
    \sum_{i=0}^{N-1}
    \varphi_v\left(\OverlapWord{i}{i+1}\right).
    \label{eq:phase-shift-XY}
\end{equation}
We claim that $\Delta_{YX}-\Delta_{XY}=0$.  The term-by-term pairing depends
on the position of $v$.  If $v\neq y'$, locality cannot see the change of
$b$ charge at $y'$, so
\begin{equation}
    \varphi_v\left(\OverlapWord{i}{i}\right)
    =
    \varphi_v\left(\OverlapWord{i}{i+1}\right).
    \label{eq:first-local-pairing}
\end{equation}
If instead $v\notin\{x,x'\}$, locality cannot see the change of $a$ charge at
the endpoints $x,x'$, so
\begin{equation}
    \varphi_v\left(\OverlapWord{i}{i}\right)
    =
    \varphi_v\left(\OverlapWord{i-1}{i}\right).
    \label{eq:second-local-pairing}
\end{equation}
After reindexing $i$, Eq.~\eqref{eq:second-local-pairing} again cancels the
two sums.  For every $v\in[x,x']$, at least one of
Eqs.~\eqref{eq:first-local-pairing} and \eqref{eq:second-local-pairing}
applies: the endpoints $x,x'$ are handled by the first pairing, while the
special point $y'$ is handled by the second.  Hence
\begin{equation}
    \langle\Omega|W_N(X,Y)|\Omega\rangle
    =
    \langle\Omega|W_N(X',Y)|\Omega\rangle .
    \label{eq:local-robustness-vacuum}
\end{equation}
The same argument applies to local perturbations of $Y$.
It also does not use any special property of the vacuum; adding an arbitrary initial background configuration to every intermediate state leaves the same local pairings.

\subsection{Higher-dimensional local cancellation}
\label{subsec:loop-loop-cancellation}

The one-dimensional cancellation is the local model in all dimensions.  Near a
generic point of $\supp(X)\cap\supp(Y)$, the two supports overlap along a line,
while the transverse directions play no role.  Thus any local perturbation sees
the same staggered sequence analyzed in
Appendix~\ref{subsec:one-dimensional-cancellation}.

For concreteness, consider two loop excitations in three spatial dimensions.
The operators $X$ and $Y$ are supported on two disks that intersect cleanly
along a line segment, shown in Fig.~\ref{fig:ordinary-and-bockstein-braiding}\subref{fig:bockstein-braiding}.  A local intermediate configuration may be denoted by
$\LinkedEllipses{i}{j}$, with the $X$ loop drawn vertically and the $Y$ loop
horizontally.  Under a local perturbation $X\mapsto XO_v$, the change in the
accumulated phase is
\begin{equation}
    \delta_v\Theta
    =
    \sum_{i=0}^{N-1}\varphi_v\left(\LinkedEllipses{i}{i}\right)
    -
    \sum_{i=0}^{N-1}\varphi_v\left(\LinkedEllipses{i}{i+1}\right).
    \label{eq:loop-local-phase-shift}
\end{equation}
Depending on the position of $v$, locality pairs terms either at fixed $i$ or
after the cyclic shift $i\mapsto i-1$.  Since the sum runs over the full
$\ZN$ orbit, both pairings give $\delta_v\Theta=0$.
Therefore,
\begin{equation}
    \langle a|W_N(X,Y)|a\rangle
    =
    \langle a|W_N(XO_v,Y)|a\rangle~,
    \label{eq:higher-dimensional-local-robustness}
\end{equation}
for any initial configuration $|a\rangle$.  The same statement holds for local
perturbations of $Y$.
Thus, the cancellation applies whenever $d=p+q+1$.\footnote{
The particle-membrane statistics for $N=2$ in $(3{+}1)$D was previously studied in
Ref.~\cite{kobayashi2024generalized}.  In Appendix~\ref{app:eq41}, we show
that the process introduced there is equivalent to the present Bockstein braiding process.
}

\section{Relation between Bockstein braiding statistics and the fusion statistics in $(1{+}1)$D}
\label{app:fusion-bockstein}

In this appendix, we relate the Bockstein braiding process to the mixed part of the fusion statistics in one spatial dimension. 
We follow the convention used in Ref.~\cite{kobayashi2024generalized}.  If $U(s)$ is a hopping operator along the segment $s$ and
$a$ is an excitation configuration, then
\begin{equation}
    U(s)|a\rangle
    =
    e^{i\theta_s(a)}
    |a+\partial s\rangle .
    \label{eq:B-theta-convention}
\end{equation}
Thus, $\theta_s(a)$ is the phase obtained by applying $U(s)$ to the configuration
$a$, and $\partial s$ is the change of configuration produced by $U(s)$.  If a word of operators is closed, its phase is obtained by summing the
corresponding $\theta$'s along the intermediate configurations; inverse
operators contribute with the opposite sign.

We also use locality identities.
A typical locality identity is a nested commutator
\begin{equation}
    [U(s_n),[U(s_{n-1}),\cdots,[U(s_2),U(s_1)]\cdots]]=1~,
\end{equation}
whenever the common intersection of the supports is empty:
\begin{equation}
    \supp(s_1)\cap\supp(s_2)\cap\cdots\cap\supp(s_n)=\emptyset .
\end{equation}
In that case the nested commutator is forced to be the identity operator.
Equivalently, after evaluating it on any configuration state and expanding
each unitary action in terms of its phase \(\theta_s(a)\), the corresponding
integer linear combination of \(\theta\)'s must vanish modulo \(2\pi\).
These vanishing combinations are the locality identities that we quotient by.
We write $P\sim_{\loc}P'$ when the two expressions differ
by such zero sums.  Equivalently, they give the same statistical phase.

Let $i<j<k$, and let $U^{(\alpha)}_{v_1 v_2}$ denote the hopping operator that
moves a particle of type $\alpha$ from $v_1$ to $v_2$.  We use the commutator
convention
\begin{equation}
    [P,Q]=P^{-1}Q^{-1}PQ .
\end{equation}
For each particle type $\alpha\in\{A,B,AB\}$, define the $(1{+}1)$D fusion
statistic $Z_3^{(\alpha)}$, following
Ref.~\cite{kobayashi2024generalized}, by
\begin{equation}
    Z_3^{(\alpha)}
    =
    \left[
        \left(U^{(\alpha)}_{ij}\right)^N,
        U^{(\alpha)}_{jk}
    \right].
    \label{eq:B-fusion-stat-def}
\end{equation}
The mixed part is
\begin{equation}
    Z_{\mathrm{mix}}
    :=
    \frac{Z_3^{(AB)}}{Z_3^{(A)}Z_3^{(B)}} .
    \label{eq:B-Zpol-def}
\end{equation}
Here, ``mixed'' means that we take the fusion statistics of the composite
$AB$ particle and divide out the pure $A$ and pure $B$ fusion statistics.

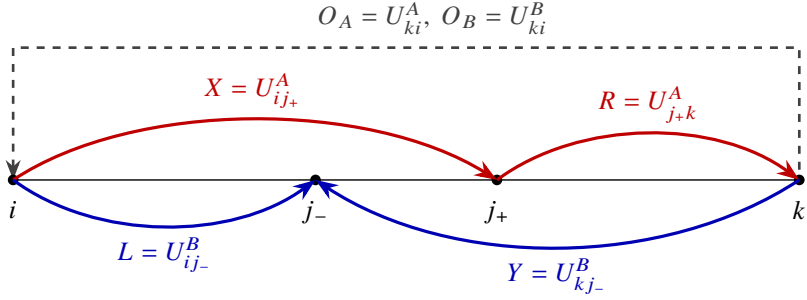
\begin{figure}[tbp]
\centering
\begin{tikzpicture}[
    x=2cm,
    y=1.25cm,
    >=Stealth,
    every node/.style={font=\small}
]
    \coordinate (i)  at (0,0);
    \coordinate (jm) at (2,0);
    \coordinate (jp) at (3.2,0);
    \coordinate (k)  at (5.2,0);

    \draw[black!75, line width=0.7pt] (i) -- (k);

    \draw[black!75, line width=1pt, dashed, <-]
    (i) -- ++(0,1.4) -- node[above=1pt] {$\qquad O_A=U^A_{ki},\ O_B=U^B_{ki}$} ++(5.2,0) -- (k);

    \foreach \p/\lab in {i/$i$, jm/$j_-$, jp/$j_+$, k/$k$} {
        \fill (\p) circle (2pt);
        \node[below=4pt] at (\p) {\lab};
    }

    \draw[red!75!black, very thick, ->]
        (i) .. controls (0.8,0.85) and (2.35,0.85) ..
        node[above] {$X=U^A_{i j_+}$} (jp);

    \draw[red!75!black, very thick, ->]
        (jp) .. controls (3.75,0.65) and (4.65,0.65) ..
        node[above] {$R=U^A_{j_+ k}$} (k);

    \draw[blue!70!black, very thick, ->]
        (i) .. controls (0.55,-0.65) and (1.45,-0.65) ..
        node[below] {$L=U^B_{i j_-}$} (jm);

    \draw[blue!70!black, very thick, ->]
        (k) .. controls (4.3,-0.95) and (2.9,-0.95) ..
        node[below] {$Y=U^B_{k j_-}$} (jm);
\end{tikzpicture}
\caption{
Resolved representative of the \(AB\)-particle hopping operators near the
fusion point on the compactified one-dimensional geometry. The fusion point
\(j\) is split into \(j_-\) and \(j_+\), and the dashed outer arc indicates the
additional edge connecting \(k\) back to \(i\). The \(A\)-type hopping operators
\(X=U^A_{i j_+}\) and \(R=U^A_{j_+ k}\) are drawn above the interval, while the
\(B\)-type hopping operators \(L=U^B_{i j_-}\) and \(Y=U^B_{k j_-}\) are drawn
below it. The orientation of \(Y\) is right-to-left, matching the Bockstein
braiding convention.
}
\label{fig:resolved-ab-hopping}
\end{figure}

We now resolve the fusion point $j$ into two nearby points
\begin{equation}
    i<j_-<j_+<k .
    \label{eq:B-resolved-fusion-point}
\end{equation}
We choose the following local representative of the $AB$ hopping operators:
\begin{equation}
    U^{(AB)}_{ij}
    :=
    U^B_{ij_-}U^A_{ij_+},
    \qquad
    U^{(AB)}_{jk}
    :=
    U^B_{j_-k}U^A_{j_+k}.
    \label{eq:B-resolved-AB-hopping}
\end{equation}
For compactness, define
\begin{equation}
    L:=U^B_{ij_-},
    \qquad
    X:=U^A_{ij_+},
    \qquad
    Y:=U^B_{kj_-},
    \qquad
    R:=U^A_{j_+k}.
    \label{eq:B-LXYR-def}
\end{equation}
Thus $Y$ is the right-to-left $B$ hopping operator, matching the Bockstein
braiding convention in the main text.  We also write
\begin{equation}
    \widetilde Y:=Y^{-1}=U^B_{j_-k}.
    \label{eq:B-Ytilde-def}
\end{equation}
Then
\begin{equation}
    U^{(AB)}_{ij} = LX,
    \qquad
    U^{(AB)}_{jk} = \widetilde Y R .
    \label{eq:B-AB-LX-YtildeR}
\end{equation}
The resolved geometry is a circle: besides the displayed segments between $i,j_-,j_+$, and $k$, there is an outside edge from
$k$ back to $i$.  Let
\begin{equation}
    O_A:=U^A_{ki},
    \qquad
    O_B:=U^B_{ki}.
    \label{eq:B-outside-operators}
\end{equation}
Set
\begin{equation}
    l:=\partial L,\quad
    x:=\partial X,\quad
    y:=\partial Y,\quad
    \widetilde y:=\partial\widetilde Y=-y,\quad
    r:=\partial R,\quad
    o_A:=\partial O_A,\quad
    o_B:=\partial O_B .
    \label{eq:B-boundary-symbols}
\end{equation}
The local charge changes are summarized in Table~\ref{tab:B-fusion-boundaries}.
The outside edges close the circle, giving
\begin{equation}
    x+r+o_A=0,
    \qquad
    l+\widetilde y+o_B=0.
    \label{eq:B-circle-relations}
\end{equation}

\begin{table}[t]
    \centering
    \begin{tabular}{c|cccc}
        \toprule
        boundary & $i$ & $j_-$ & $j_+$ & $k$ \\
        \midrule
        $l=\partial U^B_{ij_-}$
            & $-B$ & $+B$ & $0$ & $0$ \\
        $\widetilde y=\partial U^B_{j_-k}$
            & $0$ & $-B$ & $0$ & $+B$ \\
        $y=\partial U^B_{kj_-}=-\widetilde y$
            & $0$ & $+B$ & $0$ & $-B$ \\
        $o_B=\partial U^B_{ki}$
            & $+B$ & $0$ & $0$ & $-B$ \\
        \midrule
        $x=\partial U^A_{ij_+}$
            & $-A$ & $0$ & $+A$ & $0$ \\
        $r=\partial U^A_{j_+k}$
            & $0$ & $0$ & $-A$ & $+A$ \\
        $o_A=\partial U^A_{ki}$
            & $+A$ & $0$ & $0$ & $-A$ \\
        \bottomrule
    \end{tabular}
    \caption{
        Boundary conventions for the resolved circle.  A plus sign means that
        the corresponding hopping operator creates that particle at the vertex,
        and a minus sign means that it creates the opposite charge.
    }
    \label{tab:B-fusion-boundaries}
\end{table}

If $P$ is a closed word of hopping operators and $a$ is an
initial configuration, we write $(P,a)$ for the additive phase accumulated by applying $P$ to $|a\rangle$.  Thus,
\[
    \langle a|U(P)|a\rangle
    =
    \exp\bigl(i(P,a)\bigr).
\]
Equivalently, \((P,a)\) is obtained by summing the \(\theta\)-phases of all
operator steps along the process, with a minus sign for inverse operators.
All such phases are understood modulo \(2\pi\).  In particular,
\[
    (Z_3^{(A)},0)
\]
means the additive phase of the fusion-statistics process \(Z_3^{(A)}\)
evaluated in the vacuum (zero) initial configuration.

We first expand the fusion statistics.  Since
\[
    Z_3^{(A)}=[X^N,R]=(X^N)^{-1}R^{-1}X^NR ,
\]
we evaluate the word from right to left.  Starting from the vacuum
configuration \(0\), the operator \(R\) first shifts the configuration by
\(r=\partial R\).  Then the \(N\) applications of \(X\) contribute
\[
    \sum_{m=0}^{N-1}\theta_X(r+mx).
\]
Afterwards \(R^{-1}\) returns the \(R\)-endpoint, and the inverse word
\((X^N)^{-1}\) contributes
\[
    -\sum_{m=0}^{N-1}\theta_X(mx).
\]
The phases from \(R\) and \(R^{-1}\) cancel because the process is closed.
Therefore
\begin{equation}
    (Z_3^{(A)},0)
    =
    \sum_{m=0}^{N-1}
    \left[
        \theta_X(r+mx)-\theta_X(mx)
    \right].
    \label{eq:B-ZA-theta-expanded}
\end{equation}
Similarly, since \(Z_3^{(B)}=[L^N,\widetilde Y]\), we obtain
\begin{equation}
    (Z_3^{(B)},0)
    =
    \sum_{m=0}^{N-1}
    \left[
        \theta_L(\widetilde y+ml)-\theta_L(ml)
    \right].
    \label{eq:B-ZB-theta-expanded}
\end{equation}
For the composite particle,
\begin{equation}
    Z_3^{(AB)}
    =
    \left[(LX)^N,\widetilde YR\right].
\end{equation}
One copy of $LX$, applied to a configuration $a$, contributes
$\theta_X(a)+\theta_L(a+x)$, because $X$ acts first and then $L$.  The
operator $\widetilde YR$ shifts the configuration by $\widetilde y+r$.
Therefore
\begin{equation}
\begin{aligned}
    (Z_3^{(AB)},0)
    =
    \sum_{m=0}^{N-1}
    \Big[
        &\theta_X(\widetilde y+r+m(l+x))
        -\theta_X(m(l+x)) 
        +\theta_L(\widetilde y+r+m(l+x)+x)
        -\theta_L(m(l+x)+x)
    \Big].
\end{aligned}
\label{eq:B-ZAB-theta-expanded}
\end{equation}
Subtracting Eqs.~\eqref{eq:B-ZA-theta-expanded} and
\eqref{eq:B-ZB-theta-expanded} from Eq.~\eqref{eq:B-ZAB-theta-expanded}, the
mixed expression is
\begin{equation}
\begin{aligned}
    &(Z_3^{(AB)},0)-(Z_3^{(A)},0)-(Z_3^{(B)},0) 
    \\&=
    \sum_{m=0}^{N-1}
    \Big[
        \theta_X(\widetilde y+r+m(l+x))
        -\theta_X(m(l+x))
        -\theta_X(r+mx)
        +\theta_X(mx) \\
        & \qquad +\theta_L(\widetilde y+r+m(l+x)+x)
        -\theta_L(m(l+x)+x)
        -\theta_L(\widetilde y+ml)
        +\theta_L(ml)
    \Big].
\end{aligned}
\label{eq:B-epol-theta-expanded}
\end{equation}
To simplify this expression, it is useful to introduce the commutator density
\begin{equation}
    K(P,Q;a)
    :=
    \theta_P(a)+\theta_Q(a+\partial P)
    -\theta_P(a+\partial Q)-\theta_Q(a).
    \label{eq:B-K-density-def}
\end{equation}
This is the expression for the phase of the commutator $[Q,P]$ in the initial configuration $a$.
For example, if $\supp(P)\cap\supp(Q)=\emptyset$, then locality gives $K(P,Q;a)=0$ modulo $2\pi$ for every $a$.

We will also use the following nested locality identity.  If
\begin{equation}
    \supp(P)\cap\supp(Q)\cap\supp(O)=\emptyset,
\end{equation}
then the nested commutator $[O,[Q,P]]$ is the identity.  In terms of
$\theta$'s, this says
\begin{equation}
    K(P,Q;a+\partial O)-K(P,Q;a) = 0 \pmod{2\pi}.
    \label{eq:B-nested-locality-shift}
\end{equation}
This is the basic rule that lets us shift the background configuration of a
commutator density.

Using $K$, the fusion statistics are
\begin{align}
    (Z_3^{(A)},0)
    &=
    \sum_{m=0}^{N-1}K(R,X;mx), \label{eq:B-ZA-K-form}\\
    (Z_3^{(B)},0)
    &=
    \sum_{m=0}^{N-1}K(\widetilde Y,L;ml), \label{eq:B-ZB-K-form}\\
    (Z_3^{(AB)},0)
    &=
    \sum_{m=0}^{N-1}K(\widetilde YR,LX;m(l+x)). \label{eq:B-ZAB-K-form}
\end{align}
The product in the last line expands algebraically as
\begin{equation}
\begin{aligned}
    K(\widetilde YR,LX;a)
    =
    &K(R,X;a)
    +K(R,L;a+x)
    +K(\widetilde Y,X;a+r)
    +K(\widetilde Y,L;a+r+x).
\end{aligned}
\label{eq:B-K-product-expansion}
\end{equation}
This identity follows by expanding both sides in $\theta$'s, using that
$LX$ means $X$ acts first and then $L$, while $\widetilde YR$ means $R$ acts
first and then $\widetilde Y$.

Substituting $a=m(l+x)$ in Eq.~\eqref{eq:B-K-product-expansion}, we compare
with Eqs.~\eqref{eq:B-ZA-K-form} and \eqref{eq:B-ZB-K-form}.  The first term
is the pure $A$ fusion statistic:
\begin{equation}
    \sum_{m=0}^{N-1}K(R,X;m(l+x))
    \sim_{\loc}
    \sum_{m=0}^{N-1}K(R,X;mx)
    =
    (Z_3^{(A)},0).
    \label{eq:B-A-only-removal}
\end{equation}
Here we used Eq.~\eqref{eq:B-nested-locality-shift}: the support of $L$ has
empty triple intersection with $\supp(R)\cap\supp(X)$, so shifting by
multiples of $l$ is a locality identity.

Similarly, the fourth term is the pure $B$ fusion statistic.  The supports of
$\widetilde Y$ and $L$ meet at $j_-$, while $R$ and $O_A$ avoid this point.
Thus nested locality allows shifts by $r$ and $o_A$.  Using $x+r+o_A=0$,
\begin{equation}
    m(l+x)+r+x
    =
    ml-mr-(m+1)o_A ,
\end{equation}
so
\begin{equation}
    \sum_{m=0}^{N-1}K(\widetilde Y,L;m(l+x)+r+x)
    \sim_{\loc}
    \sum_{m=0}^{N-1}K(\widetilde Y,L;ml)
    =
    (Z_3^{(B)},0).
    \label{eq:B-B-only-removal}
\end{equation}
The second term is itself a locality identity, because
$\supp(R)\cap\supp(L)=\emptyset$:
\begin{equation}
    K(R,L;a) = 0 \pmod{2\pi},
    \qquad
    \text{for every }a .
    \label{eq:B-RL-locality}
\end{equation}
After subtracting the pure $A$ and $B$ pieces, the only surviving mixed term is
\begin{equation}
    (Z_3^{(AB)},0) - (Z_3^{(A)},0)- (Z_3^{(B)},0)
    \sim_{\loc}
    \sum_{m=0}^{N-1}
    K(\widetilde Y,X;m(l+x)+r).
    \label{eq:B-surviving-mixed-K}
\end{equation}
We now put this in the standard Bockstein staircase form.  The outside edges
$O_A$ and $O_B$ have empty triple intersection with
$\supp(X)\cap\supp(\widetilde Y)$, so we may shift by $o_A$ and $o_B$.
Using $o_A=-x-r$ and $o_B=-l-\widetilde y$,
\begin{equation}
\begin{aligned}
    K(\widetilde Y,X;m(l+x)+r)
    &\sim_{\loc}
    K(\widetilde Y,X;m(l+x)+r+mo_B+o_A) 
    =
    K(\widetilde Y,X;(m-1)x-m\widetilde y).
\end{aligned}
\label{eq:B-mixed-shift}
\end{equation}
Since $\widetilde y=-y$, this gives
\begin{equation}
    (Z_3^{(AB)},0) - (Z_3^{(A)},0)- (Z_3^{(B)},0)
    \sim_{\loc}
    \sum_{m=0}^{N-1}
    K(\widetilde Y,X;(m-1)x+my).
    \label{eq:B-mixed-diagonal}
\end{equation}
Finally, convert from $\widetilde Y=Y^{-1}$ to $Y$.  Since $Y$ shifts by $y$,
we have
\begin{equation}
    \theta_{\widetilde Y}(a)
    =
    -\theta_Y(a-y).
\end{equation}
Set
\begin{equation}
    b_m:=(m-1)(x+y).
\end{equation}
Then, $(m-1)x+my=b_m+y$.  Expanding
Eq.~\eqref{eq:B-mixed-diagonal} gives
\begin{equation}
\begin{aligned}
    (Z_3^{(AB)},0) - (Z_3^{(A)},0)- (Z_3^{(B)},0)
    \sim_{\loc}
    \sum_{m=0}^{N-1}
    \Big[
        \theta_X(b_m)
        +\theta_Y(b_m+x)
        -\theta_Y(b_m)
        -\theta_X(b_m+y)
    \Big].
\end{aligned}
\label{eq:B-mixed-as-W}
\end{equation}
After the cyclic reindexing $n=m-1$, the right-hand side is exactly the
expression for
\begin{equation}
    W_N(X,Y)=(Y^{-1}X^{-1})^N(YX)^N.
\end{equation}
Indeed, the staircase $(YX)^N$ contributes
\[
    \sum_{n=0}^{N-1}
    \left[
        \theta_X(n(x+y))
        +
        \theta_Y(n(x+y)+x)
    \right],
\]
while the inverse staircase contributes
\[
    -\sum_{n=0}^{N-1}
    \left[
        \theta_Y(n(x+y))
        +
        \theta_X(n(x+y)+y)
    \right].
\]
Therefore
\begin{equation}
    (Z_3^{(AB)},0) - (Z_3^{(A)},0)- (Z_3^{(B)},0)
    \sim_{\loc}
    (W_N(X,Y),0).
    \label{eq:B-epol-equals-W}
\end{equation}
Exponentiating gives
\begin{equation}
    \frac{Z_3^{(AB)}}{Z_3^{(A)}Z_3^{(B)}}
    =
    W_N(X,Y)
    =
    W_N\!\left(U^A_{ij_+},U^B_{kj_-}\right)
    .
    \label{eq:B-fusion-polarization-main}
\end{equation}
Equivalently, if one uses the left-to-right operator
\[
    \widetilde Y=U^B_{j_-k}=Y^{-1}
\]
as the second argument of the Bockstein process, then the same identity is
\begin{equation}
    \frac{Z_3^{(AB)}}{Z_3^{(A)}Z_3^{(B)}}
    =
    W_N\!\left(U^A_{ij_+},U^B_{j_-k}\right)^{-1}.
\end{equation}
Thus the apparent inverse is only an orientation convention for the
\(B\)-hopping operator.  With the main-text convention
\(Y=U^B_{kj_-}\), the mixed part of the fusion statistics is exactly the
Bockstein braiding statistic.

\section{Linearity of Bockstein braiding in the fusion label}
\label{app:WN-linearity}

In this appendix, we prove that Bockstein braiding is linear in the fusion
label of either excitation. It suffices to prove linearity in the first
argument:
\begin{equation}
    W_N(X^\alpha,Y)
    =
    W_N(X,Y)^\alpha,
    \qquad
    \alpha\in\ZN.
    \label{eq:C-WN-linearity-goal}
\end{equation}
The proof for the second argument follows by exchanging the two excitation
types. The case $\alpha=0$ is immediate. Henceforth, we choose the integer
representative $\alpha\in\{1,\ldots,N-1\}$.

Although we formulate the proof using the resolved one-dimensional geometry
of Fig.~\ref{fig:resolved-ab-hopping}, the reduction is local and applies in
every dimension $d=p+q+1$. Indeed, the $(p+1)$- and $(q+1)$-dimensional
supports of $X$ and $Y$ intersect cleanly along a one-dimensional locus. A
sufficiently small tubular neighborhood of this locus factorizes locally into
the overlap direction and transverse disks. The transverse directions are not relevant: after suppressing them, the restricted supports have precisely
the resolved-circle incidence pattern used below. Moreover, the proof depends
only on the $\ZN$ fusion labels and on locality relations asserting that
certain multiple intersections of supports are empty. These properties are
unchanged upon adjoining or removing the transverse spectator directions.
Consequently, the resolved-circle proof establishes
Eq.~\eqref{eq:C-WN-linearity-goal} for arbitrary $p$ and $q$.

We therefore work in this local normal form. Choose ordered points $i<j<k$
on the resolved circle, with $j_+$ and $j_-$ denoting the two point-split
copies of the overlap point, and define
\begin{equation}
    X=U^A_{ij_+},
    \quad
    Y=U^B_{kj_-},
    \quad
    R=U^A_{j_+k},
    \quad
    O_A=U^A_{ki},
    \quad
    L=U^B_{ij_-},
    \quad
    O_B=U^B_{ki}.
    \label{eq:C-A-operators}
\end{equation}
Set
\begin{equation}
    x:=\partial X,\qquad
    y:=\partial Y,\qquad
    r:=\partial R,\qquad
    o_A:=\partial O_A,\qquad
    l:=\partial L,\qquad
    o_B:=\partial O_B .
    \label{eq:C-boundaries}
\end{equation}
The circle closes as
\begin{equation}
    x+r+o_A=0,
    \quad
    l-y+o_B=0.
    \label{eq:C-circle-relations}
\end{equation}
For two operators $P,Q$, Eq.~\eqref{eq:B-K-density-def} defines
\begin{equation}
    K(P,Q;a)
    :=
    \theta_P(a)+\theta_Q(a+\partial P)
    -\theta_Q(a)-\theta_P(a+\partial Q).
    \label{eq:C-K-def}
\end{equation}
Therefore,
\begin{equation}
    (W_N(P,Q),0)
    =
    \sum_{m=0}^{N-1}
    K(P,Q;m(\partial P+\partial Q)).
    \label{eq:C-WN-K-form}
\end{equation}
For a configuration shift $a'$, write
\begin{equation}
    (T_{a'} K)(P,Q;a):=K(P,Q;a+a').
\end{equation}
We first prove the key locality identity used below:
\begin{equation}
    (T_x-1)(T_y-1)K(X,Y;a) = 0 \pmod{2\pi}.
    \label{eq:C-mixed-second-difference}
\end{equation}
\begin{proof}
The outside operators $O_A$ and $O_B$ have empty triple intersection with
$\supp(X)\cap\supp(Y)$.  Therefore nested locality gives
\begin{equation}
    (T_{o_A}-1)K(X,Y;a)= 0 \pmod{2\pi},
    \quad
    (T_{o_B}-1)K(X,Y;a)= 0 \pmod{2\pi} .
    \label{eq:C-outside-shifts}
\end{equation}
The same is true with $-o_A$ and $-o_B$, using the inverse operators.
Next, the following four
supports have empty total intersection:
\begin{equation}
    \supp(X)\cap\supp(Y)\cap\supp(R)\cap\supp(L)=\emptyset .
    \label{eq:C-fourfold-empty}
\end{equation}
Thus the nested commutator involving these four operators,
$[L,[R^{-1},[Y,X]]]$, is a locality identity.
In terms of the $K$ expression, this gives
\begin{equation}
    (T_l-1)(T_{-r}-1)K(X,Y;a)= 0 \pmod{2\pi} .
    \label{eq:C-LR-second-difference}
\end{equation}
Using Eq.~\eqref{eq:C-circle-relations},
\begin{equation}
    T_x-1
    =
    T_{-r-o_A}-1
    =
    T_{-r}(T_{-o_A}-1)+(T_{-r}-1),
    \label{eq:C-Tx-decomposition}
\end{equation}
and
\begin{equation}
    T_y-1
    =
    T_{l+o_B}-1
    =
    T_l(T_{o_B}-1)+(T_l-1).
    \label{eq:C-Ty-decomposition}
\end{equation}
Multiplying these two identities and applying them to $K(X,Y;a)$, all terms
containing $T_{-o_A}-1$ or $T_{o_B}-1$ become $0$ by
Eq.~\eqref{eq:C-outside-shifts}.  The only remaining term is
\[
    (T_{-r}-1)(T_l-1)K(X,Y;a),
\]
which is also $0$ by Eq.~\eqref{eq:C-LR-second-difference}.  This proves
Eq.~\eqref{eq:C-mixed-second-difference}.
\end{proof}

We now expand $W_N(X^\alpha,Y)$.  Let
\begin{equation}
    X_\alpha:=X^\alpha,
    \qquad
    \partial X_\alpha=\alpha x .
\end{equation}
We represent $X_\alpha$ by applying $\alpha$ copies of $X$ successively:
\begin{equation}
    \theta_{X_\alpha}(a)
    =
    \sum_{\lambda=0}^{\alpha-1}
    \theta_X(a+\lambda x).
    \label{eq:C-theta-Xalpha}
\end{equation}
Then
\begin{equation}
\begin{aligned}
    K(X_\alpha,Y;a)
    &=
    \theta_{X_\alpha}(a)
    +\theta_Y(a+\alpha x)
    -\theta_Y(a)
    -\theta_{X_\alpha}(a+y) \\
    &=
    \sum_{\lambda=0}^{\alpha-1}
    \Big[
        \theta_X(a+\lambda x)
        +\theta_Y(a+(\lambda+1)x)
        -\theta_Y(a+\lambda x)
        -\theta_X(a+y+\lambda x)
    \Big] \\
    &=
    \sum_{\lambda=0}^{\alpha-1}
    K(X,Y;a+\lambda x).
\end{aligned}
\label{eq:C-K-Xalpha-expansion}
\end{equation}
Therefore
\begin{equation}
\begin{aligned}
    (W_N(X^\alpha,Y),0)
    &=
    \sum_{m=0}^{N-1}
    K(X_\alpha,Y;m(\alpha x+y)) 
    =
    \sum_{m=0}^{N-1}
    \sum_{\lambda=0}^{\alpha-1}
    K(X,Y;m(\alpha x+y)+\lambda x).
\end{aligned}
\label{eq:C-WN-Xalpha-expanded}
\end{equation}
On the other hand,
\begin{equation}
    \alpha(W_N(X,Y),0)
    =
    \sum_{m=0}^{N-1}
    \sum_{\lambda=0}^{\alpha-1}
    K(X,Y;m(x+y)).
    \label{eq:C-alpha-WN-expanded}
\end{equation}
Thus
\begin{equation}
\begin{aligned}
    \Delta_\alpha
    &:=
    (W_N(X^\alpha,Y),0)-\alpha(W_N(X,Y),0) \\
    &=
    \sum_{m=0}^{N-1}
    \sum_{\lambda=0}^{\alpha-1}
    \left[
        K(X,Y;m(\alpha x+y)+\lambda x)
        -
        K(X,Y;m(x+y))
    \right].
\end{aligned}
\label{eq:C-Delta-alpha}
\end{equation}
It remains to show that $\Delta_\alpha$ is an integer combination of mixed
second differences.  Introduce formal translation variables
\begin{equation}
    \mathsf A:=T_x,
    \qquad
    \mathsf B:=T_y .
\end{equation}
Since $Nx=Ny=0$, we work in the group ring
\[
    \ZZ[\mathsf A,\mathsf B]/(\mathsf A^N-1,\mathsf B^N-1).
\]
Eq.~\eqref{eq:C-Delta-alpha} becomes
\begin{equation}
    \Delta_\alpha
    =
    F_{N,\alpha}(\mathsf A,\mathsf B)K(X,Y;0),
    \label{eq:C-Delta-polynomial-form}
\end{equation}
where
\begin{equation}
    F_{N,\alpha}(\mathsf A,\mathsf B)
    =
    \left(\sum_{m=0}^{N-1}(\mathsf A^\alpha\mathsf B)^m\right)
    \left(\sum_{\lambda=0}^{\alpha-1}\mathsf A^\lambda\right)
    -
    \alpha\sum_{m=0}^{N-1}(\mathsf A\mathsf B)^m .
    \label{eq:C-F-polynomial}
\end{equation}
We now prove that $F_{N,\alpha}$ is divisible by
$(\mathsf A-1)(\mathsf B-1)$.
It is enough to check that
\begin{equation}
    F_{N,\alpha}(1,\mathsf B)=0
    \quad\text{and}\quad
    F_{N,\alpha}(\mathsf A,1)=0~,
    \label{eq:F-specializations-zero}
\end{equation}
as identities in
\[
    \ZZ[\mathsf B]/(\mathsf B^N-1)
    \quad\text{and}\quad
    \ZZ[\mathsf A]/(\mathsf A^N-1),
\]
respectively.  Indeed, write
\begin{equation}
    F_{N,\alpha}(\mathsf A,\mathsf B)
    =
    \sum_{p,q\in\ZN}
    f_{p,q}\mathsf A^p\mathsf B^q .
\end{equation}
The first identity in Eq.~\eqref{eq:F-specializations-zero} says
\begin{equation}
    \sum_{p\in\ZN} f_{p,q}=0
    \qquad
    \text{for every } q\in\ZN ,
    \label{eq:F-column-sum-zero}
\end{equation}
and the second says
\begin{equation}
    \sum_{q\in\ZN} f_{p,q}=0
    \qquad
    \text{for every } p\in\ZN .
    \label{eq:F-row-sum-zero}
\end{equation}
These two conditions imply divisibility by
$(\mathsf A-1)(\mathsf B-1)$ on the finite torus.  Explicitly, define
\begin{equation}
    Q_{N,\alpha}(\mathsf A,\mathsf B)
    =
    \sum_{p,q\in\ZN}
    q_{p,q}\mathsf A^p\mathsf B^q,
    \qquad
    q_{p,q}
    :=
    \sum_{u=0}^{p}
    \sum_{v=0}^{q}
    f_{u,v},
    \qquad
    0\leq p,q\leq N-1 .
    \label{eq:Q-explicit-primitive}
\end{equation}
The zero row and column sums in Eqs.~\eqref{eq:F-column-sum-zero} and
\eqref{eq:F-row-sum-zero} make this definition compatible with the periodic
identifications $p\sim p+N$ and $q\sim q+N$.  A direct finite-difference
calculation then gives
\begin{equation}
    F_{N,\alpha}(\mathsf A,\mathsf B)
    =
    (\mathsf A-1)(\mathsf B-1)Q_{N,\alpha}(\mathsf A,\mathsf B)
\end{equation}
in the group ring
$\ZZ[\mathsf A,\mathsf B]/(\mathsf A^N-1,\mathsf B^N-1)$.
It remains to verify Eq.~\eqref{eq:F-specializations-zero}.  First,
\begin{equation}
\begin{aligned}
    F_{N,\alpha}(1,\mathsf B)
    &=
    \left(\sum_{m=0}^{N-1}\mathsf B^m\right)
    \left(\sum_{\lambda=0}^{\alpha-1}1\right)
    -
    \alpha\sum_{m=0}^{N-1}\mathsf B^m  
    =
    \alpha\sum_{m=0}^{N-1}\mathsf B^m
    -
    \alpha\sum_{m=0}^{N-1}\mathsf B^m
    =
    0 .
\end{aligned}
\end{equation}
Second,
\begin{equation}
    F_{N,\alpha}(\mathsf A,1)
    =
    \left(\sum_{m=0}^{N-1}\mathsf A^{\alpha m}\right)
    \left(\sum_{\lambda=0}^{\alpha-1}\mathsf A^\lambda\right)
    -
    \alpha\sum_{m=0}^{N-1}\mathsf A^m .
\end{equation}
We claim that the first product equals
$\alpha\sum_{p=0}^{N-1}\mathsf A^p$ in
$\ZZ[\mathsf A]/(\mathsf A^N-1)$.  Equivalently, for each $p\in\ZN$, the
number of pairs $(m,\lambda)$ satisfying
\begin{equation}
    \alpha m+\lambda\equiv p\pmod N,
    \qquad
    0\leq m\leq N-1,
    \qquad
    0\leq \lambda\leq \alpha-1,
\end{equation}
is exactly $\alpha$.
Let $g=\gcd(\alpha,N)$.  For fixed $\lambda$, the congruence is equivalent to
\[
    \alpha m\equiv p-\lambda\pmod N .
\]
This linear congruence has a solution only if $g\mid(p-\lambda)$,
equivalently only if
\[
    \lambda\equiv p\pmod g .
\]
Conversely, if $\lambda\equiv p\pmod g$, then $p-\lambda=g b$ for some
integer $b$.  Writing
\[
    \alpha=g\alpha',
    \qquad
    N=gN',
    \qquad
    \gcd(\alpha',N')=1,
\]
the congruence becomes
\[
    \alpha' m\equiv b\pmod {N'} .
\]
Since $\alpha'$ is invertible modulo $N'$, this reduced congruence has a
unique solution modulo $N'$.  Lifting this solution to a solution modulo $N$
gives exactly $g$ solutions,
\[
    m=m_0+tN',
    \qquad
    t=0,1,\ldots,g-1 .
\]
Therefore, for fixed $\lambda$, the congruence has $g$ solutions in $m$ if
$\lambda\equiv p\pmod g$, and has no solutions otherwise.

Among $\lambda=0,1,\ldots,\alpha-1$, exactly $\alpha/g$ values obey
$\lambda\equiv p\pmod g$, so the total number of solution pairs
$(m,\lambda)$ is
\[
    g\cdot\frac{\alpha}{g}=\alpha .
\]
Therefore
\begin{equation}
    \left(\sum_{m=0}^{N-1}\mathsf A^{\alpha m}\right)
    \left(\sum_{\lambda=0}^{\alpha-1}\mathsf A^\lambda\right)
    =
    \alpha\sum_{p=0}^{N-1}\mathsf A^p ,
\end{equation}
and so
\begin{equation}
    F_{N,\alpha}(\mathsf A,1)=0 .
\end{equation}
Thus $F_{N,\alpha}$ is divisible by
$(\mathsf A-1)(\mathsf B-1)$.

Consequently, we get
\begin{equation}
    \Delta_\alpha
    =
    Q_{N,\alpha}(T_x,T_y)(T_x-1)(T_y-1)K(X,Y;0).
\end{equation}
By Eq.~\eqref{eq:C-mixed-second-difference}, each translated copy of
$(T_x-1)(T_y-1)K(X,Y;0)$ is $0$, so
\begin{equation}
    \Delta_\alpha = 0 \pmod{2\pi}.
\end{equation}
Therefore
\begin{equation}
    (W_N(X^\alpha,Y),0)
    \sim_{\loc}
    \alpha(W_N(X,Y),0).
\end{equation}
Exponentiating the two sides gives
\begin{equation}
    W_N(X^\alpha,Y)=W_N(X,Y)^\alpha .
    \label{eq:C-WN-linearity-final}
\end{equation}
The proof for the second argument is the same after exchanging the two
species.  Hence
\begin{equation}
    W_N(X,Y^\beta)=W_N(X,Y)^\beta,
    \qquad
    \beta\in\ZN .
\end{equation}


\section{Continuum and lattice realizations of Bockstein braiding in the
$(2{+}1)$D toric code}
\label{app:toric-code-realization}

This appendix presents continuum and lattice descriptions of the same
toric-code example. We first derive the Bockstein response of the
$(2{+}1)$D $\mathbb Z_2$ gauge theory and then realize the corresponding
phase $W_2=-1$ microscopically by condensing the $m^2$ anyon of a
$\mathbb Z_4$ toric code. The lattice construction extends directly to
arbitrary $N$.

\subsection{Continuum BF-theory description}

The $\mathbb Z_2$ gauge theory has the BF action
\[
    S_{\mathrm{BF}}
    =
    \frac{2}{2\pi}\int_{\mathcal M_3} a\wedge db,
\]
where $a$ and $b$ are $U(1)$ 1-form gauge fields whose holonomies are
normalized to take values $0,\pi$ modulo $2\pi$. In the convention of
Fig.~\ref{fig:ordinary-and-bockstein-braiding}, the continuum
representatives of the $A$-type particle operator and the $B$-type loop
operator are, respectively,
\[
    \mathcal O_A(e)
    =
    \exp\!\left(i\int_e b\right),
    \qquad
    \mathcal O_B(f)
    =
    \exp\!\left(\pi i\int_f\frac{da}{2\pi}\right).
\]
Equivalently, after rescaling the fields so that their holonomies take
values in $\{0,1\}$ and denoting the rescaled fields by
$\widetilde a,\widetilde b$, these operators become
\begin{equation}
    \mathcal O_A(e)
    =
    (-1)^{\int_e\widetilde b},
    \qquad
    \mathcal O_B(f)
    =
    (-1)^{\int_f d\widetilde a/2}.
\end{equation}
The lattice construction below realizes this pair as $U_e^A$ and $U_f^B$,
respectively. By Stokes' theorem, the $B$-type operator may also be viewed
as a Wilson loop of the parent $\mathbb Z_4$ gauge theory on $\partial f$
before the even magnetic fluxes are condensed.

Let $A$ be the $\mathbb Z_2$ 2-form background for the $A$-type particle
and $B$ the $\mathbb Z_2$ 1-form background for the $B$-type loop.
Coupling the two operators to these backgrounds adds
\[
    \Delta S
    =
    \int_{\mathcal M_3}\frac{da}{2\pi}\wedge B
    -
    \frac{2}{2\pi}\int_{\mathcal M_3}b\wedge A .
\]
The equation of motion for $b$ imposes the twisted flux quantization
\[
    da=A\pmod{2\pi}.
\]
The coupling to $B$ is naturally expressed using a four-dimensional
extension $\mathcal M_4$ with
$\partial\mathcal M_4=\mathcal M_3$:
\begin{equation}
    \int_{\mathcal M_3}\frac{da}{2\pi}\wedge B
    =
    \int_{\mathcal M_4}\frac{da}{2\pi}\wedge dB
    =
    \int_{\mathcal M_4}\frac{A}{2\pi}\wedge dB .
\end{equation}
Writing $\overline A=A/\pi$ and $\overline B=B/\pi$ for the corresponding
$\mathbb Z_2$ cocycles, the Euclidean inflow action is
\[
    S_{\mathrm{inflow}}
    =
    \pi i\int_{\mathcal M_4}
    \overline A\cup\beta_2\overline B .
\]
This is precisely the Bockstein response in
Eq.~\eqref{eq:bockstein-response-intro}. The continuum theory therefore
predicts the nontrivial phase $W_2=-1$ derived microscopically below.

\subsection{Condensed-lattice realization for
\texorpdfstring{$N=2$}{N=2}}

Consider the $(2{+}1)$D $\mathbb Z_4$ toric
code~\cite{Kitaev2002Topologicalquantummemory} on an oriented lattice,
with one qudit on each oriented edge $\ell$. The generalized Pauli
operators satisfy
\begin{equation}
    X_\ell^4=Z_\ell^4=1,
    \qquad
    Z_\ell X_\ell=iX_\ell Z_\ell .
\end{equation}
We condense the $m^2$ anyon by imposing $X_\ell^2=1$ on every edge. A
commuting Hamiltonian for the condensed model is generated by the vertex
terms $A_v$, the edge terms $X_\ell^2$, and the plaquette terms $B_p^2$,
where
\begin{equation}
    A_v
    =
    \prod_{\ell\ni v}X_\ell^{\eta_{\ell v}},
    \qquad
    B_p
    =
    \prod_{\ell\in\partial p}Z_\ell^{\epsilon_{\ell p}},
    \qquad
    \eta_{\ell v},\epsilon_{\ell p}=\pm1 .
\end{equation}
Here $\eta_{\ell v}=+1$ when $\ell$ points away from $v$ and
$\eta_{\ell v}=-1$ when it points toward $v$, while
$\epsilon_{\ell p}$ is the incidence sign of $\ell$ in the oriented
boundary $\partial p$. The condensed model is equivalent to the standard
$(2{+}1)$D $\mathbb Z_2$ toric code.

We now consider the particle--loop geometry in
Fig.~\ref{fig:ordinary-and-bockstein-braiding}\subref{fig:bockstein-braiding}. For a dual path $e^\vee$, define the
$A$-type particle operator
\begin{equation}
    U_e^A
    :=
    \prod_{\ell\in e^\vee}X_\ell .
\end{equation}
This is the usual hopping operator for the $m$ anyon. Since $(U_e^A)^2$
is a product of the condensed edge terms $X_\ell^2$, it acts as the
identity in the condensed Hilbert space. The $A$-type particle therefore
obeys $\mathbb Z_2$ fusion.

For a face $f$, define the $B$-type loop operator
\begin{equation}
    U_f^B
    :=
    \prod_{\ell\in\partial f}Z_\ell^{\epsilon_{\ell f}},
    \qquad
    \epsilon_{\ell f}=\pm1,
\end{equation}
where $\epsilon_{\ell f}$ is the incidence sign of $\ell$ in
$\partial f$. This operator creates a loop excitation on $\partial f$.
Because $(U_f^B)^2=B_f^2$ is a Hamiltonian term, the loop also obeys
$\mathbb Z_2$ fusion.

Suppose that $e^\vee$ crosses $\partial f$ once, at the edge $\ell_0$.
All other edge factors commute and cancel in the Bockstein word. With the
Fig.~\ref{fig:ordinary-and-bockstein-braiding} convention
\[
    X=U_e^A,
    \qquad
    Y=U_f^B,
\]
the process is
\begin{equation}
    W_2(U_e^A,U_f^B)
    =
    \left((U_f^B)^{-1}(U_e^A)^{-1}\right)^2
    \left(U_f^B U_e^A\right)^2 .
\end{equation}
The crossed edge gives
\begin{equation}
    \begin{aligned}
    W_2(U_e^A,U_f^B)
    &=
    \left(
        Z_{\ell_0}^{-\epsilon_{\ell_0 f}}X_{\ell_0}^{-1}
    \right)^2
    \left(
        Z_{\ell_0}^{\epsilon_{\ell_0 f}}X_{\ell_0}
    \right)^2 
    =
    i^{2\epsilon_{\ell_0 f}}
    =
    -1 .
    \end{aligned}
\end{equation}
Thus the condensed $\mathbb Z_4$ toric code realizes nontrivial
particle--loop Bockstein braiding for $N=2$.

\subsection{Generalization to arbitrary \texorpdfstring{$N$}{N}}

The construction extends directly to arbitrary $N$. Let
\[
    \omega
    :=
    \exp\!\left(\frac{2\pi i}{N^2}\right),
\]
so that the generalized Pauli operators of the $\mathbb Z_{N^2}$ toric
code obey
\[
    Z_\ell X_\ell=\omega X_\ell Z_\ell .
\]
Condensing the $m^N$ anyon imposes $X_\ell^N=1$. Together with the vertex
terms $A_v$ and plaquette terms $B_p^N$, this produces the
$\mathbb Z_N$ toric code.

For an oriented dual path $e^\vee$ and an oriented face $f$, define
\begin{equation}
    U_e^A
    :=
    \prod_{\ell\in e^\vee}X_\ell^{s_{\ell e}},
    \qquad
    U_f^B
    :=
    \prod_{\ell\in\partial f}Z_\ell^{\epsilon_{\ell f}},
    \qquad
    s_{\ell e}=\pm1,
\end{equation}
where $s_{\ell e}$ records the relative orientation of the dual path and
the crossed edge. Then $(U_e^A)^N$ is a product of condensed edge terms,
while $(U_f^B)^N=B_f^N$. Both excitations therefore obey
$\mathbb Z_N$ fusion.

If the two supports cross once at $\ell_0$, define their oriented
intersection number by
\[
    I
    :=
    \epsilon_{\ell_0 f}s_{\ell_0e}
    =
    \pm1 .
\]
We choose the primal and dual orientations so that a positive unit
$A\cup\beta_N B$ intersection has $I=+1$. The only nontrivial
commutation relation is
\begin{equation}
    U_f^B U_e^A
    =
    \omega^I U_e^A U_f^B .
\end{equation}
It follows that
\begin{align}
    W_N(U_e^A,U_f^B)
    &=
    \left((U_e^A U_f^B)^N\right)^{-1}
    \left(U_f^B U_e^A\right)^N 
    =
    \omega^{NI}
    =
    \exp\!\left(\frac{2\pi i I}{N}\right).
\end{align}
Choosing $I=+1$ gives the primitive phase
\begin{equation}
    W_N(U_e^A,U_f^B)
    =
    \exp\!\left(\frac{2\pi i}{N}\right).
\end{equation}
Reversing the orientation of either support gives the complex-conjugate
primitive phase.


\section{Patch-operator diagnosis for the $\prod X$ and $\prod \mathrm{CZ}$ anomaly}
\label{app:patch-ZCZ-X-anomaly}

This appendix gives the detailed construction used in
Sec.~\ref{subsec:one-dimensional-patch-example}. We first derive the
one-dimensional anomalous symmetry action from the bulk cocycle.  We then
explain the relation between the simple $X$-$\mathrm{CZ}$ gauge and the
$X$-$Z\mathrm{CZ}$ gauge used for the patch computation.  Finally, we
construct finite symmetry patches and evaluate their Bockstein process.

\subsection{One-dimensional symmetry action from the bulk cocycle}

Put two $\ZZ_2$ variables $a_v,b_v$ on each vertex of a triangulated spatial
surface $M$.  We write $a,b\in C^0(M,\ZZ_2)$ and set $A=da$, $B=db$.  The
fixed-point bulk wavefunction for the cocycle $A\cup\beta_2 B$ is
\begin{equation}
    |\Psi\rangle
    \propto
    \sum_{a,b}
    (-1)^{\int_M a\cup\beta_2(db)}
    |a,b\rangle .
    \label{eq:E-bulk-wavefunction}
\end{equation}
For an ordered triangle $[ijk]$ with $i<j<k$,
\begin{equation}
    \bigl(\beta_2(db)\bigr)_{ijk}
    =
    (b_i+b_j)(b_j+b_k)
    \quad (\mathrm{mod}\;2) .
    \label{eq:E-beta-db-triangle}
\end{equation}
Thus the sign in Eq.~\eqref{eq:E-bulk-wavefunction} counts corners of the
$b$-domain-wall configuration, weighted by the $a$ variable at the first
vertex of the triangle.

On a one-dimensional edge, choose the primitive
\begin{equation}
    \lambda_b(j,j+1)=b_j(1-b_{j+1}),
    \qquad
    d\lambda_b=\beta_2(db).
\end{equation}
Under the global $a$ flip, the wavefunction changes by
\begin{equation}
    (-1)^{\sum_j b_j(1-b_{j+1})}
    =
    \prod_j Z^b_j\,\mathrm{CZ}^b_{j,j+1}.
\end{equation}
Therefore a natural one-dimensional representative of the two symmetries is
\begin{equation}
    U_A
    =
    \left(\prod_j X^a_j\right)
    \left(\prod_j Z^b_j\,\mathrm{CZ}^b_{j,j+1}\right),
    \qquad
    U_B=\prod_j X^b_j .
    \label{eq:E-ZCZ-symmetry}
\end{equation}
This is the $X$-$Z\mathrm{CZ}$ gauge.

For comparison, on an even closed chain the $X$-$Z\mathrm{CZ}$ gauge is
mapped to the more familiar $X$-$\mathrm{CZ}$ gauge by the finite-depth
circuit
\begin{equation}
    R=
    \prod_{j\;\mathrm{even}}
    \mathrm{CZ}_{a_j,b_j}\,
    \mathrm{CZ}_{a_j,b_{j+1}} .
    \label{eq:E-gauge-circuit}
\end{equation}
This circuit preserves $U_B$:
\begin{equation}
    R U_B R^{-1}=U_B .
\end{equation}
It maps Eq.~\eqref{eq:E-ZCZ-symmetry} to the $X$-$\mathrm{CZ}$ gauge:
\begin{equation}
    R U_A R^{-1}
    =
    U_A^{\mathrm{CZ}}
    :=
    \left(\prod_j X^a_j\right)
    \left(\prod_j \mathrm{CZ}^b_{j,j+1}\right).
    \label{eq:E-CZ-symmetry}
\end{equation}
Conversely, conjugating $U_A^{\mathrm{CZ}}$ by the same circuit returns
$U_A$.  Since the two gauges are related by a finite-depth change of basis,
they realize the same anomaly.  We use the $X$-$Z\mathrm{CZ}$ gauge below
because its finite patch operators have length-independent endpoint
decorations.

\subsection{Patch symmetry operators}

Fig.~\ref{fig:one-dimensional-patch-operators} summarizes the two finite
operators used below.
In the $X$-$Z\mathrm{CZ}$ gauge, write
\begin{equation}
    U_A=K_aF_b,
    \qquad
    K_a=\prod_j X^a_j,
    \qquad
    F_b=\prod_j Z^b_j\,\mathrm{CZ}^b_{j,j+1},
    \qquad
    U_B=\prod_j X^b_j .
\end{equation}
Equivalently, if $S^2=Z$, define the link operator
\begin{equation}
    W^b_j
    =
    S^b_j\,\mathrm{CZ}^b_{j,j+1}\,S^b_{j+1}.
    \label{eq:E-W-link}
\end{equation}
Then
\begin{equation}
    F_b=\prod_j W^b_j .
\end{equation}
Let $I=[x,x']$.  The $A$-type patch operator is
\begin{equation}
    X_I
    =
    \left(\prod_{j=x+1}^{x'}X^a_j\right)
    \left(\prod_{r=x}^{x'-1}
    S^b_r\,\mathrm{CZ}^b_{r,r+1}\,S^b_{r+1}
    \right).
    \label{eq:E-X-I-patch}
\end{equation}
In the middle of $I$, this is the restriction of the local symmetry density in
$U_A$; only the endpoints know that the interval is open.

Next, let $J=[y',y]$. We use the same convention as in the main text and
define
\begin{equation}
    Y_J
    =
    \mathrm{CZ}_{a_{y'+1},b_{y'}}\,
    \mathrm{CZ}_{a_{y'+1},b_{y'+1}}
    \left(\prod_{j=y'+1}^{y}X^b_j\right)
    \mathrm{CZ}_{a_{y+1},b_y}\,
    \mathrm{CZ}_{a_{y+1},b_{y+1}} .
    \label{eq:E-Y-J-patch-main-convention}
\end{equation}
The ordering in Eq.~\eqref{eq:E-Y-J-patch-main-convention} is part of the
definition.  We use the convention that the rightmost factor acts first.  Thus
the two controlled-$Z$ gates at the right endpoint act before the $X^b$ string,
while the two controlled-$Z$ gates at the left endpoint act after it.  In the
middle of $J$, the operator is simply the restriction of the onsite symmetry
density $X^b_j$; the controlled-$Z$ gates are endpoint decorations.

For the algebraic checks it is useful to normal-order the endpoint decorations
to the right of the $X^b$ string.  Since
$\mathrm{CZ}_{a_{y'+1},b_{y'+1}}$ does not commute with $X^b_{y'+1}$, this
produces one extra factor of $Z^a_{y'+1}$.  Thus
\begin{equation}
    Y_J=P_JE_J,
    \qquad
    P_J=\prod_{j=y'+1}^{y}X^b_j,
    \label{eq:E-Y-J-normal-ordered}
\end{equation}
where
\begin{equation}
    E_J
    =
    Z^a_{y'+1}\,
    \mathrm{CZ}_{a_{y'+1},b_{y'}}\,
    \mathrm{CZ}_{a_{y'+1},b_{y'+1}}\,
    \mathrm{CZ}_{a_{y+1},b_y}\,
    \mathrm{CZ}_{a_{y+1},b_{y+1}} .
    \label{eq:E-E-J}
\end{equation}
Eqs.~\eqref{eq:E-Y-J-patch-main-convention}--\eqref{eq:E-E-J}
give equivalent ordered and normal-ordered expressions for the same operator.
The endpoint placement is chosen so that adjacent patches glue without an
additional junction operator:
\begin{equation}
    X_{[u,v]}X_{[v,w]}=X_{[u,w]},
    \qquad
    Y_{[u,v]}Y_{[v,w]}=Y_{[u,w]}.
    \label{eq:E-patch-gluing}
\end{equation}

\subsection{Commutation with the global symmetries}

We now check that the patch operators commute with the full symmetry.  The
operator $X_I$ commutes with $U_A$ manifestly.  It also commutes with $U_B$,
because each link factor is invariant under flipping the two adjacent $b$
qubits:
\begin{equation}
    X^b_rX^b_{r+1}
    \left(S^b_r\,\mathrm{CZ}^b_{r,r+1}\,S^b_{r+1}\right)
    X^b_rX^b_{r+1}
    =
    S^b_r\,\mathrm{CZ}^b_{r,r+1}\,S^b_{r+1}.
    \label{eq:E-link-invariant}
\end{equation}
Therefore
\begin{equation}
    [X_I,U_A]=[X_I,U_B]=1 .
    \label{eq:E-XI-commutes}
\end{equation}

For $Y_J$, define
\begin{equation}
    R_J=Z^b_{y'}Z^b_{y'+1}Z^b_yZ^b_{y+1}.
    \label{eq:E-R-J}
\end{equation}
Conjugating $P_J$ by the $b$-part of $U_A$ gives
\begin{equation}
    F_bP_JF_b^{-1}
    =
    -R_JP_J .
    \label{eq:E-PJ-conj}
\end{equation}
The minus sign is the contribution from the onsite $\prod_j Z^b_j$ factor in
$F_b$, together with the controlled-$Z$ chain.  On the other hand, conjugating
the endpoint decoration by the onsite $a$ symmetry gives
\begin{equation}
    K_aE_JK_a^{-1}
    =
    -R_JE_J .
    \label{eq:E-EJ-conj}
\end{equation}
The two factors cancel in the product $Y_J=P_JE_J$.  More explicitly, the two
copies of $R_J$ can be passed through $P_J$ without a sign because $R_J$
contains two $Z^b$ operators on sites flipped by $P_J$, namely $y'+1$ and
$y$.
Hence
\begin{equation}
    [Y_J,U_A]=1 .
\end{equation}
The operator $Y_J$ also commutes with $U_B$: under conjugation by
$\prod_i X^b_i$, the four endpoint controlled-$Z$ gates in $E_J$ contribute
two factors of $Z^a_{y'+1}$ and two factors of $Z^a_{y+1}$.  Thus
\begin{equation}
    [Y_J,U_A]=[Y_J,U_B]=1 .
    \label{eq:E-YJ-commutes}
\end{equation}

The squares of the two patches are
\begin{equation}
    X_I^2=Z^b_xZ^b_{x'},
    \qquad
    Y_J^2=Z^a_{y'+1}Z^a_{y+1} .
    \label{eq:E-patch-fusion}
\end{equation}
Thus the twofold fusion of either patch creates endpoint charges of the other
symmetry.

\subsection{Evaluation of the Bockstein word}

Choose staggered intervals
\begin{equation}
    I=[x,x'],
    \qquad
    J=[y',y],
    \qquad
    x<y'<x'<y .
    \label{eq:E-staggered-intervals}
\end{equation}
We evaluate
\begin{equation}
    W_2(X_I,Y_J)=
    \left(Y_J^{-1}X_I^{-1}\right)^2
    \left(Y_JX_I\right)^2 .
    \label{eq:E-W2-def}
\end{equation}

Let $\alpha$ be the characteristic function of the sites $x<j\le x'$ acted
on by $X_I$, and let $\beta$ be the characteristic function of the sites
$y'<j\le y$ acted on by $P_J$:
\begin{equation}
    \alpha_j=
    \begin{cases}
        1, & x<j\le x',\\
        0, & \mathrm{otherwise},
    \end{cases}
    \qquad
    \beta_j=
    \begin{cases}
        1, & y'<j\le y,\\
        0, & \mathrm{otherwise} .
    \end{cases}
    \label{eq:E-characteristic-functions}
\end{equation}
On a computational basis state $|a,b\rangle$, the two patches act as
\begin{equation}
    X_I|a,b\rangle
    =
    i^{f_I(b)}
    |a+\alpha,b\rangle,
    \qquad
    Y_J|a,b\rangle
    =
    (-1)^{h_J(a,b)}
    |a,b+\beta\rangle,
    \label{eq:E-patch-actions}
\end{equation}
where
\begin{equation}
    f_I(b)
    =
    2\sum_{r=x}^{x'-1}b_rb_{r+1}
    +
    \sum_{r=x}^{x'-1}(b_r+b_{r+1})
    \quad (\mathrm{mod}\;4),
    \label{eq:E-fI}
\end{equation}
and
\begin{equation}
    h_J(a,b)
    =
    a_{y'+1}
    +a_{y'+1}(b_{y'}+b_{y'+1})
    +a_{y+1}(b_y+b_{y+1})
    \quad (\mathrm{mod}\;2).
    \label{eq:E-hJ}
\end{equation}
The phase $f_I$ comes from the product of the link operators $W^b_r$ in
$X_I$, while $h_J$ comes from the endpoint decoration $E_J$.  The term
$a_{y'+1}$ in $h_J$ is precisely the normal-ordering factor produced by
moving the left endpoint controlled-$Z$ gate through $X^b_{y'+1}$.

In the two staircases in Eq.~\eqref{eq:E-W2-def}, the $f_I$ phases cancel.
The remaining sign is the mixed second difference of $h_J$:
\begin{equation}
    W_2(X_I,Y_J)|a,b\rangle
    =
    (-1)^\Omega|a,b\rangle,
\end{equation}
where
\begin{equation}
\begin{aligned}
    \Omega
    =
    &h_J(a+\alpha,b)
    +h_J(a,b+\beta)
    +h_J(a+\alpha,b+\beta)
    +h_J(a,b)
    \quad (\mathrm{mod}\;2).
\end{aligned}
    \label{eq:E-Omega-def}
\end{equation}
Using Eq.~\eqref{eq:E-hJ}, this reduces to
\begin{equation}
    \Omega
    =
    \alpha_{y'+1}(\beta_{y'}+\beta_{y'+1})
    +
    \alpha_{y+1}(\beta_y+\beta_{y+1})
    \quad (\mathrm{mod}\;2).
    \label{eq:E-Omega-reduced}
\end{equation}
For the staggered ordering $x<y'<x'<y$,
\begin{equation}
    \alpha_{y'+1}=1,
    \qquad
    \beta_{y'}=0,
    \qquad
    \beta_{y'+1}=1,
    \qquad
    \alpha_{y+1}=0 .
\end{equation}
Therefore $\Omega=1$, and
\begin{equation}
    W_2(X_I,Y_J)=-1 .
    \label{eq:E-W2-minus-one}
\end{equation}
This nontrivial Bockstein statistic is the patch-operator diagnosis of the
mixed anomaly between $U_B=\prod_jX_j^b$ and $U_A^{\mathrm{CZ}}=(\prod_jX_j^a)(\prod_j\mathrm{CZ}_{j,j+1}^b)$.

\section{Patch symmetry operators for an anomalous $\ZZ_2^{(0)}\times\ZZ_2^{(1)}$ symmetry in $(2{+}1)$D}
\label{app:square-lattice-2d}

This appendix gives a square-lattice realization of particle-loop Bockstein
braiding in $(2{+}1)$ dimensions.
The model has an anomalous $\ZZ_2^{(0)}\times\ZZ_2^{(1)}$ symmetry obtained from the $(3{+}1)$D bulk cocycle
\begin{equation}
    \frac12 A_1\cup\bigl(B_2\cup_1 B_2\bigr).
    \label{eq:square-bulk-cocycle}
\end{equation}
We construct the two total symmetry operators, their finite disk and line
patches, and evaluate the Bockstein word directly.  The endpoint conventions
are chosen so that adjacent patches multiply exactly into the patch on the
union.

Eq.~\eqref{eq:triangular-total-symmetries-main}
gives a concise
triangular-lattice representative with one qubit on each edge.  For the
symmetry-preserving patch construction in this appendix, we instead need both
species on the square lattice, in direct analogy with
Appendix~\ref{app:patch-ZCZ-X-anomaly}: an $a$ qubit is placed on every face
and a $b$ qubit on every edge.  The face qubits provide the onsite factor
$K_a=\prod_fX_f^a$ and the endpoint decorations needed to open and glue the
1-form symmetry operators.  This auxiliary onsite sector does not change the
anomaly class; the non-onsite part of the 0-form symmetry and the 1-form
symmetry are carried by the $b$ qubits.  The square/cubical representative
used below contains $\prod_eZ_e^b$ in addition to the controlled-$Z$ circuit.
For a general branched triangulation, the standard simplicial representative
has a further branching-dependent Stiefel--Whitney chain factor.  We derive
that factor at the end of this appendix and show that, on an oriented surface,
it can be removed by a finite-depth change of basis using the face $a$ qubits.

\subsection{Boundary symmetry and a local square-lattice representative}

Let $M$ be an oriented 3-dimensional spatial manifold with boundary
$\Sigma=\partial M$.  Locally, write the bulk fields as
$A_1=\dd a$ and $B_2=\dd b$, where
$a\in C^0(M,\ZZ_2)$ and $b\in C^1(M,\ZZ_2)$.  A fixed-point wavefunction for
the bulk cocycle~\eqref{eq:square-bulk-cocycle} is
\begin{equation}
    |\Psi\rangle
    \propto
    \sum_{a,b}
    (-1)^{\int_M a\cup (\dd b \cup_1 \dd b)}
    |a,b\rangle .
\end{equation}
To obtain the boundary action, we choose a cochain representative of the
descent.  This choice is not unique at the cochain level.  On a triangulation
with a total vertex order $\prec$, the standard simplicial representative is
\begin{equation}
    \omega_{2,\prec}^{\Delta}(b)
    :=
    b\cup b+b\cup_1\dd b,
    \label{eq:square-omega2}
\end{equation}
and it satisfies the descent equation
\begin{equation}
    \dd\omega_{2,\prec}^{\Delta}(b)
    =
    \dd b\cup_1\dd b.
    \label{eq:square-descent}
\end{equation}
Equivalently, if $\widetilde b$ denotes the canonical integer lift of $b$ and
$\widetilde{\dd b}$ the canonical lift of the $\ZZ_2$-valued cochain
$\dd b$, then~\cite{CH21}
\begin{equation}
    \omega_{2,\prec}^{\Delta}(b)
    =
    \frac{\dd\widetilde b-\widetilde{\dd b}}{2}
    \quad (\mathrm{mod}\;2).
    \label{eq:square-integer-lift-identity}
\end{equation}
The numerator is even on every face, so the quotient is an integer cochain.

For any chosen representative $\omega_2$ of this descent, the boundary
$\ZZ_2$ 0-form symmetry acts as
\begin{equation}
    U_A\ket{a,b}
    =
    (-1)^{\int_\Sigma\omega_2(b)}\ket{a+1,b}.
    \label{eq:square-UA-cochain}
\end{equation}
The boundary also carries the $\ZZ_2$ 1-form symmetry
\begin{equation}
    U_B(\lambda)\ket{a,b}
    =
    \ket{a,b+\lambda},
    \qquad
    \dd\lambda=0.
    \label{eq:square-UB-cochain}
\end{equation}
Here $\lambda$ is a closed one-cochain; its support is Poincar\'e dual to a
closed line on the dual lattice.

We now specialize to a coherently oriented square cellulation of $\Sigma$.
Place one $a$ qubit at the center of each face $f$ and one $b$ qubit on each
edge $e$.  For a computational-basis configuration $b$, the $\ZZ_2$ flux
through $f$ is
\begin{equation}
    c_f:=(\dd b)(f)=\sum_{e\subset\partial f}b_e
    \quad (\mathrm{mod}\;2),
    \qquad
    \mathcal B_f:=\prod_{e\subset\partial f}Z_e^b=(-1)^{c_f}.
    \label{eq:square-flux}
\end{equation}
Let $S=\operatorname{diag}(1,i)$.  Associate to each face the diagonal circuit
\begin{equation}
    \mathcal S_f
    :=
    \left(\prod_{e\subset\partial f}S_e^b\right)
    \left(\prod_{\{e,e'\}\subset\partial f}\mathrm{CZ}_{e,e'}^b\right),
    \label{eq:square-face-circuit}
\end{equation}
where the second product is over the six unordered pairs of boundary edges.
To see its action, label the four boundary-edge bits by
$b_1,\ldots,b_4$.  If $\widetilde c_f\in\{0,1\}$ denotes the canonical
integer lift of $c_f$, then
\begin{equation}
    i^{\widetilde c_f}
    =
    i^{b_1\oplus b_2\oplus b_3\oplus b_4}
    =
    i^{\sum_jb_j}(-1)^{\sum_{j<k}b_jb_k},
\end{equation}
which is precisely the phase produced by the four $S$ gates and six
controlled-$Z$ gates in Eq.~\eqref{eq:square-face-circuit}.  Therefore
\begin{equation}
    \mathcal S_f\ket{b}=i^{\widetilde c_f}\ket{b},
    \qquad
    \mathcal S_f^2=\mathcal B_f.
    \label{eq:square-face-circuit-action}
\end{equation}

We next identify this circuit phase with a cubical descent representative.
Let
\begin{equation}
    n_f:=(\dd\widetilde b)(f)\in\ZZ.
\end{equation}
Since $n_f$ reduces modulo two to $c_f$, the difference
$n_f-\widetilde c_f$ is even.  We therefore define
\begin{equation}
    \omega_{2,\square}(b)(f)
    :=
    \frac{n_f-\widetilde c_f}{2}
    \quad (\mathrm{mod}\;2).
    \label{eq:square-cubical-omega}
\end{equation}
This is the cubical version of the integer-lift prescription in
Eq.~\eqref{eq:square-integer-lift-identity}~\cite{CH21}.
It is a separate cochain-level representative from the ordered-simplicial density
$\omega_{2,\prec}^{\Delta}$~\cite{Chen2023Highercup}; their exact circuit relation is derived in
Sec.~\ref{subsec:square-vs-triangular-gauge}.  For the square-lattice model,
we use $\omega_{2,\square}$ in the boundary action
Eq.~\eqref{eq:square-UA-cochain}.

By construction, each face obeys
\begin{equation}
    i^{\widetilde c_f}
    =
    (-1)^{\omega_{2,\square}(b)(f)}i^{n_f}.
\end{equation}
Because the faces are coherently oriented and $\Sigma$ is closed,
\begin{equation}
    \sum_f n_f
    =
    \int_\Sigma\dd\widetilde b
    =0.
\end{equation}
Multiplying the face identities, therefore, gives
\begin{equation}
    F_b:=\prod_f\mathcal S_f,
    \qquad
    F_b\ket b
    =
    (-1)^{\int_\Sigma\omega_{2,\square}(b)}\ket b.
    \label{eq:square-face-product-cocycle}
\end{equation}

It is also useful to record the gate content of $F_b$.  Every edge belongs to
two square faces, so the two $S_e^b$ factors on that edge combine to $Z_e^b$:
\begin{equation}
    F_b
    =
    \left(\prod_e Z_e^b\right)
    \left[
    \prod_f
    \prod_{\{e,e'\}\subset\partial f}
    \mathrm{CZ}_{e,e'}^b
    \right].
    \label{eq:square-Fb-gate-decomposition}
\end{equation}

The two total symmetry operators are now
\begin{equation}
    U_A=K_aF_b,
    \qquad
    K_a:=\prod_fX_f^a,
    \label{eq:square-total-UA}
\end{equation}
\begin{equation}
    U_B(\lambda)=\prod_e\bigl(X_e^b\bigr)^{\lambda(e)},
    \qquad
    \dd\lambda=0.
    \label{eq:square-total-UB}
\end{equation}
Because $\dd\lambda=0$, the transformation $b\mapsto b+\lambda$ leaves each
face flux $c_f$, and hence each $\mathcal S_f$, unchanged.  It follows that
$[U_A,U_B(\lambda)]=1$.  Finally, using
$\mathcal S_f^2=\mathcal B_f$,
\begin{equation}
    U_A^2
    =
    F_b^2
    =
    \prod_f\mathcal B_f
    =1
\end{equation}
because every edge occurs in the boundaries of exactly two faces.
The local face circuit and the two global symmetry operators are depicted in
Fig.~\ref{fig:square-total-symmetry}.

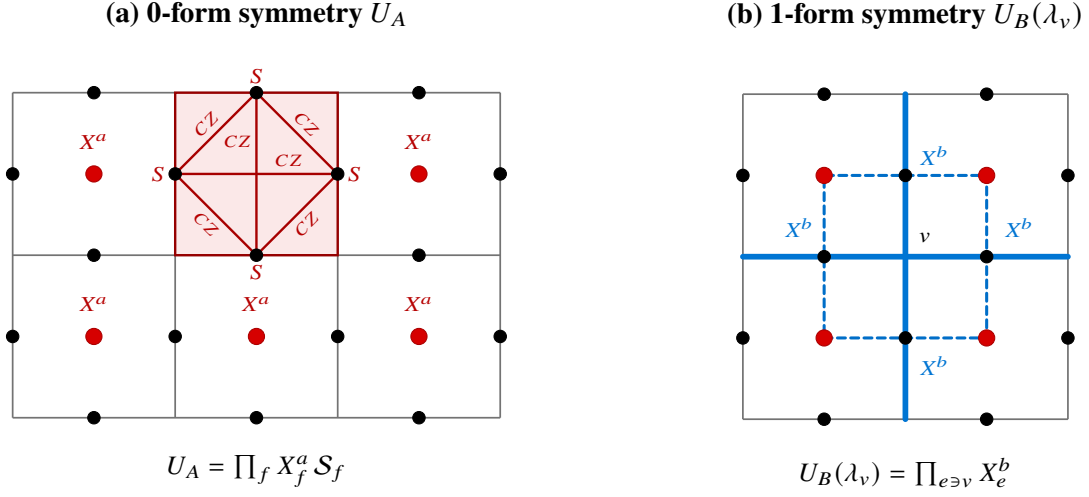
\begin{figure}[t]
    \centering

    \begin{minipage}[t]{0.56\linewidth}
    \centering
    \begin{tikzpicture}[
        x=2.15cm,
        y=2.15cm,
        lattice/.style={
            black!55,
            line width=0.65pt
        },
        facefill/.style={
            fill=braidred!9,
            draw=braidred!70!black,
            line width=0.9pt
        },
        facecircuit/.style={
            braidred!78!black,
            line width=0.9pt,
            line cap=round
        },
        asite/.style={
            circle,
            draw=braidred!80!black,
            fill=braidred,
            inner sep=2.15pt
        },
        bsite/.style={
            circle,
            draw=black,
            fill=black,
            inner sep=1.65pt
        },
        xalabel/.style={
            font=\scriptsize\bfseries,
            text=braidred!85!black,
            inner sep=0.7pt
        },
        slabel/.style={
            font=\scriptsize\bfseries,
            text=braidred!85!black,
            inner sep=0.4pt
        },
        czlabel/.style={
            font=\tiny\bfseries,
            text=braidred!85!black,
            inner sep=0.45pt
        },
        paneltitle/.style={
            font=\bfseries
        },
        formula/.style={
            font=\small,
            inner sep=1pt
        }
    ]

        \node[paneltitle] at (1.5,2.48)
            {(a) 0-form symmetry $U_A$};

        \foreach \i in {0,...,3} {
            \draw[lattice] (\i,0) -- (\i,2);
        }
        \foreach \j in {0,...,2} {
            \draw[lattice] (0,\j) -- (3,\j);
        }

        \fill[facefill] (1,1) rectangle (2,2);

        %

        \draw[facecircuit]
            (1.5,2)
            --
            node[czlabel,sloped,above=2pt,pos=0.5] {$CZ$}
            (1,1.5);

        \draw[facecircuit]
            (1.5,2)
            --
            node[czlabel,sloped,above=2pt,pos=0.5] {$CZ$}
            (2,1.5);

        \draw[facecircuit]
            (1,1.5)
            --
            node[czlabel,sloped,below=2pt,pos=0.5] {$CZ$}
            (1.5,1);

        \draw[facecircuit]
            (2,1.5)
            --
            node[czlabel,sloped,below=2pt,pos=0.5] {$CZ$}
            (1.5,1);

        \draw[facecircuit]
            (1.5,2)
            --
            node[czlabel,left=1pt,pos=0.3] {$CZ$}
            (1.5,1);

        \draw[facecircuit]
            (1,1.5)
            --
            node[czlabel,above=2pt,pos=0.7] {$CZ$}
            (2,1.5);

        \node[slabel] at (1.5,2.105) {$S$};
        \node[slabel] at (1.5,0.895) {$S$};
        \node[slabel] at (0.895,1.5) {$S$};
        \node[slabel] at (2.105,1.5) {$S$};

        \foreach \i in {0,...,3} {
            \foreach \j in {0,...,1} {
                \node[bsite] at ({\i},{\j+0.5}) {};
            }
        }

        \foreach \i in {0,...,2} {
            \foreach \j in {0,...,2} {
                \node[bsite] at ({\i+0.5},{\j}) {};
            }
        }

        \foreach \x/\y in {
            0.5/0.5,
            1.5/0.5,
            2.5/0.5,
            0.5/1.5,
            2.5/1.5
        } {
            \node[asite] at (\x,\y) {};
        }

        \foreach \x/\y in {
            0.5/0.5,
            1.5/0.5,
            2.5/0.5,
            0.5/1.5,
            2.5/1.5
        } {
            \node[xalabel] at (\x,{\y+0.20}) {$X^a$};
        }

        \node[formula] at (1.5,-0.33)
            {$U_A=\prod_f X_f^a\,\mathcal S_f$};

    \end{tikzpicture}
    \end{minipage}%
    \hfill%
    \begin{minipage}[t]{0.40\linewidth}
    \centering
    \begin{tikzpicture}[
        x=2.15cm,
        y=2.15cm,
        lattice/.style={
            black!55,
            line width=0.65pt
        },
        dualloop/.style={
            braidblue!85!black,
            densely dashed,
            line width=1.15pt,
            line cap=round
        },
        xbsupport/.style={
            braidblue!90!black,
            line width=2.15pt,
            line cap=round
        },
        asite/.style={
            circle,
            draw=braidred!80!black,
            fill=braidred,
            inner sep=2.15pt
        },
        bsite/.style={
            circle,
            draw=black,
            fill=black,
            inner sep=1.65pt
        },
        xblabel/.style={
            font=\scriptsize\bfseries,
            text=braidblue!90!black,
            inner sep=0.6pt
        },
        paneltitle/.style={
            font=\bfseries
        },
        formula/.style={
            font=\small,
            inner sep=1pt
        }
    ]

        \node[paneltitle] at (1,2.48)
            {(b) 1-form symmetry $U_B(\lambda_v)$};

        \foreach \i in {0,...,2} {
            \draw[lattice] (\i,0) -- (\i,2);
        }
        \foreach \j in {0,...,2} {
            \draw[lattice] (0,\j) -- (2,\j);
        }

        \draw[dualloop]
            (0.5,0.5)
            --
            (1.5,0.5)
            --
            (1.5,1.5)
            --
            (0.5,1.5)
            --
            cycle;


        \draw[xbsupport] (0,1) -- (1,1);
        \draw[xbsupport] (1,1) -- (2,1);
        \draw[xbsupport] (1,0) -- (1,1);
        \draw[xbsupport] (1,1) -- (1,2);

        \foreach \i in {0,...,2} {
            \foreach \j in {0,...,1} {
                \node[bsite] at ({\i},{\j+0.5}) {};
            }
        }

        \foreach \i in {0,...,1} {
            \foreach \j in {0,...,2} {
                \node[bsite] at ({\i+0.5},{\j}) {};
            }
        }

        \foreach \x/\y in {
            0.5/0.5,
            1.5/0.5,
            0.5/1.5,
            1.5/1.5
        } {
            \node[asite] at (\x,\y) {};
        }

        \node[xblabel,anchor=east]
            at (0.45,1.18) {$X^b$};

        \node[xblabel,anchor=west]
            at (1.6,1.18) {$X^b$};

        \node[xblabel,anchor=north]
            at (1.18,0.4) {$X^b$};

        \node[xblabel,anchor=south]
            at (1.18,1.55) {$X^b$};

        \node[font=\scriptsize] at (1.12,1.12)
            {$v$};

        \node[formula] at (1,-0.33)
            {$U_B(\lambda_v)=\prod_{e\ni v}X_e^b$};


    \end{tikzpicture}
    \end{minipage}
    \caption{
        Square-lattice realization of the anomalous
        $\mathbb Z_2^{(0)}\times\mathbb Z_2^{(1)}$ symmetry.
        An $a$ qubit resides at the center of each face and is represented by a red
        dot, while a $b$ qubit resides on each direct-lattice edge and is represented
        by a black dot.
        \textbf{(a)} The 0-form generator
        $U_A=\prod_f X_f^a\,\mathcal S_f$ applies $X_f^a$ to every face qubit and the
        diagonal circuit
        $\mathcal S_f=
        \bigl(\prod_{e\subset\partial f}S_e^b\bigr)
        \bigl(\prod_{\{e,e'\}\subset\partial f}\mathrm{CZ}_{e,e'}^b\bigr)$ to the four edge qubits surrounding each face.
        The highlighted face displays its four $S$ gates and the six controlled-$Z$
        gates associated with the unordered pairs of distinct boundary edges.
        The central $a$ qubit of the highlighted face is omitted from the drawing for
        visual clarity, although the corresponding $X_f^a$ factor remains part of
        $U_A$.
        On a closed square cellulation, every edge belongs to two faces, and hence
        $\prod_f\mathcal S_f=
        \bigl(\prod_e Z_e^b\bigr)
        \bigl(\prod_f\prod_{\{e,e'\}\subset\partial f}
        \mathrm{CZ}_{e,e'}^b\bigr)$.
        Its relation to the ordered-simplicial representative on an arbitrary
        branched triangulation is given in
        Sec.~\ref{subsec:square-vs-triangular-gauge}.
        \textbf{(b)} The smallest closed dual-lattice loop
        $\lambda_v^\vee=\partial v^\vee$ encircles a single direct-lattice vertex
        $v$. It crosses exactly the four direct-lattice edges incident on $v$, so the
        corresponding 1-form symmetry generator is
        $U_B(\lambda_v)=\prod_{e\ni v}X_e^b$.
        The dashed blue square denotes the dual-lattice loop, while the four thick
        blue segments denote the crossed direct-lattice edges supporting the
        $X_e^b$ operators.
    }
    \label{fig:square-total-symmetry}
\end{figure}

\providecolor{braidred}{RGB}{180,25,25}
\providecolor{braidblue}{RGB}{0,105,180}
\providecolor{braidgreen}{RGB}{0,145,85}

\begin{figure}[t]
    \centering
    \begin{tikzpicture}[
        x=1.45cm,
        y=1.45cm,
        lattice/.style={
            black!55,
            line width=0.7pt
        },
        patchred/.style={
            draw=braidred!75!black,
            fill=braidred!10,
            line width=1.0pt
        },
        xbsupport/.style={
            braidblue!90!black,
            line width=2.4pt,
            line cap=round
        },
        czgreen/.style={
            braidgreen!85!black,
            line width=1.15pt,
            line cap=round
        },
        asite/.style={
            circle,
            draw=braidred!80!black,
            fill=braidred,
            inner sep=2.1pt
        },
        bsite/.style={
            circle,
            draw=black,
            fill=black,
            inner sep=1.8pt
        },
        xalabel/.style={
            font=\normalsize\bfseries,
            text=braidred!85!black
        },
        xblabel/.style={
            font=\small\bfseries,
            text=braidblue!90!black
        },
        czgreenlabel/.style={
            font=\scriptsize\bfseries,
            text=braidgreen!85!black,
            inner sep=0.3pt
        },
        cstar/.style={
            font=\normalsize\bfseries,
            text=braidgreen!85!black
        },
        jlabel/.style={
            font=\normalsize\itshape,
            text=braidblue!90!black
        },
        paneltitle/.style={
            font=\bfseries
        },
        formula/.style={
            font=\small,
            inner sep=1pt
        }
    ]

    \begin{scope}[shift={(0,0)}]
        \node[paneltitle,anchor=west]
        at (0.15,3.95)
        {(a) Disk patch $X_I$};

        \foreach \i in {0,...,3} {
            \draw[lattice] (\i,0) -- (\i,3);
        }
        \foreach \j in {0,...,3} {
            \draw[lattice] (0,\j) -- (3,\j);
        }

        \fill[patchred] (1,1) rectangle (2,2);
        \draw[
            braidred!75!black,
            line width=1.0pt
        ]
        (1,1) rectangle (2,2);

        \node[
            text=braidred!85!black,
            font=\normalsize\itshape
        ]
        at (2.27,1.80)
        {$I$};

        \node[xalabel] at (1.5,1.72) {$X^a$};

        \node[
            text=braidred!85!black,
            font=\normalsize
        ]
        at (1.7,1.30)
        {$\mathcal S_f$};

        \foreach \i in {0,...,3} {
            \foreach \j in {0,...,2} {
                \node[bsite] at (\i,\j+0.5) {};
            }
        }

        \foreach \i in {0,...,2} {
            \foreach \j in {0,...,3} {
                \node[bsite] at (\i+0.5,\j) {};
            }
        }

        \foreach \i in {0,...,2} {
            \foreach \j in {0,...,2} {
                \node[asite] at (\i+0.5,\j+0.5) {};
            }
        }

        \node[formula]
        at (1.5,-0.48)
        {$X_I=\prod_{f\subset I}X_f^a\,\mathcal S_f$};
    \end{scope}

    \begin{scope}[shift={(5.35,0)}]
        \node[paneltitle,anchor=west]
        at (0.0,3.95)
        {(b) Line patch $Y_J$};

        \foreach \i in {0,...,3} {
            \draw[lattice] (\i,0) -- (\i,3);
        }
        \foreach \j in {0,...,3} {
            \draw[lattice] (0,\j) -- (0+3,\j);
        }

        \fill[cyan!35, opacity=0.5] (0.4,1.4) rectangle (2.6,1.6);

        \node[jlabel] at (1.5,1.75) {$J$};

        \draw[xbsupport] (1,1) -- (1,2);
        \draw[xbsupport] (2,1) -- (2,2);

        \node[xblabel] at (1.0,0.82) {$X^b$};
        \node[xblabel] at (2.0,0.82) {$X^b$};

        \draw[czgreen] (0.5,1.5) -- (0.5,2);
        \draw[czgreen] (0.5,1.5) -- (0.5,1);
        \draw[czgreen] (0.5,1.5) -- (0,1.5);
        \draw[czgreen] (0.5,1.5) -- (1,1.5);

        \node[czgreenlabel] at (0.33,1.75) {$\mathrm{CZ}$};
        \node[czgreenlabel] at (0.66,1.25) {$\mathrm{CZ}$};
        \node[czgreenlabel] at (0.25,1.35) {$\mathrm{CZ}$};
        \node[czgreenlabel] at (0.75,1.65) {$\mathrm{CZ}$};

        \node[cstar] at (0.5,2.22) {$C_{f_L}$};

        \draw[czgreen] (2.5,1.5) -- (2.5,2);
        \draw[czgreen] (2.5,1.5) -- (2.5,1);
        \draw[czgreen] (2.5,1.5) -- (2,1.5);
        \draw[czgreen] (2.5,1.5) -- (3,1.5);

        \node[czgreenlabel] at (2.33,1.75) {$\mathrm{CZ}$};
        \node[czgreenlabel] at (2.66,1.25) {$\mathrm{CZ}$};
        \node[czgreenlabel] at (2.25,1.35) {$\mathrm{CZ}$};
        \node[czgreenlabel] at (2.75,1.65) {$\mathrm{CZ}$};

        \node[cstar] at (2.5,2.22) {$C_{f_R}$};


        \foreach \i in {0,...,3} {
            \foreach \j in {0,...,2} {
                \node[bsite] at (\i,\j+0.5) {};
            }
        }

        \foreach \i in {0,...,2} {
            \foreach \j in {0,...,3} {
                \node[bsite] at (\i+0.5,\j) {};
            }
        }

        \foreach \i in {0,...,2} {
            \foreach \j in {0,...,2} {
                \node[asite] at (\i+0.5,\j+0.5) {};
            }
        }

        \node[formula]
        at (1.5,-0.47)
        {$C_f:=\prod_{e\subset\partial f}\mathrm{CZ}_{a_f,b_e}$};

        \node[formula]
        at (1.5,-0.82)
        {$P_J:=X_{e_1}^bX_{e_2}^b,\qquad
          Y_J=C_{f_L}P_JC_{f_R}$};
    \end{scope}

    \end{tikzpicture}

    \caption{
    Finite square-lattice patch operators.
    Red dots denote the $a$ qubits at face centers, and black dots denote the $b$ qubits on direct-lattice edges.
    \textbf{(a)} The disk patch $X_I$ is the restriction of the 0-form
    symmetry density to a single face $I$.
    \textbf{(b)} The line patch $Y_J$ connects the endpoint faces $f_L$ and $f_R$.
    The two thick blue direct-lattice edges support the two $X_e^b$ factors in $P_J$.
    At each endpoint face, the four green straight links show explicitly the four controlled-$Z$ gates in $C_f=\prod_{e\subset\partial f}\mathrm{CZ}_{a_f,b_e}$,
    coupling the central $a_f$ qubit to the four boundary-edge $b_e$ qubits.
    }
    \label{fig:square-patch-operators}
\end{figure}

\subsection{Disk and line patch operators}

Let $I$ be a disk-like union of faces.  The 0-form patch is the literal
restriction of the local density in Eq.~\eqref{eq:square-total-UA}:
\begin{equation}
    X_I:=\prod_{f\subset I}X_f^a\mathcal S_f.
    \label{eq:square-XI}
\end{equation}
The factors on different faces commute.  Consequently, for two face sets
with disjoint interiors,
\begin{equation}
    X_{I_1}X_{I_2}=X_{I_1\cup I_2}.
    \label{eq:square-X-gluing}
\end{equation}

The line patch requires endpoint decorations.  For every face, define the
star of controlled-$Z$ gates
\begin{equation}
    C_f:=\prod_{e\subset\partial f}\mathrm{CZ}_{a_f,b_e}.
    \label{eq:square-Cf}
\end{equation}
Let an oriented dual path $J$ pass through the sequence of faces
\begin{equation}
    f_0\xrightarrow{e_1}f_1\xrightarrow{e_2}\cdots
    \xrightarrow{e_n}f_n,
\end{equation}
where $e_r$ is the primal edge crossed between $f_{r-1}$ and $f_r$.  The
patch on one elementary dual edge is
\begin{equation}
    Y_{f_{r-1}\to f_r}:=C_{f_{r-1}}X_{e_r}^bC_{f_r},
\end{equation}
and the ordered product along $J$ is
\begin{equation}
    Y_J:={}
    Y_{f_0\to f_1}Y_{f_1\to f_2}\cdots Y_{f_{n-1}\to f_n} 
    =C_{f_L}\left(\prod_{r=1}^nX_{e_r}^b\right)C_{f_R}
    \label{eq:square-YJ}
\end{equation}
where $f_L=f_0$ and $f_R=f_n$.  The rightmost factor acts first.  The two
copies of $C_f$ at every internal face are adjacent and cancel.  The same
cancellation gives exact gluing: if $J_1$ ends where $J_2$ begins, then
\begin{equation}
    Y_{J_1}Y_{J_2}=Y_{J_1\cup J_2}.
    \label{eq:square-Y-gluing}
\end{equation}
The two patch symmetry operators are shown in Fig.~\ref{fig:square-patch-operators}.

\subsection{Symmetry and fusion}

The disk patch commutes with both total symmetries.  Its commutation with
$U_A$ is manifest, while $\dd\lambda=0$ implies that $U_B(\lambda)$ leaves
$c_f$ and hence $\mathcal S_f$ invariant.  Therefore
\begin{equation}
    [X_I,U_A]=[X_I,U_B(\lambda)]=1.
    \label{eq:square-X-commutation}
\end{equation}

We next check the line patch.  Write
\begin{equation}
    P_J:=\prod_{r=1}^nX_{e_r}^b,
    \qquad
    c_L:=c_{f_L},
    \qquad
    c_R:=c_{f_R}.
\end{equation}
The open string toggles $c_L$ and $c_R$ and leaves every interior face flux
unchanged.  Since
\begin{equation}
    \frac{i^{1-c}}{i^c}=i(-1)^c,
\end{equation}
the two endpoint faces give
\begin{equation}
    F_bP_JF_b^{-1}
    =
    -\mathcal B_{f_L}\mathcal B_{f_R}P_J.
    \label{eq:square-F-conjugation}
\end{equation}
On the other hand,
\begin{equation}
    K_aC_fK_a^{-1}=\mathcal B_fC_f.
    \label{eq:square-K-C-conjugation}
\end{equation}
Because $P_J$ intersects each endpoint face in one edge, moving
$\mathcal B_{f_R}$ through $P_J$ produces one minus sign.  Thus
\begin{equation}
    K_aY_JK_a^{-1}
    =
    -\mathcal B_{f_L}\mathcal B_{f_R}Y_J.
    \label{eq:square-K-Y-conjugation}
\end{equation}
Eqs.~\eqref{eq:square-F-conjugation} and
\eqref{eq:square-K-Y-conjugation} cancel in the product $U_A=K_aF_b$, so
the line patch commutes with $U_A$.  For the 1-form symmetry, conjugation by
$U_B(\lambda)$ gives
\begin{equation}
    U_B(\lambda)C_fU_B(\lambda)^{-1}
    =
    (Z_f^a)^{(\dd\lambda)(f)}C_f
    =C_f,
\end{equation}
and the $X_e^b$ string commutes with $U_B(\lambda)$.  Therefore
\begin{equation}
    [Y_J,U_A]=[Y_J,U_B(\lambda)]=1.
    \label{eq:square-Y-commutation}
\end{equation}
The twofold fusion of either patch leaves a defect of the other symmetry.  For
the disk patch,
\begin{equation}
\begin{aligned}
    X_I^2
    &=
    \prod_{f\subset I}\mathcal B_f
    =
    \prod_{e\subset\partial I}Z_e^b.
\end{aligned}
    \label{eq:square-X-square}
\end{equation}
For the line patch, $P_J$ crosses the boundary of each endpoint face once,
so
\begin{equation}
    P_JC_{f_L}C_{f_R}P_J
    =
    Z_{f_L}^aZ_{f_R}^aC_{f_L}C_{f_R}.
\end{equation}
After this conjugation, the endpoint stars cancel.  Hence
\begin{equation}
    Y_J^2=Z_{f_L}^aZ_{f_R}^a.
    \label{eq:square-Y-square}
\end{equation}
Eqs.~\eqref{eq:square-X-square} and \eqref{eq:square-Y-square} are the
square-lattice analogs of the one-dimensional endpoint fusion relations.

\subsection{Evaluation of the Bockstein word}

Let $\alpha_I(f)$ be the characteristic function of the disk $I$, and let
$\lambda_J$ be the one-cochain supported on the primal edges crossed by the
dual path $J$.  On a computational-basis state,
\begin{equation}
    X_I\ket{a,b}
    =
    i^{\phi_I(b)}\ket{a+\alpha_I,b},
    \qquad
    \phi_I(b):=\sum_{f\subset I}\widetilde c_f
    \quad (\mathrm{mod}\;4).
    \label{eq:square-X-action}
\end{equation}
It is useful to normal-order the line patch.  Since the string crosses the
left endpoint face once,
\begin{equation}
    Y_J
    =
    P_JE_J,
    \qquad
    E_J:=Z_{f_L}^aC_{f_L}C_{f_R}.
    \label{eq:square-Y-normal-order}
\end{equation}
Therefore
\begin{equation}
    Y_J\ket{a,b}
    =
    (-1)^{h_J(a,b)}\ket{a,b+\lambda_J},
    \label{eq:square-Y-action}
\end{equation}
where
\begin{equation}
    h_J(a,b)
    =
    a_{f_L}(1+c_L)+a_{f_R}c_R
    \quad (\mathrm{mod}\;2).
    \label{eq:square-hJ}
\end{equation}

The Bockstein word is
\begin{equation}
    W_2(X_I,Y_J)
    =
    (Y_J^{-1}X_I^{-1})^2(Y_JX_I)^2.
    \label{eq:square-W2-def}
\end{equation}
The phase $i^{\phi_I(b)}$ cancels between the two alternating staircases.
The remaining sign is the mixed second difference of $h_J$:
\begin{equation}
    W_2(X_I,Y_J)\ket{a,b}
    =
    (-1)^\Omega\ket{a,b},
\end{equation}
with
\begin{equation}
\begin{aligned}
    \Omega
    ={}&
    h_J(a+\alpha_I,b)
    +h_J(a,b+\lambda_J)+h_J(a+\alpha_I,b+\lambda_J)
    +h_J(a,b)
    \quad (\mathrm{mod}\;2).
\end{aligned}
    \label{eq:square-Omega-def}
\end{equation}
The open line toggles both endpoint fluxes, so Eq.~\eqref{eq:square-hJ}
gives
\begin{equation}
    \Omega
    =
    \alpha_I(f_L)+\alpha_I(f_R)
    \quad (\mathrm{mod}\;2).
    \label{eq:square-Omega-result}
\end{equation}
The individual definitions of $X_I$ and $Y_J$ are illustrated in
Fig.~\ref{fig:square-patch-operators}.  To evaluate the word, place the disk
$I$ and line $J$ in a staggered relative geometry: choose the left endpoint
face $f_L$ in the interior of $I$ and the right endpoint face $f_R$ outside
$I$.  Therefore,
\begin{equation}
    \alpha_I(f_L)=1,
    \qquad
    \alpha_I(f_R)=0.
\end{equation}
Consequently,
\begin{equation}
    W_2(X_I,Y_J)=-1.
    \label{eq:square-W2-minus-one}
\end{equation}
Thus, the mixed anomaly of the $\ZZ_2$ 0-form and $\ZZ_2$ 1-form symmetries is detected microscopically by particle-loop Bockstein braiding in
$(2{+}1)$ dimensions.

\subsection{Arbitrary triangulations and relation to the triangular-lattice representative}
\label{subsec:square-vs-triangular-gauge}

We first distinguish the algebraic face circuit from the standard
ordered-simplicial representative of the cochain
$\omega_2(b)=b\cup b+b\cup_1\dd b$.  The distinction is invisible in the
square-lattice calculation above but is important on a general
triangulation.

For a face $f$ of a regular cellulation $\mathcal C$ of a closed
two-manifold, define
\begin{equation}
    \mathcal S_f
    :=
    \left(\prod_{e\subset\partial f}S_e^b\right)
    \left(
    \prod_{\{e,e'\}\subset\partial f}
    \mathrm{CZ}_{e,e'}^b
    \right),
    \qquad
    C_f
    :=
    \prod_{e\subset\partial f}
    \mathrm{CZ}_{a_f,b_e},
    \label{eq:cellulation-local-factors}
\end{equation}
where the second product in $\mathcal S_f$ is over unordered pairs of
distinct boundary edges.  The disk density $X_f^a\mathcal S_f$ and the
endpoint decoration $C_f$ obey the same local symmetry, fusion, and gluing
algebra as on the square lattice.  For a mod-two edge chain
$\gamma=\sum_e\gamma_e e$, write
\begin{equation}
    Z_b(\gamma):=\prod_e\bigl(Z_e^b\bigr)^{\gamma_e},
    \qquad
    E_{\mathcal C}:=\sum_{e\in\mathcal C_1}e,
    \qquad
    C_b^{\mathcal C}
    :=
    \prod_f
    \prod_{\{e,e'\}\subset\partial f}
    \mathrm{CZ}_{e,e'}^b .
    \label{eq:cellulation-chain-notation}
\end{equation}
Every edge of such a regular cellulation belongs to two faces.  Hence
$S_e^2=Z_e^b$ gives the exact algebraic identity
\begin{equation}
\begin{aligned}
    F_b^{\mathcal C}
    &:={}
    \prod_f\mathcal S_f
    =Z_b(E_{\mathcal C})C_b^{\mathcal C}
    =\left(\prod_e Z_e^b\right)C_b^{\mathcal C},\\
    U_{A,\mathcal C}
    &:={}K_aF_b^{\mathcal C}.
\end{aligned}
    \label{eq:square-Fb-CZ-decomposition}
\end{equation}
This formula defines a useful face-local representative on any such cellulation.  It should not, by itself, be identified face by face with the
standard simplicial cup-product formula.

\paragraph{The ordered-simplicial formula.}
Let $K$ be an arbitrary finite triangulation of the closed surface $\Sigma$,
and fix a total order $\prec$ of its vertices.  Write each ordered triangle as
$f=[012]$, and let $F_b^K$ and $C_b^K$ denote the operators in
Eqs.~\eqref{eq:square-Fb-CZ-decomposition} and
\eqref{eq:cellulation-chain-notation}, respectively, specialized to this
triangulation.
With the standard simplicial cup-$1$ convention,
\begin{equation}
\begin{aligned}
    \omega_2(b)(012)
    &=
    (b\cup b)(012)+(b\cup_1\dd b)(012)\\
    &=
    b_{01}b_{12}+b_{02}(\dd b)_{012}
    =
    b_{01}b_{12}+b_{01}b_{02}+b_{02}b_{12}+b_{02}.
\end{aligned}
    \label{eq:branched-omega2-face}
\end{equation}
Thus each triangle contributes the three pairwise controlled-$Z$ gates and a
single $Z^b$ on its long edge $[02]$.
Define the mod-two long-edge chain
\begin{equation}
    L_\prec
    :=
    \sum_{f=[012]\in K_2}[02]_f,
    \qquad
    E_K:=\sum_{e\in K_1}e ,
    \label{eq:branched-long-edge-chain}
\end{equation}
where repeated occurrences of an edge cancel.  Multiplication of
Eq.~\eqref{eq:branched-omega2-face} over all triangles gives
\begin{equation}
    \mathcal U_{\omega_2}\ket b
    :=
    (-1)^{\langle\omega_2(b),[K]\rangle}\ket b,
    \qquad
    \mathcal U_{\omega_2}
    =
    Z_b(L_\prec)C_b^K .
    \label{eq:branched-omega2-circuit}
\end{equation}
This is the exact operator identity for an arbitrary ordered triangulation;
no condition $\dd b=0$ has been imposed.

\paragraph{The Goldstein--Turner chain.}
We recall the definition of a regular face used by Goldstein and Turner
\cite{GoldsteinTurner1976}.  Let
$s=[u_0\cdots u_p]$ be an ordered $p$-face of an ordered simplex $t$.
Let $B_{-1}$ be the vertices of $t$ below $u_0$, let $B_j$ be the vertices
strictly between $u_j$ and $u_{j+1}$ for $0\leq j<p$, and let $B_p$ be the
vertices above $u_p$.  The face $s$ is \emph{regular in $t$} when
$B_j=\varnothing$ for every odd $j$ with $-1\leq j\leq p$.
Denote by $\alpha_p(t)$ the mod-two sum of all
regular $p$-faces of $t$.  Every edge is regular in itself, whereas the unique
regular edge of an ordered triangle $[012]$ is its long edge $[02]$.
The Goldstein--Turner formula therefore specializes in two dimensions to
\begin{equation}
    \mathfrak w_1^{\mathrm{GT}}(K,\prec)
    :=
    \sum_{\substack{t\in K\\ \dim t\geq1}}\alpha_1(t)
    =
    E_K+L_\prec .
    \label{eq:goldstein-turner-chain}
\end{equation}
It is a cycle, and its homology class is the first Stiefel--Whitney homology
class,
\begin{equation}
    \bigl[\mathfrak w_1^{\mathrm{GT}}(K,\prec)\bigr]
    =
    w^1(T\Sigma)\cap[\Sigma]
    =
    \operatorname{PD}\bigl(w^1(T\Sigma)\bigr).
    \label{eq:goldstein-turner-class}
\end{equation}
Combining Eqs.~\eqref{eq:square-Fb-CZ-decomposition},
\eqref{eq:branched-omega2-circuit}, and
\eqref{eq:goldstein-turner-chain} yields the general-triangulation formula
\begin{equation}
    \mathcal U_{\omega_2}
    =
    \left[
    \prod_{e\in K_1}
    \bigl(Z_e^b\bigr)^{\mathfrak w_1^{\mathrm{GT}}(e)}
    \right]
    \left(\prod_{e\in K_1}Z_e^b\right)
    \left[
    \prod_{f\in K_2}
    \prod_{\{e,e'\}\subset\partial f}
    \mathrm{CZ}_{e,e'}^b
    \right].
    \label{eq:goldstein-turner-corrected-circuit}
\end{equation}
Here, $\mathfrak w_1^{\mathrm{GT}}(e)$ means the coefficient of the edge $e$ in
the particular chain in Eq.~\eqref{eq:goldstein-turner-chain}; it is not the evaluation of an arbitrary cocycle representative of $w^1$.

\paragraph{Removing the Stiefel--Whitney factor with the face qubits.}
The physical surface $\Sigma=\partial M$ in this appendix is closed and
oriented.  Fix its orientation and let $D_\prec\in C_2(K,\ZZ_2)$ be the sum
of the triangles whose ordered orientation disagrees with it.  The
orientation-wall description of the Goldstein--Turner chain gives
\begin{equation}
    \partial D_\prec
    =
    \mathfrak w_1^{\mathrm{GT}}(K,\prec).
    \label{eq:goldstein-turner-boundary}
\end{equation}
Indeed, on an edge shared by two triangles, the coefficient of
$E_K+L_\prec$ equals the parity of the two adjacent orientation disagreements,
which is exactly the coefficient of $\partial D_\prec$.
Using the face $a$ qubits, define the finite-depth circuit
\begin{equation}
    R_\prec
    :=
    \prod_{f\in D_\prec}C_f
    =
    \prod_{f\in D_\prec}
    \prod_{e\subset\partial f}\mathrm{CZ}_{a_f,b_e}.
    \label{eq:goldstein-turner-removal-circuit}
\end{equation}
The corresponding bipartite incidence graph has maximum degree three, so
these commuting controlled-$Z$ gates can be scheduled in three layers.
Because
$C_fX_f^aC_f^{-1}=X_f^a\mathcal B_f$, this circuit satisfies
\begin{equation}
    R_\prec K_aR_\prec^{-1}
    =
    K_a Z_b(\partial D_\prec)
    =
    K_a Z_b\!\left(\mathfrak w_1^{\mathrm{GT}}\right).
    \label{eq:goldstein-turner-removal-Ka}
\end{equation}
It also preserves every closed 1-form symmetry transformation:
\begin{equation}
    R_\prec U_B(\lambda)R_\prec^{-1}
    =
    U_B(\lambda)
    \prod_{f\in D_\prec}
    \bigl(Z_f^a\bigr)^{(\dd\lambda)(f)}
    =
    U_B(\lambda).
    \label{eq:goldstein-turner-removal-UB}
\end{equation}
Since $R_\prec$ commutes with all diagonal $b$-qubit gates,
Eq.~\eqref{eq:goldstein-turner-corrected-circuit} implies the exact
off-shell equivalence
\begin{equation}
    R_\prec
    \bigl(K_a\mathcal U_{\omega_2}\bigr)
    R_\prec^{-1}
    =
    K_aF_b^K .
    \label{eq:goldstein-turner-removal-UA}
\end{equation}
Thus, the Stiefel--Whitney chain factor does not affect the anomaly or the Bockstein statistic in the oriented setting considered here: after
adjoining the face $a$ qubits, it is removed by a finite-depth,
$U_B$-preserving change of basis.  Without the $a$ qubits it is not literally
the identity for an arbitrary off-shell $b$ configuration.  On a
nonorientable surface the chain in
Eq.~\eqref{eq:goldstein-turner-chain} need not be a boundary, and the same
removal while keeping the bare $U_B$ fixed need not exist.

\paragraph{The controlled-$Z$-only representative.}
There is a separate simplification when the faces of $\mathcal C$ admit a
two-coloring.  Let
$\mathcal F_\bullet$ denote one of the two face colors, so that every edge
borders exactly one face in $\mathcal F_\bullet$, and define
\begin{equation}
    R_{\mathcal C}
    :=
    \prod_{f\in\mathcal F_\bullet}
    \prod_{e\subset\partial f}
    \mathrm{CZ}_{a_f,b_e}.
    \label{eq:square-gauge-circuit}
\end{equation}
This is a finite-depth circuit whose depth is bounded by the maximal number
of edges surrounding a face. In particular, four controlled-$Z$ layers
suffice on the square lattice and three suffice on the triangular lattice.
Since every edge occurs exactly once in the product over
$\mathcal F_\bullet$,
\begin{equation}
    R_{\mathcal C}K_aR_{\mathcal C}^{-1}
    =
    K_a\prod_e Z_e^b .
    \label{eq:square-gauge-Ka}
\end{equation}
The same circuit preserves every closed 1-form transformation:
\begin{equation}
    R_{\mathcal C}U_B(\lambda)R_{\mathcal C}^{-1}
    =
    U_B(\lambda)
    \prod_{f\in\mathcal F_\bullet}
    \bigl(Z_f^a\bigr)^{(\dd\lambda)(f)}
    =
    U_B(\lambda),
    \label{eq:square-gauge-UB}
\end{equation}
where the last equality follows from $\dd\lambda=0$.  Since
$R_{\mathcal C}$ also commutes with the diagonal operator
$F_b^{\mathcal C}$,
Eqs.~\eqref{eq:square-Fb-CZ-decomposition}
and~\eqref{eq:square-gauge-Ka} imply
\begin{equation}
    R_{\mathcal C}U_{A,\mathcal C}R_{\mathcal C}^{-1}
    =
    U_{A,\mathcal C}^{\mathrm{CZ}}
    :=
    K_aC_b^{\mathcal C}.
    \label{eq:square-CZ-gauge}
\end{equation}

For the regular triangular lattice, one may take
$\mathcal F_\bullet$ to be the set of upward-pointing triangles. Then
\begin{equation}
    C_b^\triangle
    =
    \prod_f
    \prod_{\{e,e'\}\subset\partial f}
    \mathrm{CZ}_{e,e'}^b
    =
    U_0,
\end{equation}
where $U_0$ is the 0-form generator in
Eq.~\eqref{eq:triangular-total-symmetries-main}.
Likewise, the elementary closed
1-form transformation surrounding a vertex $v$ is
\begin{equation}
    U_B(\lambda_v)
    =
    \prod_{e\ni v}X_e^b
    =
    U_1(v).
\end{equation}
Thus, after adjoining the anomaly-trivial onsite $a$-qubit sector and applying
the finite-depth circuit $R_{\mathcal C}$, the all-edge face circuit becomes
the stabilized controlled-$Z$-only representative of
Eq.~\eqref{eq:triangular-total-symmetries-main}.
Starting from the standard
simplicial representative on an arbitrary ordered triangulation of the
oriented surface, one first applies $R_\prec$ to remove the
Goldstein--Turner factor and then $R_{\mathcal C}$ when a face two-coloring is
available.

We retain the square/cubical gauge in the patch calculation above because its
face-local density $X_f^a\mathcal S_f$ makes symmetry and exact gluing
manifest.  Operators in any of the equivalent gauges are obtained by
conjugating the complete patch operators by the corresponding finite-depth
circuit; individual $S$ or $Z$ factors should not simply be deleted.

\section{Patch symmetry operators for an anomalous
$\ZZ_2^{(0)}\times\ZZ_2^{(2)}$ symmetry in $(3{+}1)$D}
\label{app:cubic-lattice-3d}

This appendix generalizes the $(2{+}1)$D construction of
Appendix~\ref{app:square-lattice-2d} to three spatial dimensions.  We work
from the outset with an arbitrary regular cellulation, and then compare its
local circuit with the standard cup-product formula on an arbitrary ordered
triangulation.  The anomalous $\ZZ_2^{(0)}\times\ZZ_2^{(2)}$ symmetry descends
from the $(4{+}1)$D 5-cocycle
\begin{equation}
    \frac12 A_1\cup\bigl(B_3\cup_2B_3\bigr).
    \label{eq:cubic-bulk-cocycle}
\end{equation}
For a cocycle $B_3$, the second factor represents
$\operatorname{Sq}^1B_3=\beta_2B_3$.

\subsection{Boundary descent and a local cellulation representative}

Let $M_4$ be an oriented four-dimensional spatial manifold with closed
oriented boundary $\Sigma_3=\partial M_4$.  Locally write
\begin{equation}
    A_1=\dd a,
    \qquad
    B_3=\dd b,
    \qquad
    a\in C^0(M_4,\ZZ_2),
    \qquad
    b\in C^2(M_4,\ZZ_2).
\end{equation}
The corresponding fixed-point bulk wavefunction is
\begin{equation}
    |\Psi\rangle
    \propto
    \sum_{a,b}
    (-1)^{
        \int_{M_4}
        a\cup\bigl(\dd b\cup_2\dd b\bigr)
    }
    |a,b\rangle .
    \label{eq:cubic-bulk-wavefunction}
\end{equation}
For $\ZZ_2$-valued cochains, the higher-cup differential identity is
\begin{equation}
    \dd(x\cup_i y)
    =
    \dd x\cup_i y
    +x\cup_i\dd y
    +x\cup_{i-1}y
    +y\cup_{i-1}x .
    \label{eq:cubic-cup-i-identity}
\end{equation}
On an ordered triangulation, define the standard simplicial descent
representative
\begin{equation}
    \omega_3(b)
    :=
    b\cup_1b+b\cup_2\dd b .
    \label{eq:cubic-omega3}
\end{equation}
Indeed,
\begin{align}
    \dd(b\cup_1b)
    &=
    \dd b\cup_1b+b\cup_1\dd b,\\
    \dd(b\cup_2\dd b)
    &=
    \dd b\cup_2\dd b
    +b\cup_1\dd b
    +\dd b\cup_1b .
\end{align}
The mixed terms cancel, and therefore
\begin{equation}
    \dd\omega_3(b)
    =
    \dd b\cup_2\dd b .
    \label{eq:cubic-descent}
\end{equation}
The two boundary symmetries consequently act as
\begin{equation}
    U_A|a,b\rangle
    =
    (-1)^{
        \int_{\Sigma_3}
        \left(b\cup_1b+b\cup_2\dd b\right)
    }
    |a+1,b\rangle
    \label{eq:cubic-UA-cochain}
\end{equation}
and
\begin{equation}
    U_B(\lambda)|a,b\rangle
    =
    |a,b+\lambda\rangle,
    \qquad
    \lambda\in C^2(\Sigma_3,\ZZ_2),
    \qquad
    \dd\lambda=0 .
    \label{eq:cubic-UB-cochain}
\end{equation}
The support of $\lambda$ is Poincar\'e dual to a closed line in
$\Sigma_3$, as appropriate for a $2$-form symmetry.

We now give a local representative on an arbitrary finite regular
cellulation $\mathcal C$ of $\Sigma_3$.  Place one $a$ qubit in every
3-cell $c\in\mathcal C_3$ and one $b$ qubit on every
two-face $f\in\mathcal C_2$.  Define
\begin{equation}
    \chi_c
    :=
    (\dd b)(c)
    =
    \sum_{f\subset\partial c}b_f
    \quad (\mathrm{mod}\;2),
    \qquad
    \mathcal B_c
    :=
    \prod_{f\subset\partial c}Z_f^b
    =
    (-1)^{\chi_c}.
    \label{eq:cubic-flux}
\end{equation}
Let $S=\operatorname{diag}(1,i)$ and set
\begin{equation}
    \mathcal S_c
    :=
    \left(
        \prod_{f\subset\partial c}S_f^b
    \right)
    \left(
        \prod_{\{f,f'\}\subset\partial c}
        \mathrm{CZ}_{f,f'}^b
    \right),
    \label{eq:cubic-cell-circuit}
\end{equation}
where the second product is over all unordered pairs of distinct boundary
faces.  If $\widetilde\chi_c\in\{0,1\}$ is the canonical integer lift of
$\chi_c$, then
\begin{equation}
    \mathcal S_c|b\rangle
    =
    i^{\widetilde\chi_c}|b\rangle,
    \qquad
    \mathcal S_c^2=\mathcal B_c .
    \label{eq:cubic-cell-circuit-action}
\end{equation}
For a cell with boundary-face bits $b_1,\ldots,b_q$, this follows from
\begin{equation}
    i^{b_1\oplus\cdots\oplus b_q}
    =
    i^{\sum_rb_r}
    (-1)^{\sum_{r<s}b_rb_s}.
\end{equation}
Thus only one-qubit $S$ gates and two-qubit controlled-$Z$ gates are
required, independently of the shape of the cell.

To identify the represented cochain, orient the 3-cells coherently and
let $\widetilde b$ be the canonical integer lift of $b$.  Define
\begin{equation}
    n_c:=(\dd\widetilde b)(c)\in\ZZ,
    \qquad
    \omega_{3,\mathcal C}^{\mathrm{lift}}(b)(c)
    :=
    \frac{n_c-\widetilde\chi_c}{2}
    \quad (\mathrm{mod}\;2).
    \label{eq:cubic-cubical-omega}
\end{equation}
The numerator is even.  This integer-lift expression is a cellular
representative of the same Bockstein descent as
$b\cup_1b+b\cup_2\dd b$~\cite{CH21}.  Cell by cell,
\begin{equation}
    i^{\widetilde\chi_c}
    =
    (-1)^{\omega_{3,\mathcal C}^{\mathrm{lift}}(b)(c)}
    i^{n_c}.
\end{equation}
Since the oriented contributions of every internal face cancel,
$\sum_cn_c=\int_{\Sigma_3}\dd\widetilde b=0$, and hence
\begin{equation}
    \prod_c\mathcal S_c|b\rangle
    =
    (-1)^{
        \int_{\Sigma_3}
        \omega_{3,\mathcal C}^{\mathrm{lift}}(b)
    }|b\rangle .
    \label{eq:cubic-cell-product}
\end{equation}

It is useful to record the corresponding operator identity.  For a mod-two
two-chain $\Gamma=\sum_f\Gamma_f f$, write
\begin{equation}
    Z_b(\Gamma)
    :=
    \prod_f\bigl(Z_f^b\bigr)^{\Gamma_f},
    \qquad
    \mathsf F_{\mathcal C}
    :=
    \sum_{f\in\mathcal C_2}f,
    \qquad
    C_b^{\mathcal C}
    :=
    \prod_{c\in\mathcal C_3}
    \prod_{\{f,f'\}\subset\partial c}
    \mathrm{CZ}_{f,f'}^b .
\end{equation}
Every face belongs to two 3-cells.  Therefore $S_f^2=Z_f^b$ gives
\begin{equation}
    F_b^{\mathcal C}
    :=
    \prod_{c\in\mathcal C_3}\mathcal S_c
    =
    Z_b(\mathsf F_{\mathcal C})C_b^{\mathcal C}
    =
    \left(\prod_{f\in\mathcal C_2}Z_f^b\right)
    C_b^{\mathcal C}.
    \label{eq:cubic-cellulation-decomposition}
\end{equation}
The total symmetry generators are
\begin{equation}
    U_{A,\mathcal C}
    =
    K_aF_b^{\mathcal C},
    \qquad
    K_a
    :=
    \prod_{c\in\mathcal C_3}X_c^a,
    \label{eq:cubic-total-UA}
\end{equation}
and
\begin{equation}
    U_B(\lambda)
    =
    \prod_{f\in\mathcal C_2}
    \bigl(X_f^b\bigr)^{\lambda(f)},
    \qquad
    \dd\lambda=0 .
    \label{eq:cubic-total-UB}
\end{equation}
A closed transformation leaves every $\chi_c$ invariant, so
$[U_{A,\mathcal C},U_B(\lambda)]=1$.  Moreover,
\begin{equation}
    U_{A,\mathcal C}^2
    =
    \prod_c\mathcal B_c
    =
    1
\end{equation}
on a closed cellulation.  An elementary closed $2$-form transformation
surrounding a primal edge $e$ is
\begin{equation}
    U_2(e)
    =
    \prod_{f\supset e}X_f^b ,
    \label{eq:cubic-elementary-U2}
\end{equation}
whose support is the smallest dual-lattice loop linking $e$.

\subsection{Arbitrary ordered triangulations and the
Goldstein--Turner factor}

The cellular circuit above should not be identified face by face with the
standard ordered-simplicial density without keeping track of the ordering.
Let $K$ be an arbitrary finite triangulation of $\Sigma_3$, equipped with a
total order $\prec$ on its vertices.  On an ordered tetrahedron
$t=[0123]$, write
\begin{equation}
    x_0=b_{123},
    \qquad
    x_1=b_{023},
    \qquad
    x_2=b_{013},
    \qquad
    x_3=b_{012}.
\end{equation}
The integer coboundary and the canonical lift of the mod-two coboundary are
\begin{equation}
    n_t
    =
    x_0-x_1+x_2-x_3,
    \qquad
    \widetilde\chi_t
    =
    x_0\oplus x_1\oplus x_2\oplus x_3.
\end{equation}
For the standard simplicial cup-$i$ convention,
\begin{equation}
\begin{aligned}
    (b\cup_1b)(0123)
    &=x_1x_3+x_0x_2,\\
    (b\cup_2\dd b)(0123)
    &=(x_1+x_3)(x_0+x_1+x_2+x_3).
\end{aligned}
    \label{eq:cubic-cup-i-tetrahedron}
\end{equation}
Equivalently, the same density has the integer-lift expansion
\begin{equation}
\begin{aligned}
    \omega_{3,\prec}^{\triangle}(b)(0123)
    &:=
    \bigl(b\cup_1b+b\cup_2\dd b\bigr)(0123)\\
    &=
    \frac{n_t-\widetilde\chi_t}{2}\\
    &=
    \sum_{0\leq r<s\leq3}x_rx_s+x_1+x_3
    \quad (\mathrm{mod}\;2)\\
    &=
    \sum_{\{f,f'\}\subset\partial t}b_fb_{f'}
    +b_{023}+b_{012}
    \quad (\mathrm{mod}\;2).
\end{aligned}
    \label{eq:cubic-branched-omega3}
\end{equation}
Thus every tetrahedron contributes the six pairwise controlled-$Z$ gates
among its four boundary-face qubits, together with one $Z^b$ on each of the
two faces $[012]$ and $[023]$.

Define the $\mathbb Z_2$-valued two-chains
\begin{equation}
    L_\prec
    :=
    \sum_{t=[0123]\in K_3}
    \bigl([012]_t+[023]_t\bigr),
    \qquad
    \mathsf F_K
    :=
    \sum_{f\in K_2}f ,
    \label{eq:cubic-distinguished-face-chain}
\end{equation}
where repeated faces cancel.  Multiplying
Eq.~\eqref{eq:cubic-branched-omega3} over all tetrahedra gives the exact
off-shell identity
\begin{equation}
    \mathcal U_{\omega_3}|b\rangle
    :=
    (-1)^{
        \langle\omega_{3,\prec}^{\triangle}(b),[K]\rangle
    }|b\rangle,
    \qquad
    \mathcal U_{\omega_3}
    =
    Z_b(L_\prec)C_b^K .
    \label{eq:cubic-branched-omega3-circuit}
\end{equation}

We now identify the linear chain using the regular-face construction of
Goldstein and Turner~\cite{GoldsteinTurner1976}, recalled in
Appendix~\ref{subsec:square-vs-triangular-gauge}.  Every triangle is regular
in itself, while the regular two-faces of $[0123]$ are precisely $[012]$
and $[023]$.  Consequently,
\begin{equation}
    \mathfrak w_1^{\mathrm{GT}}(K,\prec)
    :=
    \sum_{\substack{t\in K\\ \dim t\geq2}}\alpha_2(t)
    =
    \mathsf F_K+L_\prec .
    \label{eq:cubic-goldstein-turner-chain}
\end{equation}
This is a two-cycle, and its homology class is
\begin{equation}
    \bigl[\mathfrak w_1^{\mathrm{GT}}(K,\prec)\bigr]
    =
    w^1(T\Sigma_3)\cap[\Sigma_3]
    =
    \operatorname{PD}\bigl(w^1(T\Sigma_3)\bigr).
    \label{eq:cubic-goldstein-turner-class}
\end{equation}
It follows that the standard simplicial descent is
\begin{equation}
\begin{aligned}
    \mathcal U_{\omega_3}
    &=
    Z_b\!\left(\mathfrak w_1^{\mathrm{GT}}\right)
    Z_b(\mathsf F_K)C_b^K\\
    &=
    \left[
        \prod_{f\in K_2}
        \bigl(Z_f^b\bigr)^{\mathfrak w_1^{\mathrm{GT}}(f)}
    \right]
    \left(\prod_{f\in K_2}Z_f^b\right)
    \left[
        \prod_{t\in K_3}
        \prod_{\{f,f'\}\subset\partial t}
        \mathrm{CZ}_{f,f'}^b
    \right].
\end{aligned}
    \label{eq:cubic-goldstein-turner-corrected-circuit}
\end{equation}
Here, $\mathfrak w_1^{\mathrm{GT}}(f)$ is the coefficient of $f$ in the
particular chain in Eq.~\eqref{eq:cubic-goldstein-turner-chain}; it is not
the evaluation of an arbitrary cocycle representative.

\subsection{Removing the Stiefel--Whitney factor}

Because $\Sigma_3=\partial M_4$ is oriented, the cycle in
Eq.~\eqref{eq:cubic-goldstein-turner-chain} is a boundary.  More explicitly,
fix the orientation of $\Sigma_3$ and let
$D_\prec\in C_3(K,\ZZ_2)$ be the sum of the tetrahedra whose ordered
orientation disagrees with it.  The orientation-wall description of the
Goldstein--Turner chain gives
\begin{equation}
    \partial D_\prec
    =
    \mathfrak w_1^{\mathrm{GT}}(K,\prec).
    \label{eq:cubic-goldstein-turner-boundary}
\end{equation}
For every tetrahedron, define
\begin{equation}
    C_t^{ab}
    :=
    \prod_{f\subset\partial t}
    \mathrm{CZ}_{a_t,b_f},
    \qquad
    R_\prec
    :=
    \prod_{t\in D_\prec}C_t^{ab}.
    \label{eq:cubic-goldstein-turner-removal-circuit}
\end{equation}
The tetrahedron--face incidence graph has maximum degree four, so the gates
in $R_\prec$ can be scheduled in four layers.  Conjugating the onsite
$a$-qubit flip gives
\begin{equation}
    R_\prec K_aR_\prec^{-1}
    =
    K_aZ_b(\partial D_\prec)
    =
    K_aZ_b\!\left(\mathfrak w_1^{\mathrm{GT}}\right).
    \label{eq:cubic-goldstein-turner-removal-Ka}
\end{equation}
The same circuit preserves every closed $2$-form transformation:
\begin{equation}
    R_\prec U_B(\lambda)R_\prec^{-1}
    =
    U_B(\lambda)
    \prod_{t\in D_\prec}
    \bigl(Z_t^a\bigr)^{(\dd\lambda)(t)}
    =
    U_B(\lambda).
    \label{eq:cubic-goldstein-turner-removal-UB}
\end{equation}
Since $R_\prec$ commutes with all diagonal $b$-qubit gates,
Eq.~\eqref{eq:cubic-goldstein-turner-corrected-circuit} implies
\begin{equation}
    R_\prec
    \bigl(K_a\mathcal U_{\omega_3}\bigr)
    R_\prec^{-1}
    =
    K_aF_b^K .
    \label{eq:cubic-goldstein-turner-removal-UA}
\end{equation}
Thus the order-dependent Stiefel--Whitney factor does not affect the anomaly
or the Bockstein statistic in the oriented setting considered here: after
including the $a$ qubits, it is removed by a finite-depth,
$U_B$-preserving change of basis.  It should nevertheless be retained in
the literal off-shell circuit on an arbitrary ordered triangulation.  On a
nonorientable three-manifold the same cycle need not be a boundary.

Below we use the face-local cellulation gauge, in which gluing is manifest.
On an ordered triangulation the corresponding patches are obtained by the
common conjugation $X_I\mapsto R_\prec^{-1}X_IR_\prec$ and
$Y_J\mapsto R_\prec^{-1}Y_JR_\prec$.  The Bockstein word is conjugated as a
whole, so its scalar value is unchanged.

\subsection{Volume and line patches on an arbitrary cellulation}

We return to a regular cellulation $\mathcal C$.  Let $I$ be a union of
3-cells and define the finite $0$-form patch by restricting the local
symmetry density:
\begin{equation}
    X_I
    :=
    \prod_{c\subset I}X_c^a\mathcal S_c .
    \label{eq:cubic-XI}
\end{equation}
The factors commute, so these patches glue exactly under unions with
disjoint interiors.  Their twofold fusion leaves the $2$-form symmetry
operator on the boundary:
\begin{equation}
    X_I^2
    =
    \prod_{c\subset I}\mathcal B_c
    =
    \prod_{f\subset\partial I}Z_f^b .
    \label{eq:cubic-X-square}
\end{equation}

For the $2$-form patch, let an oriented dual path pass through
\begin{equation}
    J:\quad
    c_0\xrightarrow{f_1}c_1
    \xrightarrow{f_2}\cdots
    \xrightarrow{f_n}c_n ,
\end{equation}
where $f_r$ is the primal face crossed between $c_{r-1}$ and $c_r$.  Define
\begin{equation}
    P_J
    :=
    \prod_{r=1}^nX_{f_r}^b,
    \qquad
    C_c^{ab}
    :=
    \prod_{f\subset\partial c}
    \mathrm{CZ}_{a_c,b_f}.
\end{equation}
On one elementary dual edge, set
\begin{equation}
    Y_{c_{r-1}\to c_r}
    :=
    C_{c_{r-1}}^{ab}X_{f_r}^bC_{c_r}^{ab}.
\end{equation}
The ordered product along the path is
\begin{equation}
    Y_J
    =
    Y_{c_0\to c_1}\cdots Y_{c_{n-1}\to c_n}
    =
    C_{c_L}^{ab}P_JC_{c_R}^{ab},
    \qquad
    c_L=c_0,
    \quad
    c_R=c_n .
    \label{eq:cubic-YJ}
\end{equation}
The two endpoint stars at every internal cell are adjacent and cancel.
This also gives exact gluing of adjacent line patches.

\subsection{Symmetry and fusion of the patches}

The volume patch commutes with both global symmetries:
\begin{equation}
    [X_I,U_{A,\mathcal C}]
    =
    [X_I,U_B(\lambda)]
    =
    1 .
\end{equation}
For the line patch, the open string toggles $\chi_{c_L}$ and
$\chi_{c_R}$, while every internal cell is crossed twice.  Since
$i^{1-\chi}/i^\chi=i(-1)^\chi$, the two endpoints give
\begin{equation}
    F_b^{\mathcal C}P_J
    \bigl(F_b^{\mathcal C}\bigr)^{-1}
    =
    -\mathcal B_{c_L}\mathcal B_{c_R}P_J .
    \label{eq:cubic-F-conjugation}
\end{equation}
On the other hand,
\begin{equation}
    K_aC_c^{ab}K_a^{-1}
    =
    \mathcal B_cC_c^{ab}.
\end{equation}
Moving $\mathcal B_{c_R}$ through $P_J$ produces one further minus sign,
because $J$ crosses one face of $c_R$.  Therefore
\begin{equation}
    K_aY_JK_a^{-1}
    =
    -\mathcal B_{c_L}\mathcal B_{c_R}Y_J,
\end{equation}
and the two endpoint factors cancel those in
Eq.~\eqref{eq:cubic-F-conjugation}.  Moreover,
\begin{equation}
    U_B(\lambda)C_c^{ab}U_B(\lambda)^{-1}
    =
    \bigl(Z_c^a\bigr)^{(\dd\lambda)(c)}C_c^{ab}
    =
    C_c^{ab}.
\end{equation}
It follows that
\begin{equation}
    [Y_J,U_{A,\mathcal C}]
    =
    [Y_J,U_B(\lambda)]
    =
    1 .
\end{equation}
Finally, $P_J$ intersects the boundary of each endpoint cell once, so the
endpoint stars cancel after conjugation and
\begin{equation}
    Y_J^2
    =
    Z_{c_L}^aZ_{c_R}^a .
    \label{eq:cubic-Y-square}
\end{equation}
Equations~\eqref{eq:cubic-X-square} and \eqref{eq:cubic-Y-square} show that
the twofold fusion of either patch leaves a defect of the other symmetry.

\subsection{Evaluation of the Bockstein word}

Let $\alpha_I(c)$ be the characteristic function of $I$, and let
$\lambda_J$ be the two-cochain supported on the crossed faces
$f_1,\ldots,f_n$.  On a computational-basis state,
\begin{equation}
    X_I|a,b\rangle
    =
    i^{\phi_I(b)}|a+\alpha_I,b\rangle,
    \qquad
    \phi_I(b)
    :=
    \sum_{c\subset I}\widetilde\chi_c
    \quad (\mathrm{mod}\;4).
    \label{eq:cubic-X-action}
\end{equation}
Normal-ordering the line patch gives
\begin{equation}
    Y_J
    =
    P_JE_J,
    \qquad
    E_J
    :=
    Z_{c_L}^aC_{c_L}^{ab}C_{c_R}^{ab}.
    \label{eq:cubic-Y-normal-order}
\end{equation}
Consequently,
\begin{equation}
    Y_J|a,b\rangle
    =
    (-1)^{h_J(a,b)}|a,b+\lambda_J\rangle,
\end{equation}
where
\begin{equation}
    h_J(a,b)
    =
    a_{c_L}(1+\chi_{c_L})
    +a_{c_R}\chi_{c_R}
    \quad (\mathrm{mod}\;2).
    \label{eq:cubic-hJ}
\end{equation}

The Bockstein word is
\begin{equation}
    W_2(X_I,Y_J)
    =
    (Y_J^{-1}X_I^{-1})^2(Y_JX_I)^2 .
\end{equation}
The phase $i^{\phi_I(b)}$ cancels between the two alternating staircases.
The remaining sign is the mixed second difference of $h_J$:
\begin{equation}
    W_2(X_I,Y_J)|a,b\rangle
    =
    (-1)^\Omega|a,b\rangle,
\end{equation}
with
\begin{equation}
\begin{aligned}
    \Omega
    ={}&
    h_J(a+\alpha_I,b)
    +h_J(a,b+\lambda_J)
    +h_J(a+\alpha_I,b+\lambda_J)
    +h_J(a,b)
    \quad (\mathrm{mod}\;2).
\end{aligned}
\end{equation}
The open line toggles both endpoint fluxes, so
\begin{equation}
    \Omega
    =
    \alpha_I(c_L)+\alpha_I(c_R)
    \quad (\mathrm{mod}\;2).
\end{equation}
Choose the staggered geometry in which $c_L$ lies inside $I$ and $c_R$
lies outside.  Then $\Omega=1$, and hence
\begin{equation}
    W_2(X_I,Y_J)
    =
    -1 .
    \label{eq:cubic-W2-minus-one}
\end{equation}
This is the particle--membrane Bockstein braiding invariant in three
spatial dimensions: $X_I$ creates a membrane on $\partial I$, while the
open $2$-form patch $Y_J$ creates particles at its endpoints.

\subsection{Cubic-lattice specialization and a controlled-$Z$-only gauge}

For the standard cubic cellulation, each 3-cell has six boundary faces.
Thus $\mathcal S_c$ contains six $S$ gates and fifteen controlled-$Z$ gates,
and all formulas above reduce to the cubic-lattice expressions.  There is
one further simplification because the cubes admit a checkerboard
two-coloring.  Let $\mathcal C_\bullet$ denote one cube color and define
\begin{equation}
    R_\square
    :=
    \prod_{c\in\mathcal C_\bullet}C_c^{ab}.
\end{equation}
The cube--face incidence graph has maximum degree six, so six controlled-$Z$ layers suffice.
Every face belongs to exactly one cube in $\mathcal C_\bullet$, so
\begin{equation}
    R_\square K_aR_\square^{-1}
    =
    K_a\prod_fZ_f^b,
    \qquad
    R_\square U_B(\lambda)R_\square^{-1}
    =
    U_B(\lambda).
\end{equation}
Since $R_\square$ commutes with the diagonal $b$-qubit circuit,
Eq.~\eqref{eq:cubic-cellulation-decomposition} becomes
\begin{equation}
    R_\square U_{A,\mathcal C}R_\square^{-1}
    =
    K_aC_b^{\mathcal C}.
\end{equation}
Thus, the cubic lattice admits a controlled-$Z$-only representative after a
finite-depth, $U_B$-preserving change of basis.

\section{Continuum gauge-theory constructions}
\label{app:continuum-gauge-theories}

We now describe continuum gauge theories that realize Bockstein braiding statistics.
For concreteness, we focus on $d=3$ spatial dimensions and
$p=q=1$, so the relevant symmetries are two $\mathbb Z_N$ 1-form symmetries.

\subsection{Abelian gauge theory}

Consider a $U(1)$ gauge theory in $(3{+}1)$D coupled to a Higgs scalar
of charge $N$.  For $N=2$, this is an effective description of an s-wave
superconductor \cite{Hansson_2004}.  The theory has a $\mathbb Z_N$ electric
center 1-form symmetry and a $U(1)$ magnetic 1-form symmetry.  We focus on
a discrete subgroup $\mathbb Z_N\subset U(1)$ of the magnetic symmetry.  These
two $\mathbb Z_N$ 1-form symmetries have the mixed Bockstein anomaly with
$k=1$ \cite{Benini:2018reh,Hsin:2020nts,Hsin:2025ido}.

The excitations that realize the Bockstein braiding are fractional electric
and magnetic loop excitations.  They should not be viewed as genuine electric
charges or monopole particles; rather, they are loop operators associated with
fractional charge and fractional magnetic flux.
The anomaly implies that the theory cannot flow to a gapped phase preserving both $\mathbb Z_N$ 1-form symmetries.

There are two familiar regimes.
\begin{itemize}
    \item If the Higgs field does not condense, the theory flows to gapless
    Maxwell theory.  The electric center 1-form symmetry is enlarged to
    $U(1)\supset\mathbb Z_N$, and the $\mathbb Z_N\times\mathbb Z_N$ subgroup
    has the same mixed anomaly.  Both 1-form symmetries are spontaneously
    broken because electric charges and monopoles are deconfined.  The
    Bockstein braiding descends from the statistical interaction between
    electric and magnetic sources \cite{Gaiotto:2014kfa}.  For example, if a
    dyon of electric and magnetic charges $(q,m)$ is placed at the origin and
    another dyon $(q',m')$ is moved along a closed path subtending solid angle
    $\Omega$, the Berry phase is
    \[
        \exp\left(\frac{i}{2}(q'm-m'q)\Omega\right).
    \]

    \item If the Higgs field condenses, the theory flows to $\mathbb Z_N$ gauge
    theory.  Electric charges are deconfined, so the electric center 1-form
    symmetry is spontaneously broken.  Monopoles are confined by the Higgs
    condensate, so the magnetic 1-form symmetry remains unbroken.  However,
    the magnetic symmetry is fractionalized: the Abrikosov--Nielsen--Olesen
    string carries flux $2\pi/N$ in the resulting $\mathbb Z_N$ gauge theory
    \cite{Hsin:2019fhf,Hsin:2025jot,Hsin:2025ido}.
\end{itemize}

\subsection{Non-Abelian gauge theory}

Next, consider $SU(N^2)/\mathbb Z_N$ gauge theory in $(3{+}1)$D
coupled to $N_f$ adjoint Weyl fermions.  The dynamics is not known for all
$N$ and $N_f$, but the symmetry structure is universal.  The theory has a
$\mathbb Z_N$ center 1-form symmetry and a $\mathbb Z_N$ magnetic 1-form
symmetry, and these two symmetries have the mixed Bockstein anomaly with
$k=1$ \cite{Benini:2018reh,Hsin:2020nts,Hsin:2025ido}.  This anomaly follows
from the nonsplit extension
\begin{equation}
    1\longrightarrow \mathbb Z_N
    \longrightarrow SU(N^2)
    \longrightarrow SU(N^2)/\mathbb Z_N
    \longrightarrow 1 .
\end{equation}

The generator of the center 1-form symmetry is a fractional monopole surface
operator carrying holonomy in the center of the gauge group.  It braids with Wilson lines.
The generator of the magnetic 1-form symmetry is a topological symmetry
surface operator measuring the $\mathbb Z_N$-valued discrete 't Hooft flux,
namely the discrete part of the Goddard--Nuyts--Olive (GNO) magnetic flux.
Equivalently, this is the obstruction class in
$H^2(M,\mathbb Z_N)$ to lifting the $SU(N^2)/\mathbb Z_N$ gauge bundle
to an $SU(N^2)$ bundle.

The corresponding Bockstein braiding excitations are fractional monopole loops and
fractional charge loops.  Again, these are not genuine monopole or charge
particles; they are loop excitations.  The anomaly forbids the theory from flowing to a gapped phase preserving these symmetries.
Consider several special cases for which the dynamics is known:
\begin{itemize}
    \item For $N_f=0$, the theory is pure Yang-Mills.  The Wilson line is
    confined, and the low-energy theory is a $\mathbb Z_N$ 2-form gauge
    theory, equivalently the dual of $\mathbb Z_N$ gauge theory
    \cite{Gaiotto:2014kfa}.  The 't Hooft line becomes the Wilson line of the
    emergent $\mathbb Z_N$ gauge theory, so the magnetic 1-form symmetry is
    spontaneously broken.  The center 1-form symmetry is unbroken because
    Wilson lines are confined, but it fractionalizes on the $\mathbb Z_N$ flux
    of the low-energy gauge theory \cite{Hsin:2025ido}.

    \item For $N_f=1$, the related $\mathrm{SU}(N^2)$ theory is
    $\mathcal N=1$ super--Yang--Mills and has $N^2$
    vacua~\cite{Witten:1982df}. These vacua differ by stacking SPT phases
    for the $\mathbb Z_{N^2}$ center 1-form
    symmetry~\cite{Gaiotto:2014kfa},
    \begin{equation}
        2\pi i\frac{\nu}{2N^2}
        \int_{\mathcal M_4}\mathcal P(B),
        \qquad
        \nu=0,1,\ldots,N^2-1,
    \end{equation}
    where $\mathcal M_4$ is a four-dimensional spin spacetime,
    $B\in H^2(\mathcal M_4,\mathbb Z_{N^2})$ is the background field for
    the center 1-form symmetry, and
    $\mathcal P(B)=B\cup B+B\cup_1 dB$ denotes its Pontryagin square.
    The $SU(N^2)/\mathbb Z_N$ theory is obtained by gauging the
    $\mathbb Z_N\subset\mathbb Z_{N^2}$ subgroup of the center 1-form
    symmetry.  On this subgroup, one can write $B=Nb$ with $b$ a
    $\mathbb Z_N$ 2-form gauge field, and the SPT terms above become trivial
    as functions of $b$.  Thus, in each vacuum, gauging produces
    $\mathbb Z_N$ gauge theory, the low-energy description of the
    three-dimensional $\mathbb Z_N$ toric code.  The anomaly is matched as
    follows: the $\mathbb Z_N$ charge spontaneously breaks one
    $\mathbb Z_N$ 1-form symmetry, while the $\mathbb Z_N$ flux carries
    fractional charge under the other \cite{Hsin:2019fhf,Hsin:2025jot,Hsin:2025ido}.
\end{itemize}

\section{Equivalence with the previous particle-membrane invariant}
\label{app:eq41}

Ref.~\cite{kobayashi2024generalized} introduced a mutual statistical process
in $(3{+}1)$D between a $\ZZ_2$ particle and a membrane, as shown in
Fig.~\ref{fig:particle_membrane}. The process is defined as
\begin{equation}
    P
    =
    U_{t_2}^{-1}[U_e,U_{t_1}]U_{t_2}
    U_e^{-1}[U_e,U_{t_1}]U_e .
\end{equation}

We now show that this process is equivalent to the Bockstein braiding process
\begin{equation}
    W_2(U_{t_1},U_e)
    =
    \left(U_e^{-1}U_{t_1}^{-1}\right)^2
    \left(U_eU_{t_1}\right)^2 .
\end{equation}

First, consider the outer region, namely the complement of $t_1\cup t_2$. Let
$U_{\mathrm{out}}$ denote the corresponding volume operator. Since the support
of $U_{\mathrm{out}}$ is disjoint from that of $[U_e,U_{t_1}]$, we may insert
$U_{\mathrm{out}}^{-1}U_{\mathrm{out}}$ without changing the statistical
process:
\begin{equation}
    P
    \sim
    U_{t_2}^{-1}U_{\mathrm{out}}^{-1}
    [U_e,U_{t_1}]
    U_{\mathrm{out}}U_{t_2}
    U_e^{-1}[U_e,U_{t_1}]U_e .
\end{equation}
Here $\sim$ denotes equivalence of statistical processes.

The combined operator $U_{\mathrm{out}}U_{t_2}U_{t_1}$ preserves the
excitation configuration. Since the corresponding Abelian configuration
sector is one-dimensional, its restriction to that sector is scalar and
therefore commutes with $[U_e,U_{t_1}]$ when evaluating the statistical
process. Therefore,
\begin{align}
    P
    &\sim
    U_{t_2}^{-1}U_{\mathrm{out}}^{-1}
    [U_e,U_{t_1}]
    U_{\mathrm{out}}U_{t_2}
    U_e^{-1}[U_e,U_{t_1}]U_e \nonumber\\
    &=
    U_{t_1}
    \left(U_{\mathrm{out}}U_{t_2}U_{t_1}\right)^{-1}
    [U_e,U_{t_1}]
    \left(U_{\mathrm{out}}U_{t_2}U_{t_1}\right)
    U_{t_1}^{-1}
    U_e^{-1}[U_e,U_{t_1}]U_e \nonumber\\
    &\sim
    U_{t_1}[U_e,U_{t_1}]U_{t_1}^{-1}
    U_e^{-1}[U_e,U_{t_1}]U_e .
\end{align}
By the initial-state independence theorem, we may further conjugate the process by $U_e$, obtaining
\begin{equation}
    P
    \sim
    U_eU_{t_1}[U_e,U_{t_1}]U_{t_1}^{-1}
    U_e^{-1}[U_e,U_{t_1}] .
\end{equation}
Equivalently, the same process may be written as
\begin{equation}
\begin{aligned}
    P&\sim
    U_e^{-1}U_{t_1}^{-1}[U_e,U_{t_1}]U_{t_1}
    U_e[U_e,U_{t_1}] \\
    &=
    U_e^{-1}U_{t_1}^{-1}
    \left(U_e^{-1}U_{t_1}^{-1}U_eU_{t_1}\right)
    U_{t_1}U_e
    \left(U_e^{-1}U_{t_1}^{-1}U_eU_{t_1}\right) \\
    &=
    U_e^{-1}U_{t_1}^{-1}U_e^{-1}U_{t_1}^{-1}
    U_eU_{t_1}U_eU_{t_1} \\
    &=
    \left(U_e^{-1}U_{t_1}^{-1}\right)^2
    \left(U_eU_{t_1}\right)^2 .
\end{aligned}
\end{equation}
This is precisely the Bockstein braiding process $W_2(U_{t_1},U_e)$.

\begin{figure}
    \centering
    \includegraphics[width=0.3\linewidth]{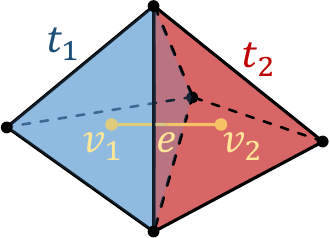}
    \caption{Geometry of particle-membrane statistics (adapted from
    Ref.~\cite{kobayashi2024generalized}).}
    \label{fig:particle_membrane}
\end{figure}

\bibliographystyle{utphys}
\bibliography{bibliography}

@article{Hsin:2025ria,
    author = "Hsin, Po-Shen and Kobayashi, Ryohei and Zhang, Carolyn",
    title = "{Anomalies of coset non-invertible symmetries}",
    eprint = "2503.00105",
    archivePrefix = "arXiv",
    primaryClass = "cond-mat.str-el",
    doi = "10.21468/SciPostPhys.20.1.006",
    journal = "SciPost Phys.",
    volume = "20",
    number = "1",
    pages = "006",
    year = "2026"
}

@article{Hsin:2019fhf,
    author = "Hsin, Po-Shen and Turzillo, Alex",
    title = "{Symmetry-enriched quantum spin liquids in (3 + 1)$d$}",
    eprint = "1904.11550",
    archivePrefix = "arXiv",
    primaryClass = "cond-mat.str-el",
    reportNumber = "CALT-TH-2019-014",
    doi = "10.1007/JHEP09(2020)022",
    journal = "JHEP",
    volume = "09",
    pages = "022",
    year = "2020"
}

@article{Hsin:2025jot,
    author = "Hsin, Po-Shen and Kobayashi, Ryohei",
    title = "{Generalized Hall conductivities in local commuting projector models: Generalized symmetries and protected surface modes}",
    eprint = "2505.20384",
    archivePrefix = "arXiv",
    primaryClass = "cond-mat.str-el",
    doi = "10.1103/94bv-37h1",
    journal = "Phys. Rev. B",
    volume = "113",
    number = "7",
    pages = "075106",
    year = "2026"
}

@article{Hsin:2025ido,
    author = "Hsin, Po-Shen",
    title = "{Generalized Symmetries Phase Transitions with Local Quantum Fields}",
    eprint = "2506.07688",
    archivePrefix = "arXiv",
    primaryClass = "cond-mat.str-el",
    month = "6",
    year = "2025"
}

@article{Else2014Classifying,
  title = {Classifying symmetry-protected topological phases through the anomalous action of the symmetry on the edge},
  author = {Else, Dominic V. and Nayak, Chetan},
  journal = {Phys. Rev. B},
  volume = {90},
  issue = {23},
  pages = {235137},
  numpages = {19},
  year = {2014},
  month = {Dec},
  publisher = {American Physical Society},
  doi = {10.1103/PhysRevB.90.235137},
  url = {https://link.aps.org/doi/10.1103/PhysRevB.90.235137}
}

@article{Else2020LSM,
  title = {Topological theory of {L}ieb-{S}chultz-{M}attis theorems in quantum spin systems},
  author = {Else, Dominic V. and Thorngren, Ryan},
  journal = {Phys. Rev. B},
  volume = {101},
  issue = {22},
  pages = {224437},
  numpages = {35},
  year = {2020},
  month = {Jun},
  publisher = {American Physical Society},
  doi = {10.1103/PhysRevB.101.224437},
  url = {https://link.aps.org/doi/10.1103/PhysRevB.101.224437}
}

@article{Kapustin:2013uxa,
    author = "Kapustin, Anton and Thorngren, Ryan",
    title = "{Higher Symmetry and Gapped Phases of Gauge Theories}",
    eprint = "1309.4721",
    archivePrefix = "arXiv",
    primaryClass = "hep-th",
    doi = "10.1007/978-3-319-59939-7_5",
    journal = "Prog. Math.",
    volume = "324",
    pages = "177--202",
    year = "2017"
}

@Article{CH21,
	title={{Exactly solvable lattice Hamiltonians and gravitational anomalies}},
	author={Yu-An Chen and Po-Shen Hsin},
	journal={SciPost Phys.},
	volume={14},
	pages={089},
	year={2023},
	publisher={SciPost},
	doi={10.21468/SciPostPhys.14.5.089},
	url={https://scipost.org/10.21468/SciPostPhys.14.5.089},
}

@article{Gaiotto:2014kfa,
    author = "Gaiotto, Davide and Kapustin, Anton and Seiberg, Nathan and Willett, Brian",
    title = "{Generalized Global Symmetries}",
    eprint = "1412.5148",
    archivePrefix = "arXiv",
    primaryClass = "hep-th",
    doi = "10.1007/JHEP02(2015)172",
    journal = "JHEP",
    volume = "02",
    pages = "172",
    year = "2015"
}

@article{Chen2023Highercup,
    author = {Chen, Yu-An and Tata, Sri},
    title = {Higher cup products on hypercubic lattices: Application to lattice models of topological phases},
    journal = {Journal of Mathematical Physics},
    volume = {64},
    number = {9},
    pages = {091902},
    year = {2023},
    month = {09},
    issn = {0022-2488},
    doi = {10.1063/5.0095189},
    url = {https://doi.org/10.1063/5.0095189},
}

@article{Hansson_2004,
   title={Superconductors are topologically ordered},
   volume={313},
   ISSN={0003-4916},
   url={http://dx.doi.org/10.1016/j.aop.2004.05.006},
   DOI={10.1016/j.aop.2004.05.006},
   number={2},
   journal={Annals of Physics},
   publisher={Elsevier BV},
   author={Hansson, T.H. and Oganesyan, Vadim and Sondhi, S.L.},
   year={2004},
   month=Oct, pages={497–538} }

@article{Cordova:2019bsd,
    author = "C{\'o}rdova, Clay and Ohmori, Kantaro",
    title = "{Anomaly Obstructions to Symmetry Preserving Gapped Phases}",
    eprint = "1910.04962",
    archivePrefix = "arXiv",
    primaryClass = "hep-th",
    month = "10",
    year = "2019"
}

@article{Witten:1982df,
    author = "Witten, Edward",
    title = "{Constraints on Supersymmetry Breaking}",
    reportNumber = "PRINT-82-0163 (PRINCETON)",
    doi = "10.1016/0550-3213(82)90071-2",
    journal = "Nucl. Phys. B",
    volume = "202",
    pages = "253",
    year = "1982"
}

@article{Hsin:2020nts,
    author = "Hsin, Po-Shen and Lam, Ho Tat",
    title = "{Discrete theta angles, symmetries and anomalies}",
    eprint = "2007.05915",
    archivePrefix = "arXiv",
    primaryClass = "hep-th",
    doi = "10.21468/SciPostPhys.10.2.032",
    journal = "SciPost Phys.",
    volume = "10",
    number = "2",
    pages = "032",
    year = "2021"
}

@book{hatcher2002algebraic,
  author    = {Hatcher, Allen},
  title     = {Algebraic Topology},
  publisher = {Cambridge University Press},
  address   = {Cambridge},
  year      = {2002},
  isbn      = {978-0-521-79540-1},
  url       = {https://pi.math.cornell.edu/~hatcher/AT/AT.pdf}
}

@article{Benini:2018reh,
      author         = "Benini, Francesco and C{\'o}rdova, Clay and Hsin, Po-Shen",
      title          = "{On 2-Group Global Symmetries and their Anomalies}",
      journal        = "JHEP",
      volume         = "03",
      year           = "2019",
      pages          = "118",
      doi            = "10.1007/JHEP03(2019)118",
      eprint         = "1803.09336",
      archivePrefix  = "arXiv",
      primaryClass   = "hep-th",
      reportNumber   = "SISSA 10/2018/FISI, SISSA-10-2018-FISI",
      SLACcitation   = "%%CITATION = ARXIV:1803.09336;%%"
}

@article{Kitaev2003Fault,
   title={Fault-tolerant quantum computation by anyons},
   volume={303},
   ISSN={0003-4916},
   url={http://dx.doi.org/10.1016/S0003-4916(02)00018-0},
   DOI={10.1016/s0003-4916(02)00018-0},
   number={1},
   journal={Annals of Physics},
   publisher={Elsevier BV},
   author={Kitaev, A.Yu.},
   year={2003},
   month={Jan},
   pages={2–30}
}

@article{Else2017Cheshire,
  title = {Cheshire charge in (3+1)-dimensional topological phases},
  author = {Else, Dominic V. and Nayak, Chetan},
  journal = {Phys. Rev. B},
  volume = {96},
  issue = {4},
  pages = {045136},
  numpages = {17},
  year = {2017},
  month = {Jul},
  publisher = {American Physical Society},
  doi = {10.1103/PhysRevB.96.045136},
  url = {https://link.aps.org/doi/10.1103/PhysRevB.96.045136}
}

@article{FHH21,
  title = {Gravitational anomaly of $(3+1)$-dimensional $\mathbb{Z}_2$ toric code with fermionic charges and fermionic loop self-statistics},
  author = {Fidkowski, Lukasz and Haah, Jeongwan and Hastings, Matthew B.},
  journal = {Phys. Rev. B},
  volume = {106},
  issue = {16},
  pages = {165135},
  numpages = {33},
  year = {2022},
  month = {Oct},
  publisher = {American Physical Society},
  doi = {10.1103/PhysRevB.106.165135},
  url = {https://link.aps.org/doi/10.1103/PhysRevB.106.165135}
}

@article{CarolynZhangSPTEntangler,
  title = {Topological invariants for symmetry-protected topological phase entanglers},
  author = {Zhang, Carolyn},
  journal = {Phys. Rev. B},
  volume = {107},
  issue = {23},
  pages = {235104},
  numpages = {23},
  year = {2023},
  month = {Jun},
  publisher = {American Physical Society},
  doi = {10.1103/PhysRevB.107.235104},
  url = {https://link.aps.org/doi/10.1103/PhysRevB.107.235104}
}

@article{freedman2003topological,
  author  = {Freedman, Michael H. and Kitaev, Alexei and Larsen, Michael J. and Wang, Zhenghan},
  title   = {Topological quantum computation},
  journal = {Bulletin of the American Mathematical Society},
  volume  = {40},
  number  = {1},
  pages   = {31--38},
  year    = {2003},
  doi     = {10.1090/S0273-0979-02-00964-3}
}

@article{Nayak2008NonAbelian,
  title = {Non-{A}belian anyons and topological quantum computation},
  author = {Nayak, Chetan and Simon, Steven H. and Stern, Ady and Freedman, Michael and Das Sarma, Sankar},
  journal = {Rev. Mod. Phys.},
  volume = {80},
  issue = {3},
  pages = {1083--1159},
  numpages = {0},
  year = {2008},
  month = {Sep},
  publisher = {American Physical Society},
  doi = {10.1103/RevModPhys.80.1083},
  url = {https://link.aps.org/doi/10.1103/RevModPhys.80.1083}
}

@article{Wang2014braiding,
  title = {Braiding Statistics of Loop Excitations in Three Dimensions},
  author = {Wang, Chenjie and Levin, Michael},
  journal = {Phys. Rev. Lett.},
  volume = {113},
  issue = {8},
  pages = {080403},
  numpages = {5},
  year = {2014},
  month = {Aug},
  publisher = {American Physical Society},
  doi = {10.1103/PhysRevLett.113.080403},
  url = {https://link.aps.org/doi/10.1103/PhysRevLett.113.080403}
}

@article{Jiang2014Generalized,
  title = {Generalized Modular Transformations in $(3+1)\mathrm{D}$ Topologically Ordered Phases and Triple Linking Invariant of Loop Braiding},
  author = {Jiang, Shenghan and Mesaros, Andrej and Ran, Ying},
  journal = {Phys. Rev. X},
  volume = {4},
  issue = {3},
  pages = {031048},
  numpages = {16},
  year = {2014},
  month = {Sep},
  publisher = {American Physical Society},
  doi = {10.1103/PhysRevX.4.031048},
  url = {https://link.aps.org/doi/10.1103/PhysRevX.4.031048}
}

@article{kobayashi2024generalized,
  title = {Generalized Statistics on Lattices},
  author = {Kobayashi, Ryohei and Li, Yuyang and Xue, Hanyu and Hsin, Po-Shen and Chen, Yu-An},
  journal = {Phys. Rev. X},
  volume = {16},
  issue = {1},
  pages = {011010},
  numpages = {51},
  year = {2026},
  month = {Jan},
  publisher = {American Physical Society},
  doi = {10.1103/6k88-w52n},
  url = {https://link.aps.org/doi/10.1103/6k88-w52n}
}

@article{feng2025anyonic,
  title = {Anyonic Membranes and {P}ontryagin Statistics},
  author = {Feng, Yitao and Xue, Hanyu and Li, Yuyang and Cheng, Meng and Kobayashi, Ryohei and Hsin, Po-Shen and Chen, Yu-An},
  journal = {Phys. Rev. Lett.},
  volume = {136},
  issue = {8},
  pages = {086601},
  numpages = {7},
  year = {2026},
  month = {Feb},
  publisher = {American Physical Society},
  doi = {10.1103/4jww-6b6t},
  url = {https://link.aps.org/doi/10.1103/4jww-6b6t}
}

@article{feng2025higherformanomalieslattices,
  title = {Higher-Form Anomalies on Lattices},
  author = {Feng, Yitao and Kobayashi, Ryohei and Chen, Yu-An and Ryu, Shinsei},
  journal = {Phys. Rev. Lett.},
  volume = {136},
  issue = {4},
  pages = {046504},
  numpages = {7},
  year = {2026},
  month = {Jan},
  publisher = {American Physical Society},
  doi = {10.1103/2jz1-m1lb},
  url = {https://link.aps.org/doi/10.1103/2jz1-m1lb}
}

@article{xue2025statistics,
  title = {Statistics of {A}belian topological excitations},
  author = {Xue, Hanyu},
  journal = {Phys. Rev. B},
  volume = {113},
  issue = {4},
  pages = {045143},
  numpages = {36},
  year = {2026},
  month = {Jan},
  publisher = {American Physical Society},
  doi = {10.1103/g3nc-fwqg},
  url = {https://link.aps.org/doi/10.1103/g3nc-fwqg}
}

@article{Kitaev2002Topologicalquantummemory,
    author = {Dennis, Eric and Kitaev, Alexei and Landahl, Andrew and Preskill, John},
    title = {Topological quantum memory},
    journal = {Journal of Mathematical Physics},
    volume = {43},
    number = {9},
    pages = {4452-4505},
    year = {2002},
    month = {09},
    issn = {0022-2488},
    doi = {10.1063/1.1499754},
    url = {https://doi.org/10.1063/1.1499754},
}

@article{kapustin2025higher,
    author = "Kapustin, Anton and Xu, Shixiong",
    title = "{Higher symmetries and anomalies in quantum lattice systems}",
    eprint = "2505.04719",
    archivePrefix = "arXiv",
    primaryClass = "math-ph",
    month = "5",
    year = "2025"
}

@article{kapustin2025higher2,
    author = "Kapustin, Anton and Spodyneiko, Lev",
    title = "{Higher symmetries, anomalies, and crossed squares in lattice gauge theory}",
    eprint = "2507.16966",
    archivePrefix = "arXiv",
    primaryClass = "hep-th",
    month = "7",
    year = "2025"
}

@article{Levin2015loopbraiding,
  title = {Loop braiding statistics in exactly soluble three-dimensional lattice models},
  author = {Lin, Chien-Hung and Levin, Michael},
  journal = {Phys. Rev. B},
  volume = {92},
  issue = {3},
  pages = {035115},
  numpages = {22},
  year = {2015},
  month = {Jul},
  publisher = {American Physical Society},
  doi = {10.1103/PhysRevB.92.035115},
  url = {https://link.aps.org/doi/10.1103/PhysRevB.92.035115}
}

@article{Bi2014anyonloopbraiding,
  title = {Anyon and loop braiding statistics in field theories with a topological $\ensuremath{\Theta}$ term},
  author = {Bi, Zhen and You, Yi-Zhuang and Xu, Cenke},
  journal = {Phys. Rev. B},
  volume = {90},
  issue = {8},
  pages = {081110},
  numpages = {4},
  year = {2014},
  month = {Aug},
  publisher = {American Physical Society},
  doi = {10.1103/PhysRevB.90.081110},
  url = {https://link.aps.org/doi/10.1103/PhysRevB.90.081110}
}

@article{AharonovBohm1959,
  title = {Significance of Electromagnetic Potentials in the Quantum Theory},
  author = {Aharonov, Y. and Bohm, D.},
  journal = {Phys. Rev.},
  volume = {115},
  issue = {3},
  pages = {485--491},
  numpages = {0},
  year = {1959},
  month = {Aug},
  publisher = {American Physical Society},
  doi = {10.1103/PhysRev.115.485},
  url = {https://link.aps.org/doi/10.1103/PhysRev.115.485}
}

@article{Wilczek1982,
  title = {Quantum Mechanics of Fractional-Spin Particles},
  author = {Wilczek, Frank},
  journal = {Phys. Rev. Lett.},
  volume = {49},
  issue = {14},
  pages = {957--959},
  numpages = {0},
  year = {1982},
  month = {Oct},
  publisher = {American Physical Society},
  doi = {10.1103/PhysRevLett.49.957},
  url = {https://link.aps.org/doi/10.1103/PhysRevLett.49.957}
}

@article{ArovasSchriefferWilczek1984,
  title = {Fractional Statistics and the Quantum {H}all Effect},
  author = {Arovas, Daniel and Schrieffer, J. R. and Wilczek, Frank},
  journal = {Phys. Rev. Lett.},
  volume = {53},
  issue = {7},
  pages = {722--723},
  numpages = {0},
  year = {1984},
  month = {Aug},
  publisher = {American Physical Society},
  doi = {10.1103/PhysRevLett.53.722},
  url = {https://link.aps.org/doi/10.1103/PhysRevLett.53.722}
}

@article{Bockstein1942,
  author  = {Bockstein, M.},
  title   = {Universal systems of {$\nabla$}-homology rings},
  journal = {C. R. (Doklady) Acad. Sci. URSS (N.S.)},
  volume  = {37},
  year    = {1942},
  pages   = {243--245},
  mrnumber = {0008701}
}

@article{feng2026paulistabilizerformalismtopological,
      title={Pauli stabilizer formalism for topological quantum field theories and generalized statistics}, 
      author={Yitao Feng and Hanyu Xue and Ryohei Kobayashi and Po-Shen Hsin and Yu-An Chen},
      year={2026},
      eprint={2601.00064},
      archivePrefix={arXiv},
      primaryClass={quant-ph},
      url={https://arxiv.org/abs/2601.00064}, 
}

@article{Kapustin2025Anomalous,
	author = {Kapustin, Anton and Sopenko, Nikita},
	da = {2025/09/01},
	doi = {10.1007/s00220-025-05422-2},
	id = {Kapustin2025},
	isbn = {1432-0916},
	journal = {Communications in Mathematical Physics},
	number = {10},
	pages = {238},
	title = {Anomalous Symmetries of Quantum Spin Chains and a Generalization of the {L}ieb--{S}chultz--{M}attis Theorem},
	ty = {JOUR},
	url = {https://doi.org/10.1007/s00220-025-05422-2},
	volume = {406},
	year = {2025}}

@article{Wilbur2026DisentanglingAnomaly,
  title = {Disentangling Anomaly-Free Symmetries of Quantum Spin Chains},
  author = {Seifnashri, Sahand and Shirley, Wilbur},
  journal = {Phys. Rev. Lett.},
  volume = {136},
  issue = {21},
  pages = {216603},
  numpages = {9},
  year = {2026},
  month = {May},
  publisher = {American Physical Society},
  doi = {10.1103/bscj-r5tg},
  url = {https://link.aps.org/doi/10.1103/bscj-r5tg}
}

@article{GoldsteinTurner1976,
 ISSN = {00029939, 10886826},
 URL = {http://www.jstor.org/stable/2041412},
 author = {Richard Z. Goldstein and Edward C. Turner},
 journal = {Proceedings of the American Mathematical Society},
 number = {1},
 pages = {339--342},
 publisher = {American Mathematical Society},
 title = {A Formula for {S}tiefel-{W}hitney Homology Classes},
 urldate = {2026-07-13},
 volume = {58},
 year = {1976}
}

@InProceedings{Preskill1999TopologicalQuantumComputation,
author="Walter Ogburn, R.
and Preskill, John",
editor="Williams, Colin P.",
title="Topological Quantum Computation",
booktitle="Quantum Computing and Quantum Communications",
year="1999",
publisher="Springer Berlin Heidelberg",
address="Berlin, Heidelberg",
pages="341--356",
isbn="978-3-540-49208-5"
}

@article{Kapustin2014Anomalous,
  title = {Anomalous Discrete Symmetries in Three Dimensions and Group Cohomology},
  author = {Kapustin, Anton and Thorngren, Ryan},
  journal = {Phys. Rev. Lett.},
  volume = {112},
  issue = {23},
  pages = {231602},
  numpages = {4},
  year = {2014},
  month = {Jun},
  publisher = {American Physical Society},
  doi = {10.1103/PhysRevLett.112.231602},
  url = {https://link.aps.org/doi/10.1103/PhysRevLett.112.231602}
}

@article{kapustin2014anomaliesdiscretesymmetriesvarious,
      title={Anomalies of discrete symmetries in various dimensions and group cohomology}, 
      author={Anton Kapustin and Ryan Thorngren},
      year={2014},
      eprint={1404.3230},
      archivePrefix={arXiv},
      primaryClass={hep-th},
      url={https://arxiv.org/abs/1404.3230}, 
}

@article{Shao2023HigherGauging,
	author = {Roumpedakis, Konstantinos and Seifnashri, Sahand and Shao, Shu-Heng},
	da = {2023/08/01},
	doi = {10.1007/s00220-023-04706-9},
	id = {Roumpedakis2023},
	isbn = {1432-0916},
	journal = {Communications in Mathematical Physics},
	number = {3},
	pages = {3043--3107},
	title = {Higher Gauging and Non-invertible Condensation Defects},
	ty = {JOUR},
	url = {https://doi.org/10.1007/s00220-023-04706-9},
	volume = {401},
	year = {2023}}

@article{Theo2022Classification,
	author = {Johnson-Freyd, Theo},
	da = {2022/07/01},
	doi = {10.1007/s00220-022-04380-3},
	id = {Johnson-Freyd2022},
	isbn = {1432-0916},
	journal = {Communications in Mathematical Physics},
	number = {2},
	pages = {989--1033},
	title = {On the Classification of Topological Orders},
	ty = {JOUR},
	url = {https://doi.org/10.1007/s00220-022-04380-3},
	volume = {393},
	year = {2022}}

\end{document}